\documentclass[12pt]{amsart}   

\usepackage{amssymb}
\usepackage{graphicx}
\usepackage{amsmath}
\usepackage{amscd}
\usepackage{latexsym}
\usepackage{tikz}
\usetikzlibrary{decorations.pathreplacing}
\usepackage{siunitx}
\usepackage{float}

\newtheorem{thm}{Theorem}[section]

\newtheorem{question}[thm]{Question}
\newtheorem{defn}[thm]{Definition}
\newtheorem{rmk}[thm]{Remark}
\newtheorem{cor}[thm]{Corollary}
\newtheorem{examp}[thm]{Example}
\newtheorem{prop}[thm]{Proposition}

\newcommand{\mc}{\mathcal}

\newcommand{\ts}{\textsc}

\DeclareMathOperator{\Aut}{Aut}
\DeclareMathOperator{\Perm}{Perm}

\DeclareMathOperator{\Fact}{Fact}
\DeclareMathOperator{\Eq*}{Eq*}
\DeclareMathOperator{\Eqq}{Eq}
\DeclareMathOperator{\LL}{L}
\DeclareMathOperator{\U}{U}

\begin{document}

\title[Wigner's Theorem for an Infinite Set]{Wigner's Theorem for an Infinite Set}

\author[Harding]{John Harding$^1$}
\address{Department of Mathematical Sciences\\
New Mexico State University\\
Las Cruces, NM 88003, USA}

\email{$^1$jharding@nmsu.edu\phantom{,}}

\thanks{}


\subjclass[2010]{}

\begin{abstract}
It is well known that the closed subspaces of a Hilbert space form an orthomodular lattice. Elements of this orthomodular lattice are the propositions of a quantum mechanical system represented by the Hilbert space, and by Gleason's theorem atoms of this orthomodular lattice correspond to pure states of the system. Wigner's theorem says that the automorphism group of this orthomodular lattice corresponds to the group of unitary and anti-unitary operators of the Hilbert space. This result is of basic importance in the use of group representations in quantum mechanics. 

The closed subspaces $A$ of a Hilbert space $\mc{H}$ correspond to direct product decompositions $\mc{H}\simeq A\times A^\perp$ of the Hilbert space, a result that lies at the heart of the superposition principle. In \cite{Harding 1} it was shown that the direct product decompositions of any set, group, vector space, and topological space form an orthomodular poset. This is the basis for a line of study in foundational quantum mechanics replacing Hilbert spaces with other types of structures. It is the purpose of this note to prove a version of Wigner's theorem: for an infinite set $X$, the automorphism group of the orthomodular poset $\Fact(X)$ of direct product decompositions of $X$ is isomorphic to the permutation group of $X$. 

The structure $\Fact(X)$ plays the role for direct product decompositions of a set that the lattice of equivalence relations plays for surjective images of a set. So determining its automorphism group is of interest independent of its application to quantum mechanics. Other properties of $\Fact(X)$ are determined in proving our version of Wigner's theorem, namely that $\Fact(X)$ is atomistic in a very strong way. 
\end{abstract}

\maketitle

\section{Introduction}

Standard treatments of quantum mechanics attach a separable complex Hilbert space $\mc{H}$ to a quantum mechanical system. The propositions of the system are given by closed subspaces of $\mc{H}$, and it is well known \cite{Mackey,Varadarajan} that the closed subspaces of $\mc{H}$ form an orthomodular lattice $\mc{C}(\mc{H})$. Using the spectral theorem and Gleason's theorem the observables and states of the system are completely described through this orthomodular lattice $\mc{C}(\mc{H})$. Wigner's theorem \cite{Ulhorn,Wigner} shows that automorphisms of the state space of the system, or equivalently automorphisms of the orthomodular lattice $\mc{C}(\mc{H})$ correspond to unitary and anti-unitary operators of $\mc{H}$. This result is of basic importance in motivating the use of group representations in quantum mechanics \cite{Stephanie,Varadarajan,Wigner}.

It has long been a question why a separable complex Hilbert space should be used to model a quantum system. In \cite{Harding 1} a different viewpoint was suggested where other types of structures can be used in place of a Hilbert space. To briefly explain, the closed subspaces $A$ of a Hilbert space $\mc{H}$ correspond to direct product decompositions $\mc{H}\simeq A\times A^\perp$. If $X$ is any set, group, vector space, topological space, and so forth, the collection $\Fact(X)$ of direct product decompositions of $X$ forms a type of structure known as an orthomodular poset \cite{Mackey,Ptak}. When applied to a Hilbert space $\mc{H}$, $\Fact(\mc{H})$ is the orthomodular lattice of closed subspaces $\mc{C}(\mc{H})$. In an ongoing series of papers \cite{Harding 1,Harding regular,Harding axioms,Harding 2,Harding5,Taewon,Taewon2} this path has been followed and rudimentary aspects of quantum mechanics constructed based on decompositions of structures other than Hilbert spaces. It is the hope that this path will provide explanation as to why Hilbert spaces are the structures necessary for quantum mechanics, or perhaps provide viable alternative structures.

The purpose of this paper is to prove the following analogue of Wigner's theorem. 
\vspace{2ex}

\noindent {\bf Main Theorem} {\em 
For an infinite set $X$, the group of automorphisms of the orthomodular poset $\Fact(X)$ is isomorphic to the permutation group of $X$. }
\vspace{2ex}

Apart from Wigner's theorem, there are several results directly, or indirectly related to the main theorem. There is the fundamental theorem of projective geometry that is usually provided as motivation for Wigner's theorem \cite{Varadarajan}. Roughly, this states that automorphisms of the lattice of subspaces of a finite dimensional vector space correspond to semilinear automorphisms of underlying vector space. A related direction involves a construction of an orthomodular poset $M^{(2)}$ from the complementary pairs of elements of any bounded modular lattice $M$. This construction arose in \cite{Harding 1,Mushtari}, and when applied to the modular lattice of subspaces of a vector space yields the orthomodular poset $\Fact(V)$ of direct product decompositions of $V$. In a series of papers \cite{Chevalier1,Chevalier2,Chevalier3,Ovchinnikov} the automorphisms of $M^{(2)}$ were described in terms of automorphisms of the lattice $M$, and in the finite dimensional vector space setting this can further be refined using the fundamental theorem of projective geometry. The result most closely related to the Main Theorem is that of \cite{Tim} where the same statement is established for a set $X$ with $27$ elements. This is a non-trivial result requiring a substantial amount of combinatorics. This result for 27-element sets will be used in a crucial way in establishing the proof of the main theorem. 

The structure $\Fact(X)$ itself is of independent interest since it plays the role for direct product decompositions of a set $X$ that the lattice $\Eqq(X)$ of equivalence relations of $X$ plays for surjections. Many properties of $\Fact(X)$ have previously been established \cite{Harding 2,Taewon}. Here, in addition to describing the automorphism group of $\Fact(X)$ we show that this orthomodular poset is atomistic in a very strong way. 

This paper is arranged in the following way. Section~2 contains background on $\Fact(X)$. Section~3 gives preliminary results about equivalence relations and decompositions. Section~4 is an overview of the proof of the main theorem --- that there is an isomorphism between the automorphism group of $\Fact(X)$ and the permutation group of $X$. The remainder of the paper provides the proofs of the results described in Section~4. Section~5 gives results about atoms and atomicity in $\Fact(X)$. Sections~6 and~7 provide needed background for Section~8 where it is shown that an automorphism $\alpha$ of $\Fact(X)$ induces an automorphism $\beta$ of the poset $\Eq*(X)$ of finite regular equivalence relations of $X$. Sections~9 and~10 use the automorphism $\beta$ to produce mappings of certain 4-element and 2-element subsets of $X$. These results are used in Section~11 to construct a permutation $\sigma$ of $X$ from which the automorphism $\alpha$ of $\Fact(X)$ is constructed. Section~12 has concluding remarks and open questions. 

\section{Preliminaries about $\Fact(X)$}

In this section we provide the basics of $\Fact(X)$. For further details, see \cite{Harding 1}. We begin by describing a familiar situation that is closely related. 

\begin{defn}
For a set $X$, a surjection with domain $X$ consists of a set $A$ and a surjective function $\varphi:X\to A$. 
\end{defn}

For any non-empty set $X$ there is a proper class of surjections with domain $X$ since there is a proper class of even 1-element sets $A$ that can serve as the image of a surjection. However, it is frequently the case that the specific elements of the image are not of interest, and it is only the behavior of the surjection $\varphi$ that is important. The familiar approach is to define surjections $\varphi:X\to A$ and $\psi:X\to B$ to be equivalent if there is a bijection $\alpha:B\to A$ with $\varphi=\alpha\circ\psi$. 

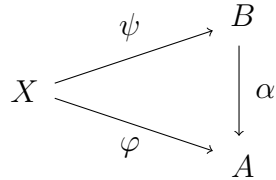
\vspace{2ex}

\begin{center}
\begin{tikzpicture}
\node (X) at (0,0) {$X$};
\node (A) at (2.9,1) {$B$};
\node (B) at (2.9,-1) {$A$};
\draw[->] (0.4,0.1) -- (2.5,0.8);
\draw[->] (0.4,-0.1) -- (2.5,-0.8);
\draw[->] (2.85,.6) -- (2.85,-.6);
\node at (3.2,0) {$\alpha$};
\node at (1.4,.85) {$\psi$};
\node at (1.4,-.75) {$\varphi$};
\end{tikzpicture}
\end{center}
\vspace{1ex}

Modulo the usual set-theoretic considerations, this gives an equivalence relation on the surjections. We let $[\varphi]$ be the equivalence class of the surjection $\varphi$. There remains a further aspect of surjections that is vital to their study, the matter of placing structure on the set of equivalence classes of surjections with domain $X$. For surjections $\varphi:X\to A$ and $\psi:X\to B$ define $[\varphi]\leq[\psi]$ if there is a surjection $\alpha:B\to A$ with $\varphi=\alpha\circ\psi$. It is well known that this gives a partial ordering on the set of equivalence class of surjections. In fact, more is true as is summarized in the following result. 

\begin{thm}
For a surjection $\varphi:X\to A$, its kernel $\ker\varphi =\{(x,y):\varphi (x) = \varphi (y)\}$ is the unique equivalence relation $\theta$ on $X$ whose canonical quotient map $\kappa_\theta:X\to X/\theta$ belongs to the equivalence class $[\varphi]$. This provides an isomorphism between the lattice $\Eqq(X)$ of equivalence relations on $X$ and the partially ordered set of equivalence classes of surjections with domain $X$. 
\end{thm}

\begin{examp}{\em 
Consider a 3-element set $X=\{a,b,c\}$. There are five elements of $\Eqq(X)$, the smallest equivalence relation $\Delta$, the largest $\nabla$, and the three relations that each have one block of two elements, $\theta_1$, $\theta_2$, and $\theta_3$. The lattice $\Eqq(X)$ is shown below and it is isomorphic to the lattice of equivalence classes of surjections. 
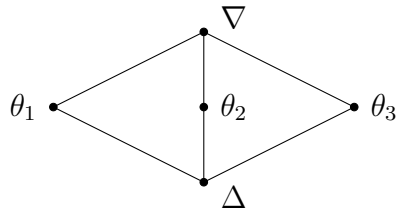
\vspace{2ex}

\begin{center}
\begin{tikzpicture}
\draw[fill] (0,0) circle [radius=0.05]; 
\draw[fill] (-2,1) circle [radius=0.05];
\draw[fill] (0,1) circle [radius=0.05];
\draw[fill] (2,1) circle [radius=0.05];
\draw[fill] (0,2) circle [radius=0.05];
\draw[thin] (0,0) -- (-2,1) -- (0,2) -- (0,1) -- (0,0) -- (2,1) -- (0,2);
\node at (.4,-.2) {$\Delta$};
\node at (.4,2.2) {$\nabla$};
\node at (-2.4,1) {$\theta_1$};
\node at (.4,1) {$\theta_2$};
\node at (2.4,1) {$\theta_3$};
\end{tikzpicture}
\end{center}
\vspace{2ex}

\noindent Note that there three equivalence classes of surjections where the codomain $A$ has two elements, one where the surjection $\varphi$ maps $a$ and $b$ to the same place, one where $\varphi$ maps $a$ and $c$ to the same place, and one where $\varphi$ maps $b$ and $c$ to the same place. 
}
\end{examp}

We now come to the matter of primary interest, direct product decompositions of a set $X$. 

\begin{defn}
For a natural number $n$, an $n$-ary direct product decomposition of a set $X$ consists of sets $A_1,\ldots,A_n$ and a bijection $\varphi:X\to A_1\times \cdots\times A_n$. 
\end{defn}

When $n=2$ a direct product decomposition is called a binary decomposition, and when $n=3$ a ternary decomposition. For any given $n\geq 1$ there is a proper class of $n$-ary direct product decompositions of $X$. We will define an equivalence relation on these much as we did with surjections. To do so, we note that for maps $\alpha_i:A_i\to B_i$ for $i=1,\ldots,n$ there is an obvious map $\alpha_1\times\cdots\times\alpha_n$ from $A_1\times\cdots\times A_n$ to $B_1\times\cdots\times B_n$. 

\begin{defn}
For a given $n$, two $n$-ary direct product decompositions $\varphi:X\to A_1\times\cdots\times A_n$ and $\psi:X\to B_1\times\cdots\times B_n$ are equivalent if there are bijections $\alpha_i:A_i\to B_i$ for each $i\leq n$ making the following diagram commute. 
\end{defn}
\vspace{1ex}

\begin{center}
\begin{tikzpicture}
\node (X) at (0,0) {$X$};
\node (A) at (4,1) {$A_1\times\,\cdots\,\times A_n$};
\node (B) at (4,-1) {$B_1\times\,\cdots\,\times\, B_n$};
\draw[->] (0.4,0.1) -- (2.5,0.8);
\draw[->] (0.4,-0.1) -- (2.5,-0.8);
\draw[->] (2.85,.6) -- (2.85,-.6);
\draw[->] (5.05,.6) -- (5.05,-.6);
\node at (3.2,0) {$\alpha_1$};
\node at (5.45,0) {$\alpha_n$};
\node at (1.2,.75) {$\varphi$};
\node at (1.2,-.75) {$\psi$};
\end{tikzpicture}
\end{center}
\vspace{1ex}

Just as surjections can be conveniently treated through equivalence relations, so too can decompositions. The key notion is that of a factor pair. To define this, the composition of relations is given by $\theta\circ\phi = \{(x,y):\mbox{there is $y$ with $(x,y)\in\theta$ and $(y,z)\in\phi$}\}$, and we say that the relations $\theta$ and $\phi$ permute if $\theta\circ\phi=\phi\circ\theta$. 

\begin{defn} 
A factor pair of a set $X$ is an ordered pair $(\theta,\theta')$ of equivalence relations on $X$ where 
\vspace{1ex}
\begin{enumerate}
\item $\theta\cap\theta'=\Delta$
\item $\,\theta\circ\theta'=\nabla$
\end{enumerate}
\vspace{1ex}

\noindent Let $\Fact(X)$ be the set of all factor pairs on $X$. 
\end{defn}

For any factor pair $(\theta,\theta')$, it is known \cite{Burris,McKenzieMcNultyTaylor} that $\theta$ and $\theta'$ permute, and so $\theta'\circ\theta=\nabla$. For permuting equivalence relations, their join in the lattice of equivalence relations is given by their relational product, so factor pairs are complementary elements of the lattice of equivalence relations that also permute. The following results \cite{Burris,McKenzieMcNultyTaylor} is the primary reason for interest in factor pairs. 

\begin{prop}
An ordered pair $(\theta_1,\theta_2)$ of equivalence relations on a set $X$ is a factor pair if and only if the canonical map $\kappa:X\to X/\theta_1\times X/\theta_2$ is a bijection.
\end{prop}

The following result \cite{Harding 1} shows that factor pairs of $X$ and equivalence classes of binary direct product decompositions of $X$ are in bijective correspondence. 

\begin{prop}
For a binary direct product decomposition $\varphi:X\to A_1\times A_2$ there is a unique factor pair $(\theta_1,\theta_2)$ whose canonical bijection $\kappa:X\to X/\theta_1\times X/\theta_2$ belongs to $[\varphi]$. 
\end{prop}

A key aspect of matters remains, putting structure on the collection of binary direct product decompositions of $X$, or equivalently, putting structure on $\Fact(X)$. We first describe this structure in terms of $\Fact(X)$ since it is easier to state, if somewhat non-transparent. Then we describe the underlying meaning in terms of the natural setting of direct product decompositions. 

\begin{defn}
On the set $\Fact(X)$ of factor pairs of a set $X$ define constants $0,1$, a unary operation $\perp$, and a binary relation $\leq$ as follows. 
\vspace{1ex}
\begin{enumerate}
\item $0=(\nabla,\Delta)$
\item $1=(\Delta,\nabla)$
\item $(\theta,\theta')^\perp=(\theta',\theta)$
\item $(\theta,\theta')\leq(\phi,\phi')$ iff $\phi\subseteq\theta$, $\theta'\subseteq\phi'$, and $\phi$ and $\theta'$ permute
\end{enumerate}
\end{defn}

The definitions are given in a way that is easy to verify. When working with the relation $\leq$ the following \cite{Harding 1} is useful. 

\begin{prop}
For factor pairs $(\theta,\theta')$ and $(\phi,\phi')$ of $X$, these are equivalent.
\vspace{1ex}
\begin{enumerate}
\item $(\theta,\theta')\leq(\phi,\phi')$
\item $\kappa:X\,\to\, X/\theta\,\times\, X/(\theta\circ\phi')\,\times\, X/\phi'$ is a ternary direct product decomposition
\item $\{\Delta,\,\theta,\,\theta',\,\phi,\,\phi',\,\theta\cap\phi',\,\theta'\circ\phi,\,\nabla\}$ is a Boolean sublattice of $\Eqq(X)$ consisting of pairwise permuting equivalence relations. 
\end{enumerate}
\vspace{1ex}

\noindent In the third condition, it is not assumed that the elements are all distinct. 
\label{duck}
\end{prop}
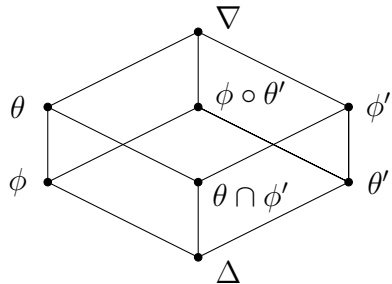
\vspace{2ex}

\begin{center}
\begin{tikzpicture}
\draw[fill] (0,0) circle [radius=0.05]; 
\draw[fill] (-2,1) circle [radius=0.05];
\draw[fill] (0,1) circle [radius=0.05];
\draw[fill] (2,1) circle [radius=0.05];
\draw[fill] (-2,2) circle [radius=0.05];
\draw[fill] (0,2) circle [radius=0.05];
\draw[fill] (2,2) circle [radius=0.05];
\draw[fill] (0,3) circle [radius=0.05];
\draw[thin] (0,0) -- (-2,1) -- (0,2) -- (2,1) -- (0,2);
\draw[thin] (-2,1) -- (-2,2) -- (0,3) -- (2,2) -- (2,1) -- (2,2) -- (0,1) -- (-2,2);
\draw[thin] (0,2) -- (0,3);
\draw[thin] (0,1) -- (0,0) -- (2,1);
\node at (.4,-.2) {$\Delta$};
\node at (.4,3.2) {$\nabla$};
\node at (-2.4,2) {$\theta$};
\node at (-2.4,1) {$\phi$};
\node at (.7,.8) {$\theta\cap\phi'$};
\node at (.7,2.15) {$\phi\circ\theta'$};
\node at (2.4,1) {$\theta'$};
\node at (2.4,2) {$\phi'$};
\end{tikzpicture}
\end{center}
\vspace{2ex}

The nature of the operations on $\Fact(X)$ is best understood when they are translated to equivalence classes of binary direct product decompositions of $X$. If $\{*\}$ is a 1-element set, the constant $0$ is the equivalence class of the decomposition $\zeta:X\to \{*\}\times X$ and $1$ is the equivalence class of $\iota:X\to X\times\{*\}$. For any decomposition $\varphi:X\to A_1\times A_2$, there is a natural decomposition $\varphi':X\to A_2\times A_1$, and $\perp$ takes the equivalence class $[\varphi]$ to $[\varphi']$.

For the meaning of the relation $\leq$ let $\gamma: X\to A_1\times A_2\times A_3$ be a ternary decomposition. From this we can build two binary decompositions 

\[\gamma_{23}:X\to A_1\times (A_2\times A_3) \quad\mbox{ and }\quad \gamma_{12}:X\to (A_1\times A_2)\times A_3\]
\vspace{-1ex}

\noindent  Intuitively we think of the first of these as being smaller than the second because its first factor is smaller. This is made precise by defining $\leq$ to be the relation consisting of all instances of $[\gamma_{12}]\leq[\gamma_{23}]$. Thus one binary decomposition is $\leq$ related to another if they can be built in this way from a common ternary refinement. To relate this to the definition in terms of factor pairs, we refer to Proposition~\ref{duck} and let $\gamma$ be the canonical map $\gamma:X\to X/\theta\,\times\, X/(\phi\circ\theta')\,\times\, X/\phi'$. We will make use here only of the definitions in terms of factor pairs, and refer the reader to \cite{Harding 1,Harding regular,Harding axioms} for further discussion of $n$-ary decompositions.

\begin{defn}
An orthomodular poset $(P,\leq,\perp,0,1)$ is a poset $P$ with least element $0$, largest element $1$, and a unary operation $\perp$ that satisfies the following where $x\perp y$ means $x\leq y^\perp$
\vspace{1ex}
\begin{enumerate}
\item $x\leq y\Rightarrow y^\perp\leq x^\perp$
\item $x^{\perp\perp}=x$
\item $x,x^\perp$ have $0$ as their only lower bound and $1$ as their only upper bound
\item if $x \perp y$ then $x,y$ have a least upper bound written $x\oplus y$
\item if $x\perp y$ then $x\oplus(x\oplus y)^\perp = y^\perp$
\end{enumerate}
\vspace{1ex}

\noindent An orthomodular poset that is a lattice is an orthomodular lattice. 
\end{defn}

Now the foundational result of \cite{Harding 1}. 

\begin{thm}
For $X$ a non-empty set, $(\Fact(X),\leq,\perp,0,1)$ is an orthomodular poset. 
\end{thm}

We note that this theorem has applicability to setting much broader than sets \cite{Harding 1, Harding 2, Harding5}, but our concern here is only with $\Fact(X)$ as applied to sets. A simple example is useful. 

\begin{examp}{\em 
Let $X=\{a,b,c,d\}$. Consider the equivalence relations $\theta_1,\theta_2,\theta_3$ that each have two blocks of two elements each. These are $\theta_1$ with blocks $\{a,b\}$ and $\{c,d\}$, $\theta_2$ with blocks $\{a,c\}$ and $\{b,d\}$, and $\theta_3$ with blocks $\{a,d\}$ and $\{b,c\}$. There are eight factor pairs of $X$. The trivial pairs $0=(\nabla,\Delta)$ and $1=(\Delta,\nabla)$, and the six factor pairs $F_{i,j}=(\theta_i,\theta_j)$ for $i\neq j$. A diagram of the orthomodular poset $\Fact(X)$ is shown below. 
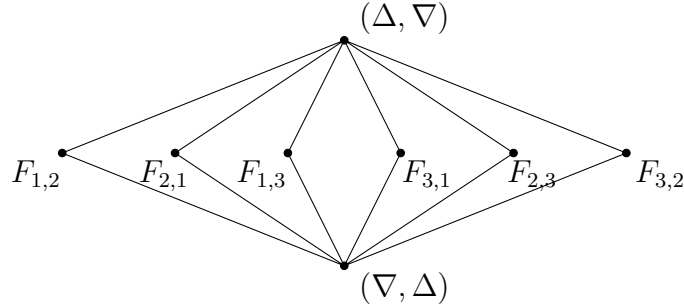
\vspace{2ex}

\begin{center}
\begin{tikzpicture}
\draw[fill] (0,0) circle [radius=0.05]; 
\draw[fill] (-3.75,1.5) circle [radius=0.05];
\draw[fill] (-2.25,1.5) circle [radius=0.05];
\draw[fill] (-.75,1.5) circle [radius=0.05];
\draw[fill] (.75,1.5) circle [radius=0.05];
\draw[fill] (2.25,1.5) circle [radius=0.05];
\draw[fill] (3.75,1.5) circle [radius=0.05];
\draw[fill] (0,3) circle [radius=0.05];
\draw[thin] (0,0) -- (-3.75,1.5) -- (0,3) -- (-2.25,1.5) -- (0,0) -- (-.75,1.5) -- (0,3) -- (.75,1.5) -- (0,0) -- (2.25,1.5) -- (0,3) -- (3.75,1.5) -- (0,0);
\node at (.8,-.3) {$(\nabla,\Delta)$};
\node at (.8,3.3) {$(\Delta,\nabla)$};
\node at (-4.1,1.2) {$F_{1,2}$};
\node at (-2.4,1.2) {$F_{2,1}$};
\node at (-1.1,1.2) {$F_{1,3}$};
\node at (1.1,1.2) {$F_{3,1}$};
\node at (2.5,1.2) {$F_{2,3}$};
\node at (4.2,1.2) {$F_{3,2}$};
\end{tikzpicture}
\end{center}
\vspace{2ex}

\noindent There are also eight equivalence classes of binary direct product decompositions of $X$. One as the product of a 1-element set times a 4-element set, one as the product of a 4-element set times a 1-element set, and six as the product of two 2-element sets. The reason there are six is that the way in which the decompositions collapse elements with respect to the first or second projections will differ. This is analogous to there being three ways to obtain a 2-element surjective image of a 3-element set. 
}
\end{examp}

Many properties of $\Fact(X)$ for a set $X$ are discussed in \cite{Tim, Harding 1, Harding regular}. In particular, \cite{Tim} has much information about $\Fact(X)$ when $X$ is a finite set. If the cardinality $|X|=p$ is prime, then $\Fact(X)$ has exactly two elements since there are no non-trivial direct product decompositions. If $|X|=pq$ is the product of two primes, then $\Fact(X)$ looks much like the diagram above. It has a bottom, a top, and a single level of incomparable elements in the middle paired into complements. One can view this as a collection of 4-element Boolean algebras glued together at the top and bottom. When $|X|$ is the product of three primes, $\Fact(X)$ is a collection of 8-element Boolean algebras glued together, but perhaps not just at their top and bottom. This is similar to the lattice of subspaces of $\mathbb{R}^3$. The sizes of these structures are also given in \cite{Tim}, and they grow quickly. For example, when $|X|=27$, the orthomodular poset $\Fact(X)$ has 10,002,268,381,116,211,200,002 elements, but still has just four levels, i.e. is of height three.

\section{Preliminary results about equivalence relations and permutations}

In this section we collect several results about relationships between equivalence relations and factor pairs that will be used throughout the paper. We begin with the following basic definition that applies to equivalence relations on finite, and on infinite sets. 

\begin{defn}
An equivalence relation $\theta$ on a set $X$ is regular if each of its blocks has the same cardinality. 
\end{defn}

The importance of regular equivalence relations in our study is given by the following. Here, and in future, when we refer to the number of elements in a set, we mean the cardinal number, and allow the possibility that it is infinite. 

\pagebreak[3]

\begin{prop}
Let $(\theta,\theta')$ be a factor pair of $X$. Then 
\vspace{1ex}
\begin{enumerate}
\item both $\theta$ and $\theta'$ are regular equivalence relations
\item the number of blocks of $\theta$ is the common cardinality of the blocks of $\theta'$
\item the number of blocks of $\theta'$ is the common cardinality of the blocks of $\theta$
\end{enumerate}
\label{zoop}
\end{prop}

\begin{proof}
Since the natural map $\varphi:X\to X/\theta\times X/\theta'$ is a bijection, each block of $\theta$ has cardinality $|X/\theta'|$, and this is equal to the number of blocks of $\theta'$, and each block of $\theta'$ has cardinality $|X/\theta|$, and this is equal to the number of blocks of $\theta$. 
\end{proof}

\begin{prop}
For an equivalence relation $\theta$ on a set $X$, the following are equivalent. 
\vspace{1ex}
\begin{enumerate}
\item $\theta$ is regular
\item there is an equivalence relation $\theta'$ with $(\theta,\theta')$ a factor pair of $X$
\end{enumerate}
\label{stinky}
\end{prop}

\begin{proof}
We show (1) $\Rightarrow$ (2), the other direction is given by the previous result. Let $E_i$ $(i\in \kappa)$~be an indexing of the blocks of $\theta$, and suppose that each block has $\lambda$ elements. For each $i\in \kappa$ choose an enumeration $b_{i,j}$ $(j\in\lambda)$ of the elements of $B_i$. Then $b_{i,j}$ where $i\in \kappa$ and $j\in\lambda$ enumerates the elements of $X$, and $b_{i,j}$ is $\theta$-related to $b_{p,q}$ iff $i=p$. Define $\theta'$ to be the relation where $b_{i,j}$ is $\theta'$-related to $b_{p,q}$ iff $j=q$. It is easily seen that $\theta'$ is an equivalence relation, that $\theta\cap\theta'=\Delta$, and for any $i,j,p,q$ we have $(b_{i,j},b_{i,q})\in \theta$ and $(b_{i,q},b_{j,q})\in\theta'$, so $\theta\circ\theta'=\nabla$. 
\end{proof} 

\begin{defn}
For a natural number $n$, an equivalence relation $\theta$ on $X$ is an $n$-relation if each equivalence class of $\theta$ has $n$ elements. 
\end{defn}

\begin{defn} 
An equivalence relation $\theta$ is finite regular if it is an $n$-relation for some $n$. Let $\Eq*(X)$ be the poset of finite regular equivalence relations on $X$, partially ordered by set inclusion. 
\label{fool}
\end{defn}

A permutation of a set $X$ is a bijection $\sigma:X\to X$. The collection $\Perm(X)$ of all permutations of $X$ is a group under function composition. This group is usually called the permutation group of $X$, or the symmetric group of $X$. 

\begin{defn}
For $\sigma$ a permutation of $X$ and $\theta$ an equivalence relation on $X$ define 
\[\sigma(\theta)=\{(\sigma(x),\sigma(y)):(x,y)\in\theta\}\]
\end{defn}

The following is easily verified. 

\begin{prop}
Let $\sigma$ be a permutation of $X$ and $\theta,\phi$ be equivalence relations on $X$. 
\vspace{1ex}
\begin{enumerate}
\item $\sigma(\theta)$ is an equivalence relation
\item $\theta$ is an $n$-relation iff $\sigma(\theta)$ is an $n$-relation
\item $\theta\subseteq\phi$ iff $\sigma(\theta)\subseteq\sigma(\phi)$
\item $\sigma(\theta\cap\phi) = \sigma(\theta)\cap\sigma(\phi)$
\item $\sigma(\theta\circ\phi) = \sigma(\theta)\circ\sigma(\phi)$
\end{enumerate}
\label{lucky}
\end{prop}

We begin to relate permutations of $X$ to automorphisms of the posets $\Eqq(X)$ and $\Eq*(X)$, and to automorphisms of the orthomodular poset $\Fact(X)$. The easy direction is moving from permutations of $X$ to automorphisms of these more complicated structures. 

\begin{defn}
An automorphism of a poset $P$ is an order preserving bijection $\beta$ from $P$ to itself whose inverse is order preserving. Let $\Aut(P)$ be the group of automorphisms of $P$. 
\end{defn} 

\begin{prop} 
Let $X$ be an infinite set. Then there are one-one group homomorphisms $\Psi:\Perm(X)\to\Aut(\Eqq(X))$ and $\Psi^*:\Perm(X)\to\Aut(\Eq*(X))$ given by  
\vspace{1ex}
\begin{align*}
\Psi(\sigma)(\theta) \,=\, \sigma(\theta) \,\,&\mbox{ for all $\,\,\theta\in\Eqq(X)$}\\
\Psi^*(\sigma)(\theta)\, =\, \sigma(\theta) \,\,&\mbox{ for all $\,\,\theta\in\Eq*(X)$}
\end{align*}
\vspace{-2ex}
\label{sad}
\end{prop}

\begin{proof}
Proposition~\ref{lucky} shows that $\Psi(\sigma)$ is an order preserving map from $\Eqq(X)$ to itself. The permutation $\sigma$ has an inverse $\sigma^{-1}$ and it is easily seen that $\Psi(\sigma^{-1})$ is an inverse of $\Psi(\sigma)$. Therefore $\Psi(\sigma)$ is an automorphism of $\Eqq(X)$. For permutations $\sigma$ and $\delta$ of $X$ we have 

\[\Psi(\sigma\circ\delta)(\theta)=\sigma(\delta(\theta))=(\Psi(\sigma)\circ\Psi(\delta))(\theta)\]
\vspace{-1ex}

\noindent  So $\Psi$ is a group homomorphism. An identical argument shows that $\Psi^*$ is a group homomorphism once it is noted that Proposition~\ref{lucky} shows that $\theta$ being a finite regular equivalence relation implies that $\sigma(\theta)$ is such also. 

We will show that $\Psi^*$ is one-one. Our argument will also apply to show that $\Psi$ is one-one. Suppose $\sigma$ is a permutation of $X$ that is not the identity. Then there is $x\in X$ with $\sigma(x)\neq x$. Since $X$ is infinite there is $y$ with $y,\sigma(y)\not\in\{x,\sigma(x)\}$. Here $y$ might or might not be equal to $\sigma(y)$, it does not matter. Then since $X$ is infinite, there is a 3-relation $\theta$ with one block of $\theta$ containing $x,y,\sigma(y)$ and not containing $\sigma(x)$. It cannot be the case that $\sigma(\theta)=\theta$ because $(x,y)\in\theta$ implies that $(\sigma(x),\sigma(y))\in\sigma(\theta)$, but by construction $(\sigma(x),\sigma(y))\not\in\theta$. So $\Psi^*(\sigma)$ is not the identity automorphism of $\Eq*(X)$. 
\end{proof}

\begin{rmk}{\em 
It is easily seen that if $|X| > 2$ the map $\Psi:\Perm(X)\to\Aut(\Eqq(X))$ is one-one, and it is a nice exercise to show that it is onto as well. When $|X|$ is prime or $|X|=4$ the map $\Psi^*:\Perm(X)\to\Aut(\Eq*(X))$ is not one-one. In all other cases, it is one-one as can be seen by modifying the proof given above. It remains open whether $\Psi^*$ is a group isomorphism when $X$ is infinite. However, a related problem sufficient for our purposes is answered positively in Section~9.}
\end{rmk}

\begin{defn}
An automorphism of an orthomoular poset $Q$ is an automorphism of the poset underlying $Q$ that satisfies $\alpha(x^\perp)=\alpha(x)^\perp$. Let $\Aut(Q)$ be the group of automorphisms of the orthomodular poset $Q$. 
\end{defn}

\begin{prop} 
Let $X$ be an infinite set. Then there is a one-one group homomorphism $\Gamma:\Perm(X)\to\Aut(\Fact(X))$  given by  
\[\Gamma(\sigma)(\theta,\theta') \,=\, (\sigma(\theta),\sigma(\theta'))\]  
\vspace{-2ex}
\label{dopey}
\end{prop}

\begin{proof}
It follows from Proposition~\ref{lucky} that if $(\theta,\theta')$ is a factor pair, so is $(\sigma(\theta),\sigma(\theta'))$. It also follows from Proposition~\ref{lucky} and the definition of $\leq$ in $\Fact(X)$, that $\Gamma(\sigma)$ is order preserving. Note that $\Gamma(\sigma)$ preserves the orthocomplementation $\perp$ in $\Fact(X)$ since 

\[ \Gamma(\sigma)((\theta,\theta')^\perp)=\Gamma(\sigma)(\theta',\theta)=(\sigma(\theta'),\sigma(\theta))=(\Gamma(\sigma)(\theta,\theta'))^\perp\]
\vspace{-1ex}

\noindent  It is easily seen that $\Gamma(\sigma^{-1})$ is the inverse of $\Gamma(\sigma)$, so $\Gamma(\sigma)$ is an automorphism of $\Fact(X)$, and  clearly $\Gamma$ is a group homomorphism. To see that $\Gamma$ is one-one, if $\sigma$ is a permutation of $X$ that is not the identity, Proposition~\ref{sad} shows there is a regular relation $\theta$ with $\sigma(\theta)\neq\theta$. Proposition~\ref{stinky} gives $\theta'$ with $(\theta,\theta')$ a factor pair, and then $\Gamma(\sigma)(\theta,\theta')\neq(\theta,\theta')$, so $\Gamma(\sigma)$ is not the identity automorphism. 
\end{proof}

\section{Outline of the proof of the main theorem}

At this point we have enough of the basics to outline of the proof of the main theorem which we state in a more precise form below. Much of the remainder of the paper is devoted to establishing the steps in this outline. 

\begin{thm} 
For an infinite set $X$ there is a group isomorphism $\Gamma:\Perm(X)\to\Aut(\Fact(X))$ given by 
\[ \Gamma(\sigma)(\theta,\theta')=(\sigma(\theta),\sigma(\theta'))\]
\vspace{-2ex}
\label{main}
\end{thm}

In Proposition~\ref{dopey} we have shown that $\Gamma$ is a one-one group homomorphism, so it only remains to show that $\Gamma$ is onto. This amounts to establishing the following. 

\begin{thm}
If $X$ is an infinite set and $\alpha$ is an automorphism of $\Fact(X)$, then there is a permutation $\sigma$ of $X$ with $\Gamma(\sigma)=\alpha$. 
\label{main'}
\end{thm}

We outline the proof of Theorem~\ref{main'}. We begin with two results established previously. The first was established in \cite{Taewon}. Before stating it, we recall the well-known result \cite{Kalmbach,Ptak} that an interval $[0,a]$ of an orthomodular poset $P$ naturally forms an orthomodular poset.  

\begin{thm}
If $(\theta,\theta')$ is a factor pair of a set $X$, then the interval $[(\nabla,\Delta),(\theta,\theta')]$ of $\Fact(X)$ is isomorphic as an orthomodular poset to $\Fact(X/\theta)$.
\label{interval}
\end{thm}

The second result was established in \cite{Tim}. It is much more difficult to establish than one may first think and involves a considerable amount of combinatorics. The reader may recall that it was mentioned in the introduction that for $|X|=27$, that $\Fact(X)$ has some $10^{22}$ elements, so a simple-minded approach to the computation of its automorphism group is not feasible. 

\begin{thm}
If $Y$ is a set with $27$ elements, then every automorphism $\alpha$ of the orthomodular poset $\Fact(Y)$ is given by $\Gamma(\sigma)$ for some permutation $\sigma$ of $Y$. 
\label{27}
\end{thm}

The general strategy of the proof is similar to that of several theorems about Hilbert spaces, such as Gleason's theorem. We push the problem down to the bottom of the structure $\Fact(X)$, use the fact that an interval near the bottom of $\Fact(X)$ is isomorphic to $\Fact(Y)$ for some smaller set $Y$, have a result for $\Fact(Y)$ when $Y$ is a suitable small set, and then push this back to a result for all of $\Fact(X)$. In proofs with Hilbert spaces one usually pushes things down to an interval isomorphic to the subspaces of $\mathbb{R}^3$. Here we push things to an interval isomorphic to $\Fact(Y)$ for $Y$ a 27-element set. The first step is to understand the bottom of $\Fact(X)$. 

\begin{defn}
Let $P$ be a poset with least element $0$ and largest element $1$, and let $a,c\in P$. 
\vspace{1ex}
\begin{enumerate}
\item $a$ is an atom of $P$ if $0<a$ and there is no $x$ with $0<x<a$
\item $c$ is a coatom of $P$ if $c<1$ and there is no $y$ with $c<y<1$
\item $a$ is of finite height if there is a finite bound on the size of chains in $[0,a]$
\item $c$ is of cofinite height if there is a finite bound on the size of chains in $[c,1]$
\end{enumerate}
\end{defn}

Section~5 is devoted to proving the following two results. 

\begin{prop}
Let $(\theta,\theta')$ be a factor pair of an infinite set $X$. 
\vspace{1ex}
\begin{enumerate}
\item $(\theta,\theta')$ is an atom iff $\theta'$ is a $p$-relation for some prime $p$
\item $(\theta,\theta')$ is of finite height iff $\theta'$ is finite regular
\end{enumerate}
\label{lunk}
\end{prop}

A simple argument using the orthocomplementation $\perp$ provides corresponding results for coatoms and elements of cofinite height. If $(\theta,\theta')$ is an atom and $\theta'$ is $p$-regular, we call $(\theta,\theta')$ a $p$-atom. 

\begin{prop}
Let $(\theta,\theta')$ be a factor pair of $X$ with $X/\theta$ infinite and let $p$ be prime. 
\vspace{1ex}
\begin{enumerate}
\item $(\theta,\theta')$ is the join of the $p$-atoms beneath it
\item $\theta\,\,=\bigcap\{\phi\,:\mbox{there is a $p$-atom $(\phi,\phi')$ with }(\phi,\phi')\leq(\theta,\theta')\}$
\item $\theta'=\hspace{.3mm}\bigcup\{\phi':\mbox{there is a $p$-atom $(\phi,\phi')$ with }(\phi,\phi')\leq(\theta,\theta')\}$
\end{enumerate}
\label{b}
\end{prop}

Section~6 establishes the following. 

\begin{prop}
Let $(\theta,\theta')\in\Fact(X)$ with $\alpha(\theta,\theta')=(\gamma,\gamma')$. Then 
\vspace{1ex}
\begin{enumerate}
\item $\theta\,$ is an $m$-relation iff $\gamma$ is an $m$-relation
\item $\theta'$ is an \hspace{.3mm}$n$-relation iff $\gamma'$ is an $n$-relation
\end{enumerate}
\label{e}
\end{prop}

Section~7 provides background that is used in Section~8 where we begin the process of taking an automorphism of $\Fact(X)$ and building from it a permutation of $X$. Section~8 does half of the work, taking an automorphism of $\Fact(X)$ and building from it an automorphism of $\Eq*(X)$. Here, and in the following, we assume that $X$ is an infinite set and that $\alpha$ is an automorphism of $\Fact(X)$. Section~8 establishes the following two results. 

\begin{prop}
Let $(\theta,\theta')$, $(\theta,\theta'')\in \Fact(X)$ with $X/\theta$ infinite. Then $\alpha(\theta,\theta')$ and $\alpha(\theta,\theta'')$ have the same first component. 
\label{f}
\end{prop}

A finite regular equivalence relation $\theta$ has infinitely many blocks, so $X/\theta$ is infinite. So this result, together with Proposition~\ref{e} combine to give the half-way result that was the aim. 

\begin{thm}
There is an automorphism $\beta$ of $\Eq*(X)$ with the following properties.
\vspace{1ex}
\begin{enumerate}
\item $\theta$ is an $n$-relation iff $\beta(\theta)$ is an $n$-relation
\item if $(\theta,\theta')\in\Fact(X)$ with $\theta\in\Eq*(X)$, then $\beta(\theta)$ is the first component of $\alpha(\theta,\theta')$
\end{enumerate}
\label{g}
\end{thm}

Next comes the second half of the process, taking an automorphism $\beta$ of $\Eq*(X)$ with the above properties, and producing from it a permutation $\sigma$ of $X$. This is the content of Sections~9 through~11 where we prove the following. 

\begin{thm}
Suppose $\beta$ is an automorphism of $\Eq*(X)$ so that for every relation $\theta$ we have that $\theta$ is an $n$-relation iff $\beta(\theta)$ is an $n$-relation. Then there is a permutation $\sigma$ of $X$ so that for every $2$-relation $\pi$ we have $\beta(\pi)=\sigma(\pi)$. 
\label{todo}
\end{thm}

We now show how these results to be proved in later sections prove Theorem~\ref{main'}, hence also Theorem~\ref{main}. Here $\alpha$ is an automorphism of $\Fact(X)$, $\beta$ is the automorphism of $\Eq*(X)$ for this $\alpha$ given by Theorem~\ref{g}, and $\sigma$ is the permutation of $X$ for this $\beta$ given by Theorem~\ref{todo}. 

\begin{prop}
\label{j}
If $(\rho,\rho')$ is finite height in $\Fact(X)$, the first coordinate of $\alpha(\rho,\rho')$ is $\sigma(\rho)$. 
\end{prop}

\begin{proof}
Suppose $(\rho,\rho')$ is of finite height and that $\alpha(\rho,\rho')=(\lambda,\lambda')$. Since $\alpha$ is compatible with orthocomplementation, $\alpha(\rho',\rho)=(\lambda',\lambda)$. Since $(\rho,\rho')$ is of finite height, Proposition~\ref{lunk} shows that $\rho'$ is finite resgular, hence by Proposition~\ref{e} so is $\lambda'$. Then $X/\rho'$ and $X/\lambda'$ are infinite. So we may apply Proposition~\ref{b} part (3) to $(\rho',\rho)$ and to $(\lambda',\lambda)$ to obtain 
\begin{align*}
\rho &= \bigcup\,\{\theta:\mbox{there is a 2-atom }(\theta',\theta)\leq(\rho',\rho)\}\\
\lambda&=\bigcup\,\{\gamma:\mbox{there is a 2-atom }(\gamma',\gamma)\leq(\lambda',\lambda)\}
\end{align*}
\vspace{-1ex}

By Proposition~\ref{e} the image and inverse image under $\alpha$ of a 2-atom is a 2-atom. So $(\gamma',\gamma)$ is a 2-atom with $(\gamma',\gamma)\leq(\lambda',\lambda)$ iff $(\gamma',\gamma)=\alpha(\theta',\theta)$ for some 2-atom $(\theta',\theta)$ with $(\theta',\theta)\leq(\rho',\rho)$. So 
\[\lambda=\bigcup\,\{\gamma:(\gamma',\gamma)=\alpha(\theta',\theta)\mbox{ for some 2-atom }(\theta',\theta)\leq(\rho',\rho)\}\]
\vspace{-0ex}

\noindent If $(\theta',\theta)$ is a 2-atom, then $\theta$ is a 2-relation. So by Theorem~\ref{g} the first coordinate of $\alpha(\theta,\theta')$ is $\beta(\theta)$, and using orthocomplementation the second coordinate of $\alpha(\theta',\theta)$ is $\beta(\theta)$. Thus 
\[\lambda=\bigcup\{\beta(\theta):\mbox{there is a 2-atom }(\theta',\theta)\leq(\rho',\rho)\}\]
\vspace{-2ex}

\noindent Theorem~\ref{todo} says that for a 2-relation $\theta$, that $\beta(\theta)=\sigma(\theta)$. It is easily seen that the operation of taking the image of a relation under $\sigma$ commutes with arbitrary unions. It then follows that $\lambda=\sigma(\bigcup\{\theta:\mbox{there is a 2-atom }(\theta',\theta)\leq(\rho',\rho)\})$. Therefore $\lambda=\sigma(\rho)$ as required. 
\end{proof}

\begin{prop}
\label{k}
For each $(\theta,\theta')\in\Fact(X)$, the first coordinate of $\alpha(\theta,\theta')$ is $\sigma(\theta)$. 
\end{prop}

\begin{proof}
If $(\theta,\theta')$ is of finite height, then the result is given by Proposition~\ref{j}. So assume that $(\theta,\theta')$ is not of finite height and let $\alpha(\theta,\theta')=(\lambda,\lambda')$. Since $\alpha$ is an automorphism, it follows that $(\lambda,\lambda')$ is not of finite height. Propositions~\ref{lunk} and~\ref{zoop} imply that $X/\theta$ and $X/\lambda$ are infinite. So we may apply Proposition~\ref{b} part (2) to $(\theta,\theta')$ and $(\lambda,\lambda')$ to obtain 
\begin{align*}
\theta&=\bigcap\,\{\phi:\mbox{there is a 2-atom }(\phi,\phi')\leq(\theta,\theta')\}\\
\lambda&=\bigcap\,\{\gamma:\mbox{there is a 2-atom }(\gamma,\gamma')\leq(\lambda,\lambda')\}
\end{align*}
\vspace{-1ex}

Using the same reasoning as in Proposition~\ref{j}, $(\gamma,\gamma')$ is a 2-atom with $(\gamma,\gamma')\leq(\lambda,\lambda')$ iff $(\gamma,\gamma')=\alpha(\phi,\phi')$ for some 2-atom $(\phi,\phi')$ with $(\phi,\phi')\leq(\theta,\theta')$. So 
\[\lambda=\bigcap\,\{\gamma:(\gamma,\gamma')=\alpha(\phi,\phi')\mbox{ for some 2-atom }(\phi,\phi')\leq(\theta,\theta')\}\]
\vspace{-1ex}

If $(\phi,\phi')$ is a 2-atom, then Proposition~\ref{j} gives that the first coordinate of $\alpha(\phi,\phi')$ is $\sigma(\phi)$. It is easily seen that the operation of taking the image of a relation under $\sigma$ commutes with taking arbitrary intersections. It then follows that $\lambda=\sigma(\bigcap\{\phi:\mbox{there is a 2-atom }(\phi,\phi')\leq(\theta,\theta')\})$. Therefore $\lambda=\sigma(\theta)$ as required. 
\end{proof}

\begin{thm}
\label{m}
For each $(\theta,\theta')\in\Fact(X)$ we have $\alpha(\theta,\theta')=(\sigma(\theta),\sigma(\theta'))$. 
\end{thm}

\begin{proof}
Suppose $\alpha(\theta,\theta')=(\lambda,\lambda')$. Since $\alpha$ is compatible with orthocomplementation, $\alpha(\theta',\theta)=(\lambda',\lambda)$. 
Proposition~\ref{k} shows that the first coordinate of $\alpha(\theta,\theta')$ is $\sigma(\theta)$, and that the first coordinate of $\alpha(\theta',\theta)$ is $\sigma(\theta')$. So $\lambda=\sigma(\theta)$ and $\lambda'=\sigma(\theta')$. 
\end{proof}

\section{Atoms in $\Fact(X)$}

In this section we prove Propositions~\ref{lunk} and~\ref{b}, which are restated as Propositions~\ref{lunk'} and~\ref{b'} below. 

\begin{prop}
Let $(\theta,\theta')$ be a factor pair of an infinite set $X$. 
\vspace{1ex}
\begin{enumerate}
\item $(\theta,\theta')$ is an atom iff $\theta'$ is a $p$-relation for some prime $p$
\item $(\theta,\theta')$ is of finite height iff $\theta'$ is finite regular
\end{enumerate}
\label{lunk'}
\end{prop}

\begin{proof}
By Theorem~\ref{interval} the interval $[(\nabla,\Delta),(\theta,\theta')]$ is isomorphic to $\Fact(X/\theta)$. Having $(\theta,\theta')$ be an atom is therefore equivalent to having $\Fact(X/\theta)$ have exactly two elements, which in turn is equivalent to having $X/\theta$ be directly indecomposible. This amounts to $X/\theta$ being finite with a prime number of elements $p$, and therefore to $\theta$ having $p$ equivalence classes. By Proposition~\ref{zoop}, this is equivalent to $\theta'$ being a $p$-relation. This establishes (1). For (2), it is not difficult to show that there is a bound on the lengths of chains in $\Fact(X/\theta)$ iff $X/\theta$ is finite. But by Proposition~\ref{zoop} this is equivalent to $\theta'$ being finite regular. 
\end{proof}

\begin{prop}
Let $(\theta,\theta')$ be a factor pair of $X$ with $X/\theta$ infinite and let $p$ be prime. 
\vspace{1ex}
\begin{enumerate}
\item $(\theta,\theta')$ is the join of the $p$-atoms beneath it
\item $\theta\,\,=\bigcap\{\phi\,:\mbox{there is a $p$-atom $(\phi,\phi')$ with }(\phi,\phi')\leq(\theta,\theta')\}$
\item $\theta'=\hspace{.3mm}\bigcup\{\phi':\mbox{there is a $p$-atom $(\phi,\phi')$ with }(\phi,\phi')\leq(\theta,\theta')\}$
\end{enumerate}
\label{b'}
\end{prop}

\begin{proof} 
To simplify notation and allow illustrative diagrams, we assume $p=2$. The modifications for other primes are discussed after the result is proved for $p=2$. 
Enumerate the equivalence classes of $\theta'$ by $E_i$ $(i\in I)$ for some cardinal $I$, and enumerate the equivalence classes of $\theta$ by $F_j$ $(j\in J)$ for some cardinal $J$. Then each element $x$ of $X$ determines an ordered pair $(i,j)$ where $x\in E_i$ and $x\in F_j$. So for elements $x,y\in X$ we have 
\begin{align*}
x\,\theta'\,y\,\, & \mbox{ iff  $x,y$ have the same $i$ component}\\
x\,\theta\,y\,\,\, & \mbox{ iff  $x,y$ have the same $j$ component}
\end{align*}
\vspace{-2ex}

Since $\theta\cap\theta'=\Delta$, different elements of $X$ give different ordered pairs. For each ordered pair $(i,j)$, take some element $y\in E_i$ and some element $z\in F_j$. Since $\theta'\circ\theta=\nabla$, there is $x$ with $y\,\theta'\, x\,\theta z$. Then $x\in E_i$ and $x\in F_j$, so the ordered pair for $x$ is $(i,j)$. Thus, there is a bijective correspondence between the elements of $X$ and the ordered pairs $(i,j)$ where $i\in I$ and $j\in J$. So we assume the elements of $X$ are indexed as $x_{i,j}$ where $i\in I$ and $j\in J$. 
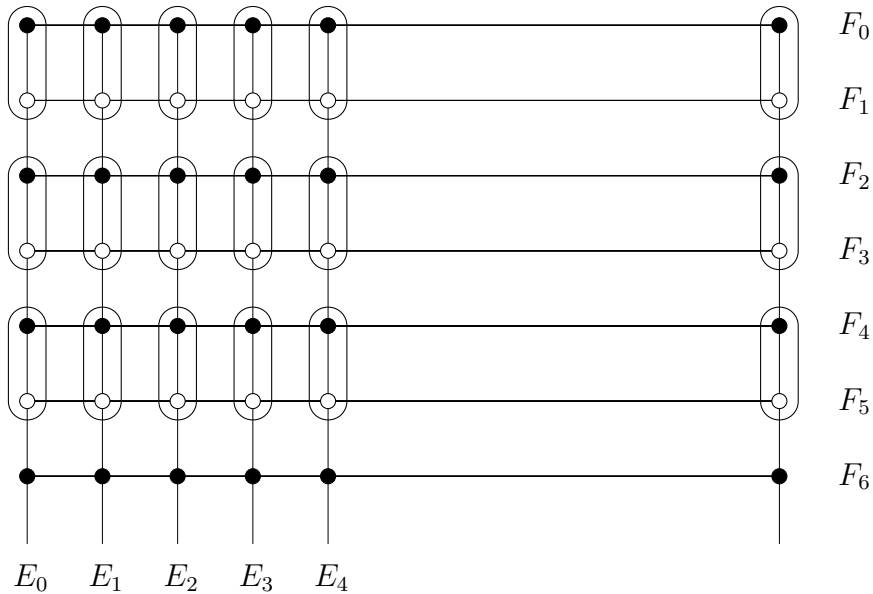
\vspace{4ex}

\begin{center}
\begin{tikzpicture}
\foreach \x in {0,1,2,3,4,10} {  \draw[ultra thin] (\x,-.9) -- (\x,6);
   \foreach \y in {0,1,2,3,4,5,6} {  \draw[ultra thin] (0,\y) -- (10,\y); } }
\foreach \x in {0,1,2,3,4,10} { \foreach \y in {5,3,1} {\draw[rounded corners=7pt] (\x - .25,\y - .25) rectangle ++(0.5,1.5); } }
\foreach \x in {0,1,2,3,4,10} { \foreach \y in {0,2,4,6} { \draw[fill] (\x, \y) circle [radius=0.1]; } }
\foreach \x in {0,1,2,3,4,10} { \foreach \y in {1,3,5} { \draw[fill=white] (\x, \y) circle [radius=0.1]; } }
\foreach \x in {0,1,2,3,4} { \node at (\x + .05,-1.35) {$E_{\x}$}; }
\foreach \y in {0,1,2,3,4,5,6} { \node at (11,6-\y) {$F_{\y}$}; }
\end{tikzpicture}
\end{center}
\vspace{4ex}

In the figure we have drawn the elements $x_{i,j}$ with the top row being $x_{0,0}, x_{1,0}, \ldots$. It depicts a final entry in this row, as there would be if $I$ is finite, but this need not be the case and is not of importance in the argument that follows. The second row from the top has its elements $x_{0,1},x_{1,1},\ldots$, and so forth. Of importance is that $J$ is infinite, and so there are infinitely many rows in this picture. The horizontal lines in this picture indicate $\theta$-classes, and the vertical lines indicate $\theta'$-classes. Also indicated are two more equivalence relations $\phi$ and $\phi'$. The equivalence relation $\phi$ has two equivalence classes, one consisting of those elements represented by filled in circles, and the other those elements represented by open circles. The equivalence relation $\phi'$ has infinitely many equivalence classes, each with exactly two elements. These are represented by the ovals. 

We describe $\phi$ and $\phi'$ more precisely. Each ordinal $\alpha$ has a Cantor normal form. In this, it has a finite portion that is an ordinary natural number. We call this the finite part of $\alpha$. For example, the ordinal $\omega^3 + 5\omega^2 + 4\omega + 11$ has finite part 11. We can then define whether an ordinal is congruent to $0$ or $1$ modulo 2 in an obvious way. For the elements $x_{i,j}$, we have that $i$ ranges over all ordinals less than $I$ and $j$ over all ordinals less than $J$. Then we define $\phi$ so that $x_{i,j}$ is $\phi$ related to $x_{p,q}$ if $j$ is congruent to $q$ modulo 2. Roughly, the even rows belong to one class and the odd rows to the other. Each ordinal $\alpha$ has its finite part equal to either $2k$ for some natural number $k$, or to $2k+1$ for some natural number $k$. The equivalence classes of $\phi'$ are defined to be all sets $\{x_{i,j},x_{i,j+1}\}$ where $j$ is an ordinal less than $J$ whose finite part equals $2k$ for some natural number $k$, and $j+1$ is ordinal addition. Since $J$ is infinite, $j<J$ implies that $j+1<J$, so this is well defined. 

It follows from this description that $\phi\cap\phi'=\Delta$. Also, $\phi\circ\phi'=\nabla$. To see this, let $x,y$ be elements of $X$. We show $(x,y)\in\phi\circ\phi'$. If they belong to the same $\phi$-class, then this is obvious. Otherwise, we can move from $x$ along a $\phi'$ oval to an element in the same class as $\phi$ class as $y$. It is also clear that $\theta\subseteq\phi$ since the $\phi$-classes are unions of the rows that are $\theta$-classes. Also $\phi'\subseteq\theta'$ since the $\phi'$ classes are 2-element sets that belong to the columns that are the $\theta'$-classes. To show that $(\phi,\phi')\leq(\theta,\theta')$, it remains to show that $\theta$ permutes with $\phi'$. Suppose $x\,\theta\, y\,\phi' z$. Then $x$ and $y$ are in the same row, and $z$ is either in the same row as $x$, or the one above or below it, depending on whether $x$'s row is a $2k$ row or a $2k+1$ row. Then it is easy to see that there is $w$ with $x\,\phi'\,w\,\theta\,z$. So $\theta\circ\phi'=\phi'\circ\theta$. 

Since $\phi$ has two equivalence classes, the factor pair $(\phi,\phi')$ is a 2-atom, and we have shown that $(\phi,\phi')\leq(\theta,\theta')$. Consider 

\[B = \{(\phi,\phi'):(\phi,\phi') \mbox{ is a 2-atom with }(\phi,\phi')\leq(\theta,\theta')\}\]
\vspace{-1ex}

\noindent We show that (2) $\theta=\bigcap\{\phi:(\phi,\phi')\in B\}$ and (3) $\theta'=\bigcup\{\phi':(\phi,\phi')\in B\}$. The conditions for $(\phi,\phi')$ to be in $B$ require that $\theta\subseteq\phi$, $\phi'\subseteq\theta'$, and $\theta\circ\phi'=\phi'\circ\theta$. From this, it is clear that in (2) the containment  $\subseteq$ holds, and in (3) the containment $\supseteq$ holds. The key to obtaining the other containments is noting that altering the enumerations of the equivalence classes of $\theta$ and $\theta'$ will produce other pairs $(\phi,\phi')$ belonging to $B$. 

Given any two distinct elements $u,v$ of $X$ that are $\theta'$ related, $u,v$ are in the same $\theta'$ class $E$, and so are in different $\theta$ classes $F,F'$. When we enumerated the $\theta'$ classes, we could have enumerated them so that $E=E_0$, and when we enumerated the $\theta$ classes, we could have enumerated them so that $F=F_0$ and $F'=F_1$. This would have had the elements $u,v$ indexed as $u=x_{0,0}$ and $v=x_{0,1}$. With this alternate enumeration, we would have obtained a factor pair $(\phi,\phi')\in B$ with $(u,v)\in\phi'$. Since $(u,v)$ was an arbitrary ordered pair in $\theta'$, it follows that $\theta'=\bigcup\{\phi':(\phi,\phi')\in B\}$, establishing (3). 

To establish (2) suppose that $(u,v)\not\in\theta$. Then $u$ is in one $\theta$ class $F$ and $v$ is in another $F'$. When we enumerate the $\theta$ classes, we can do so with $F=F_0$ and $F'=F_1$. Take an arbitrary enumeration of the $\theta'$ classes. Then with these enumerations, $u=x_{i,0}$ and $v=x_{p,1}$ for some $i,p$. We then have another $(\phi,\phi')\in B$ and this time $(u,v)\not\in\phi$ since their second components are not congruent modulo 2. So $(u,v)\not\in\bigcap\{\phi:(\phi,\phi')\in B\}$, and it follows that (2) holds. 

Surely $(\theta,\theta')$ is an upper bound of $B$ in $\Fact(X)$. Suppose that $(\gamma,\gamma')$ is another upper bound. Then for each $(\phi,\phi')\in B$ we have $\gamma\subseteq\phi$, $\phi'\subseteq\gamma'$, and $\gamma\circ\phi'=\phi'\circ\gamma$. It follows that $\gamma\subseteq \bigcap\{\phi:(\phi,\phi')\in B\}$ and that $\bigcup\{\phi':(\phi,\phi')\in B\}\subseteq\gamma'$. Thus, by (2) and (3) we have $\gamma\subseteq\theta$ and $\theta'\subseteq\gamma'$. To show that $(\theta,\theta')\leq(\gamma,\gamma')$ it remains only to show that $\gamma\circ\theta'=\theta'\circ\gamma$. Suppose $x\,\gamma\,y\,\theta'\,z$. By (3) there is some $(\phi,\phi')\in B$ with $(y,z)\in\phi'$. But $(\phi,\phi')\leq(\gamma,\gamma')$ implies that $\gamma$ and $\phi'$ permute. Then since $x\,\gamma\,y\,\phi'\,z$, there is some $w$ with $x\,\phi'\,w\,\gamma\,z$, hence $x\,\theta'\,w\,\gamma\,z$. So $\gamma\circ\theta'\subseteq\theta'\circ\gamma$. A similar argument shows the other containment. So $(\theta,\theta')\leq(\gamma,\gamma')$. Thus $(\theta,\theta')$ is the least upper bound of $B$. This establishes the proof when $p=2$. 

To modify the proof to accommodate an arbitrary prime $p$, make the $\phi$ equivalence classes modulo p rather than modulo 2, and make the $\phi'$ classes consist of sets $\{x_{i,j},x_{i,j+1},\ldots,x_{i,j+p-1}\}$ where the finite part of $j$ is equal to $pk$ for some integer $k$. 
\end{proof}

\begin{rmk}{\em 
Proposition~\ref{b'} need not hold without the assumption that $X/\theta$ is finite. Indeed, suppose that $|X/\theta|=n$. By Theorem~\ref{interval}, the interval $[(\nabla,\Delta),(\theta,\theta')]$ is isomorphic to $\Fact(X/\theta)$. If $p$ does not divide $n$, then $\Fact(X/\theta)$ will have no $p$-relations at all. 
}
\end{rmk}

We recall that a poset with a least element $0$ is atomic if every non-zero element has an atom beneath it, and the poset is atomistic if each element is the least upper bound of the atoms that lie beneath it. Proposition~\ref{b'} has the following consequence that is of independent interest. 

\begin{cor}
For $X$ an infinite set, $\Fact(X)$ is atomistic. 
\end{cor}

\section{Preservation of $n$-relations}

In this section we prove Proposition~\ref{e}, that if $\theta$ is an $n$-relation and $\alpha(\theta,\theta')=(\gamma,\gamma')$, then $\gamma$ is an $n$-relation, and conversely. This is restated below as Proposition~\ref{e'}. We begin with the following result from \cite[Cor.~2.3]{Tim}, restated slightly to match our terminology.

\begin{prop}
Let $X$ be a finite set with $mn$ elements. Then the number of factor pairs $(\theta,\theta')$ where $\theta$ is an $n$-relation is 
\[\frac{(mn)!}{m!n!}\]
\vspace{0ex}
\label{snakes}
\end{prop} 

\begin{cor}
If $X$ is a finite set with $|X|=pq$ for primes $p$ and $q$, then 

\[|\Fact(X)\,| =  \begin{cases}  
\,2\dfrac{(pq)!}{p!q!}+2 &\mbox{if } p\neq q\\[1.5em]
\,\,\dfrac{(p^2)!}{(p!)^2}+2 &\mbox{if } p=q
\end{cases}\]
\vspace{1ex}
\label{size}
\end{cor}

\begin{proof}
In this situation the elements of $\Fact(X)$ are the two bounds, the elements $(\theta,\theta')$ where $\theta$ is a $p$-relation and $\theta'$ is a $q$-relation, and the elements $(\theta,\theta')$ where $\theta$ is a $q$-relation and $\theta'$ is a $p$-relation. The result follows from Proposition~\ref{snakes}.
\end{proof}

\begin{prop}
Let $X$ be an arbitrary set, and let $(\gamma,\gamma')\in\Fact(X)$. Then for any $m,n$ natural numbers, these are equivalent 
\vspace{1ex}
\begin{enumerate}
\item $|X/\gamma\,|=mn$ 
\item there are $(\theta,\theta')\oplus(\phi,\phi')=(\gamma,\gamma')$ with $|X/\theta\,|=m$ and $|X/\phi\,|=n$
\end{enumerate}
\label{adz}
\end{prop}

\begin{proof}
``$\Rightarrow$'' Let $f:X\to X/\gamma\times X/\gamma'$ be the natural bijection. Let $A$ be an $m$-element set and let $B$ be an $n$-element set. Since $X/\gamma$ has $mn$ elements, there is bijection $g:X/\gamma\to A\times B$. From this we produce a bijection $h:X\to A\times B\times X/\gamma'$. Let $\pi_A$ and $\pi_B$ be the projections from $A\times B\times X/\gamma'$ to $A$ and $B$, then set $\theta=\ker(\pi_A\circ h)$ and $\phi=\ker(\pi_B\circ h)$. As described in \cite{Harding 1} and discussed in Section~2, the following is a Boolean sublattice of $\Eqq(X)$ consisting of pairwise permuting equivalence relations. Then $(\theta,\theta')\oplus(\phi,\phi')=(\gamma,\gamma')$ with $|X/\theta\,|=m$ and $|X/\phi\,|=n$. 

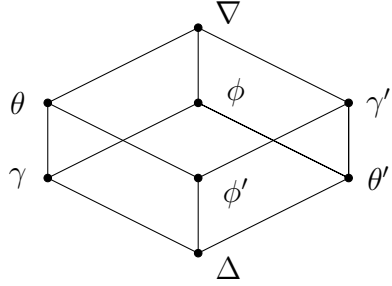
\begin{center}
\begin{tikzpicture}
\draw[fill] (0,0) circle [radius=0.05]; 
\draw[fill] (-2,1) circle [radius=0.05];
\draw[fill] (0,1) circle [radius=0.05];
\draw[fill] (2,1) circle [radius=0.05];
\draw[fill] (-2,2) circle [radius=0.05];
\draw[fill] (0,2) circle [radius=0.05];
\draw[fill] (2,2) circle [radius=0.05];
\draw[fill] (0,3) circle [radius=0.05];
\draw[thin] (0,0) -- (-2,1) -- (0,2) -- (2,1) -- (0,2);
\draw[thin] (-2,1) -- (-2,2) -- (0,3) -- (2,2) -- (2,1) -- (2,2) -- (0,1) -- (-2,2);
\draw[thin] (0,2) -- (0,3);
\draw[thin] (0,1) -- (0,0) -- (2,1);
\node at (.4,-.2) {$\Delta$};
\node at (.4,3.2) {$\nabla$};
\node at (-2.4,2) {$\theta$};
\node at (-2.4,1) {$\gamma$};
\node at (.5,.8) {$\phi'$};
\node at (.5,2.15) {$\phi$};
\node at (2.4,1) {$\theta'$};
\node at (2.4,2) {$\gamma'$};
\end{tikzpicture}
\end{center}

``$\Leftarrow$'' The assumption of (2) provides that the diagram above is a Boolean sublattice of $\Eqq(X)$ consisting of pairwise permuting equivalence relations. Then $X/\gamma$ is naturally isomorphic to $X/\theta\times X/\phi$, and it follows that $|X/\gamma\,|=mn$. 
\end{proof}

We now begin to specialize these results to our current setting. Here, and in the remainder of this section, $X$ will be an infinite set and $\alpha$ a given automorphism of $\Fact(X)$. 

\begin{prop}
For a prime $p$, $(\theta,\theta')$ is a $p$-atom of $\Fact(X)$ iff $\alpha(\theta,\theta')$ is a $p$-atom. 
\label{crud}
\end{prop}

\begin{proof}
We first prove this for $p=2$. To do so, we show that $(\theta,\theta')$ is a 2-atom iff there is an atom $(\phi,\phi)$ that is orthogonal to it with the interval $[(\nabla,\Delta),(\theta,\theta')\oplus(\phi,\phi')]$ having exactly 8 elements. Since this is a property that is preserved by any automorphism, and its inverse, the result for $p=2$ will follow. 

$p=2$ ``$\Rightarrow$'' Suppose that $(\theta,\theta')$ is a 2-atom. Then $X/\theta'$ is infinite. So there is a 2-element set $A$ and an infinite set $B$ with a bijection $f:X/\theta'\to A\times B$. Let $\phi=\ker(\pi_A\circ f)$ and $\gamma'=\ker(\pi_B\circ f)$. Then the natural map $X\to X/\theta\times X/\phi\times X/\gamma'$ is a bijection. So there again is a Boolean sublattice of $\Eqq(X)$ that consists of pairwise permuting relations, contains $\theta,\phi$ and $\gamma'$, and is labelled exactly as in the diagram in Proposition~\ref{adz}. Setting $\phi'=\theta\cap\gamma'$ gives a 2-atom $(\phi,\phi')$ that is orthogonal to $(\theta,\theta')$ with $(\theta,\theta')\oplus(\phi,\phi')=(\gamma,\gamma')$. Since $X/\theta$ and $X/\theta'$ both have 2 elements, Proposition~\ref{adz} shows that $X/\gamma$ has 4 elements. So by Corollary~\ref{size}, $\Fact(X/\gamma)$ has 8 elements. Theorem~\ref{interval} shows that $[(\nabla,\Delta),(\gamma,\gamma')]$ is isomorphic to $\Fact(X/\gamma)$, and therefore has 8 elements. 

$p=2$ ``$\Leftarrow$'' Suppose that $(\phi,\phi')$ is an atom orthogonal to $(\theta,\theta')$ with $(\theta,\theta')\oplus(\phi,\phi')=(\gamma,\gamma')$ and that the interval $[(\nabla,\Delta),(\gamma,\gamma')]$ has 8 elements. Since $(\theta,\theta')$ and $(\phi,\phi')$ are atoms, $X/\theta$ has $p$ elements and $X/\phi$ has $q$ elements for some primes $p$ and $q$. By Proposition~\ref{adz}, $X/\gamma$ has $pq$ elements. By Theorem~\ref{interval}, the interval $[(\nabla,\Delta),(\gamma,\gamma')]$ is isomorphic to $\Fact(X/\gamma)$, so $\Fact(X/\gamma)$ has 8 elements. The only way to get the expression in Corollary~\ref{size} to equal 8 is to have $p=q=2$. Thus $(\theta,\theta')$ is a 2-atom. 

Having established the result for $p=2$ we extend it to arbitrary primes $p\neq 2$. We will show that $(\theta,\theta')$ is a $p$-atom iff it is not a 2-atom and there is a 2-atom $(\phi,\phi')$ that is orthogonal to $(\theta,\theta')$ so that
\[[(\nabla,\Delta),(\theta,\theta')\oplus(\phi,\phi')]\,\,\mbox{ has }\,\,\dfrac{(2p)!}{p!}+2\,\mbox{ elements }\]
\vspace{-1ex}

\noindent Since the property of being a 2-atom is preserved by automorphisms and their inverses, the above property is preserved by automorphisms and their inverses. It then follows that $(\theta,\theta')$ is a $p$-atom iff $\alpha(\theta,\theta')$ is a $p$-atom. 

$p\neq 2$ ``$\Rightarrow$'' Suppose that $(\theta,\theta')$ is a $p$-atom. Then $X/\theta'$ is infinite, and proceeding as in the forward direction of the $p=2$ case we can find a 2-atom $(\phi,\phi')$ that is orthogonal to $(\theta,\theta')$. Setting $(\theta,\theta')\oplus(\phi,\phi')=(\gamma,\gamma')$, Proposition~\ref{adz} gives that $X/\gamma$ has $2p$ elements. Theorem~\ref{interval} gives that $[(\nabla,\Delta),(\gamma,\gamma')]$ is isomorphic to $\Fact(X/\gamma)$, and therefore by Corollary~\ref{size} has the required number of elements. 

$p\neq 2$ ``$\Leftarrow$'' Suppose that $(\phi,\phi')$ is a 2-atom orthogonal to $(\theta,\theta')$ with $(\theta,\theta')\oplus(\phi,\phi')=(\gamma,\gamma')$ and that 
\[[(\nabla,\Delta),(\gamma,\gamma')]\,\,\mbox{ has }\,\,\dfrac{(2p)!}{p!}+2\,\mbox{ elements }\]
\vspace{1ex}

\noindent Since $(\theta,\theta')$ is an atom that it not a 2-atom, $X/\theta$ has $q$ elements for some prime $q\neq 2$. Then by Proposition~\ref{adz}, $X/\gamma$ has $2q$ elements. So by Corollary~\ref{size}

\[ \Fact(X/\gamma)\,\mbox{ has }\,\dfrac{(2q)!}{q!}+2\,\mbox{ elements}\]
 \vspace{1ex}
 
 \noindent By Theorem~\ref{interval}, $[(\nabla,\Delta),(\gamma,\gamma')]$ is isomorphic to $\Fact(X/\gamma)$, so 

\[\dfrac{(2p)!}{p!}+2=\dfrac{(2q)!}{q!}+2\]
\vspace{-.5ex}

We claim that this implies that $p=q$, concluding the proof. Suppose not, and assume without loss of generality that $p>q$, and let $r$ be such that $p=q+r$. Then an obvious simplification of the above equation yields  
\[ \dfrac{(2p)!}{(2q)!} = \dfrac{p!}{q!}\]
\vspace{-2ex}

\noindent This then yields 
\[(2p)(2p-1)\cdots (2p-r+1)(2p-r)\cdots  (2p-2r+1) = p(p-1)\cdots (p-r+1)\]
\vspace{-2ex}

\noindent The left side of this equation has $2r$ terms. The first $r$ terms on the left are given by $2p-i$ for $i=0,\ldots,r-1$, and the $r$ terms on the right are $p-i$ for $i=0,\ldots,r-1$. Each of the first $r$ terms on the left side is greater than its partner on the right. This is impossible, so we must in fact have $p=q$. 
\end{proof}

We come now to the result that is the purpose of this section. 

\begin{prop}
Let $(\theta,\theta')\in\Fact(X)$ with $\alpha(\theta,\theta')=(\gamma,\gamma')$. Then 
\vspace{1ex}
\begin{enumerate}
\item $\theta\,$ is an $m$-relation iff $\gamma$ is an $m$-relation
\item $\theta'$ is an \hspace{.3mm}$n$-relation iff $\gamma'$ is an $n$-relation
\end{enumerate}
\label{e'}
\end{prop}

\pagebreak[3]

\begin{proof}
We prove the second statement. Then the first follows since $\alpha(\theta',\theta)=(\gamma',\gamma)$. So suppose that $\theta'$ is an $n$-relation and that $n=p_1\cdots p_k$ is the prime factorization of $n$ where repetitions of the same prime $p$ are allowed in this list. The natural map $X\to X/\theta\times X/\theta'$ is a bijection since $(\theta,\theta')$ is a factor pair, Also, since $\theta'$ is an $n$-relation, $|X/\theta\,|=n$. So there are sets $A_i$ for $i=1,\ldots,k$ with $|A_i|=p_i$ and a bijection 
\[f:X\to A_1\times\cdots\times A_k\times X/\theta'\]

Let $\phi_i=\ker(\pi_i\circ f)$ for $i=1,\ldots,k$ where $\pi_i$ is the projection from this product onto $A_i$. By \cite[Prop.~2.4.3]{Harding regular} the relations $\phi_1,\ldots,\phi_k,\theta'$ are the coatoms of a Boolean sublattice of $\Eqq(X)$ consisting of pairwise permuting relations. Let $\phi_i'$ be the complement of $\phi_i$ in this Boolean lattice for each $i=1,\ldots,k$. Then $(\phi_i,\phi_i')$ is a factor pair with $|X/\phi_i\,|=p_i$ for each $i=1,\ldots,k$., and these atoms are pairwise orthogonal. These atoms have an orthogonal join $\bigoplus_{i=1}^k (\phi_i,\phi_i')$ and by \cite[Cor.~2.3.5]{Harding regular} this is equal to $(\theta,\theta')$. 

Since $\alpha$ is an automorphism, $\alpha(\theta,\theta')=(\gamma,\gamma')$ is the orthogonal join $\bigoplus_{i=1}^k\alpha(\phi_i,\phi_i')$ and by Proposition~\ref{crud} $\alpha(\phi_i,\phi_i')$ is a $p_i$-atom for each $i=1,\ldots,k$. Then with repeated application of Proposition~\ref{adz} we have that $|X/\gamma\,|=p_1\cdots p_k$. Hence $|X/\gamma\,|=n$, so $\gamma'$ is an $n$-relation. This provides one direction of item~(2). The other direction is obtained by applying what we have proved with the automorphism $\alpha^{-1}$.
\end{proof}


\section{Linking $p$-atoms}

In this section we provide a technical result used in the following section where we show that if $X/\theta$ is infinite, then the images $\alpha(\theta,\theta')$ and $\alpha(\theta,\theta'')$ of two factor pairs with the same first coordinate have the same first coordinate. The proof in the following section will reduce this to the corresponding result for two 3-atoms with the same first coordinate. We will need to know that any two such 3-atoms can be linked by ones that are in some sense close together. That is the content of this section. The result is easier to illustrate and understand for 2-atoms, and we do this first. Then we use this to prove the case of interest, that of 3-atoms. 

\begin{defn}
For a prime $p$, define a relation $\sim_p$ on the set of $p$-atoms of $\Fact(X)$ by setting $(\theta,\theta')\sim_p(\theta,\theta'')$ if they have the same first coordinate and there is a factor pair $(\gamma,\gamma')$ with
\vspace{1ex}
\begin{enumerate}
\item  $|X/\gamma|=p^3$
\item $(\theta,\theta')$ and $(\theta,\theta'')$ belong to $[(\nabla,\Delta),(\gamma,\gamma')]$
\end{enumerate}
\vspace{1ex}
Let $\approx_p$ be the transitive closure of $\sim_p$. 
\end{defn}

The purpose of this section is to prove that two 3-atoms are $\approx_3$-related if and only if they have the same first coordinate. We first show the corresponding result for 2-atoms. 

\begin{rmk}{\em We will need a special setup to precisely talk about several constructions we make, and face a challenge with notation. We will consider maps $\xi:\{A_i,\ldots,H_i:i\in I\}\to X$, and need to label diagrams indicating elements of $X$. Having the notation $\xi(A_i)$ occur many times on these diagrams makes them difficult to read. We use the convention of letting $$\xi(A_i)=a_i,\,\ldots,\,\xi(H_i)=h_i$$ \vspace{-1ex}

\noindent If another map $\xi':\{A_s,\ldots,H_s:s\in S\}\to X$ is required, we use $\xi'(A_s)=a_s',\ldots,\xi'(H_s)=h_s'$. }
\end{rmk}

\begin{defn}
A special 8-enumeration of $X$ consists of a set $I$ and a bijection 
\vspace{-1ex}

\[\xi:\{A_i,\ldots,H_i:i\in I\}\to X\]
\vspace{-2ex}

\noindent  We say $\xi$ is compatible with the factor pair $(\theta,\theta')$ being under $(\gamma,\gamma')$ if \vspace{-1ex}
\vspace{-1ex}

\begin{align*}
\theta\,\mbox{ has blocks }\, & \{a_i, c_i, e_i, g_i:i\in I\} \mbox{ and }\{ b_i, d_i, f_i, h_i:i\in I\}\\
\theta'\mbox{ has blocks }\,&\{ a_i, b_i\}, \{ c_i, d_i\}, \{ e_i, f_i\}, \mbox{ and }\{ g_i, h_i\}\,\mbox{ for each }i\in I\\
\gamma\,\mbox{ has blocks }\,&\{ a_i:i\in I\},\ldots,\{ h_i:i\in I\}\\
\gamma'\,\mbox{ has blocks }&\{ a_i,\ldots, h_i\}\,\mbox{ for each }\,i\in I
\end{align*}
\vspace{-2ex}

\noindent The situation is illustrated in Figure~\ref{2-atom setup}.
\label{special 8}
\end{defn}
\vspace{-2ex}

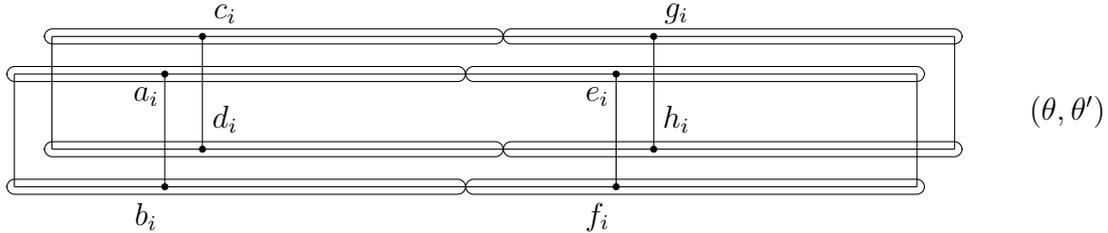
\begin{figure}[h]
\begin{center}
\begin{tikzpicture}
\foreach \x in {0,2,8,12} {  \draw (\x,0) -- (\x,1.5); \draw (\x+0.5,0.5) -- (\x+0.5,2); }
\foreach \x in {2,8} { \draw[fill] (\x, 0) circle [radius=0.04]; \draw[fill] (\x+0.5,0.5) circle [radius=0.04]; \draw[fill] (\x,1.5) circle [radius=0.04]; \draw[fill] (\x+0.5,2) circle [radius=0.04];}
\draw (0,0) -- (12,0); \draw (0,1.5) -- (12,1.5); \draw (0.5,0.5) -- (12.5,0.5); \draw (0.5,2) -- (12.5,2);
\draw[rounded corners=3pt] (-.1,-.1) rectangle ++(6.1,0.2); \draw[rounded corners=3pt] (6.0,-.1) rectangle ++(6.1,0.2);
\draw[rounded corners=3pt] (-.1,1.4) rectangle ++(6.1,0.2); \draw[rounded corners=3pt] (6.0,1.4) rectangle ++(6.1,0.2);
\draw[rounded corners=3pt] (.4,.4) rectangle ++(6.1,0.2); \draw[rounded corners=3pt] (6.5,.4) rectangle ++(6.1,0.2);
\draw[rounded corners=3pt] (.4,1.9) rectangle ++(6.1,0.2); \draw[rounded corners=3pt] (6.5,1.9) rectangle ++(6.1,0.2);
\node at (1.75,1.2) {$a_i$}; \node at (1.75,-.4) {$b_i$}; \node at (2.8,2.3) {$c_i$}; \node at (2.8,.9) {$d_i$};
\node at (7.75,1.2) {$e_i$}; \node at (7.75,-.4) {$f_i$}; \node at (8.8,2.3) {$g_i$}; \node at (8.8,.9) {$h_i$};
\node at (14,1) {$(\theta,\theta')$};
\end{tikzpicture}
\end{center}
\caption{A special 8-enumeration $\xi$ compatible with $(\theta,\theta')$ being under $(\gamma,\gamma')$}
\label{2-atom setup}
\end{figure}

In Figure~\ref{2-atom setup} the elements $a_i,\ldots,h_i$ $(i\in I)$ of $X$ are  placed on a cube with infinitely many points along its lengthwise edges, and two points on each of its vertical and depthwise edges. The $\theta$ classes are the bottom and top and top of the cube, and the $\theta'$ classes are vertical lines connecting the bottom to the top. The blocks of $\gamma$ are the 8 ovals, and the infinitely many blocks of $\gamma'$ have one element from each of these ovals, they are 8-element cubes. 

\begin{prop}
Let $(\theta,\theta')$ and $(\gamma,\gamma')$ be factor pairs. Then there is a special 8-relation $\xi$ that is compatible with $(\theta,\theta')$ being under $(\gamma,\gamma')$ iff the following conditions hold. 
\vspace{1ex}
\begin{enumerate}
\item $\theta$ has two blocks and $\theta'$ is a 2-relation
\item $\gamma$ has eight blocks and $\gamma'$ is an 8-relation
\item $(\theta,\theta')\leq(\gamma,\gamma')$
\end{enumerate}
\label{skunk}
\end{prop}

\begin{proof}
``$\Rightarrow$'' If $\xi$ is compatible with $(\theta,\theta')$ being under $(\gamma,\gamma')$, then the conditions describing the blocks of $\theta,\theta',\gamma,\gamma'$ imply that $\theta$ has two blocks, the blocks of $\theta'$ each have two elements, $\gamma$ has eight blocks, and the blocks of $\gamma'$ each have eight elements. It remains to show the third condition, that $(\theta,\theta')\leq(\gamma,\gamma')$. Note that $\gamma\subseteq\theta$ since the $\gamma$ blocks are the ovals that are contained in bottom and top which are the blocks of $\theta$. Also $\theta'\subseteq\gamma'$ since the blocks of $\theta'$ are the vertical lines that are part of the 8-element cubes that are the $\gamma'$ blocks. That $\gamma\circ\theta'=\theta'\circ\gamma$ follows since each composite gives an equivalence relation with blocks that are two ovals one above the other. 

``$\Leftarrow$'' As described in Proposition~\ref{duck}, having $(\theta,\theta')\leq(\gamma,\gamma')$ implies that there is a ternary direct product decomposition $f:X\to X_1\times X_2\times X_3$ where 

\[ \theta=\ker(\pi_1\circ f),\, \theta'=\ker(\pi_{2,3}\circ f),\, \gamma=\ker(\pi_{1,2}\circ f)\,\mbox{ and }\gamma'=\ker(\pi_3\circ f)\]
\vspace{-2ex}

\noindent Here $\pi_i:X_1\times X_2\times X_3\to X_i$, and $\pi_{i,j}:X_1\times X_2\times X_3\to X_i\times X_j$ are the natural projections. Since $X_1$ is isomorphic to $X/\theta$ it has two elements. Since $X/\gamma$ has eight elements and is isomorphic to $X_1\times X_2$ then the set $X_2$ has four elements. Therefore $X_3$ is infinite and with the same cardinality as $X$. Let $X_1=\{1,2\}$, $X_2=\{1,2,3,4\}$ and $X_3=I$ for some infinite cardinal $I$. Define a bijection $g:\{A_i,\ldots,H_i:i\in I\}\to X_1\times X_2\times X_3$ as follows.

\[ \begin{matrix}
y& & A_i & B_i & C_i & D_i & E_i & F_i & G_i & H_i \\[2ex]
g(y) & & (1,1,i) & (2,1,i) & (1,2,i) & (2,2,i) & (1,3,i) & (2,3,i) & (1,4,i) & (2,4,i) 
\end{matrix} \]
\vspace{1ex}

Let $\xi:\{A_i,\ldots,H_i:i\in I\}\to X$ be $f^{-1}\circ g$. To show that $\xi$ is compatible with $(\theta,\theta')$ and $(\gamma,\gamma')$ we must show that the blocks of these relations are as described in Definition~\ref{special 8}. We do this for $\theta$, the others are similar. Note that 

\[\{g(a_i),g(c_i),g(e_i),g(h_i):i\in I\} = \{(1,1,i),(1,2,i),(1,3,i),(1,4,i):i\in I\}\]
\vspace{-1ex}

\noindent This is the set of all elements of $X_1,\times X_2\times X_3$ whose first coordinate is 1. Since $\theta =\ker(\pi_1\circ f)$, the image of this set under $f^{-1}$ is one of the blocks of $\theta$. Thus $\{a_i,c_i,e_i,g_i:i\in I\}$ is one block of $\theta$, and therefore $\{b_i,d_i,f_i,h_i:i\in I\}$ is the other. 
\end{proof}

There are many special 8-enumerations that are compatible with $(\theta,\theta')$ being under $(\gamma,\gamma')$, and we will use this to advantage in establishing properties of $\approx_2$. First we make use of a special 8-enumeration to describe some other 2-atoms $(\theta,\theta'')$ that have the same first component as $(\theta,\theta')$ and also lie beneath $(\gamma,\gamma')$.  These will satisfy $(\theta,\theta')\sim_2(\theta,\theta'')$. We do not need to find all such $(\theta,\theta'')$, just enough for use in our later proofs. 

\begin{prop}
Let $\xi$ be a special 8-enumeration that is compatible with $(\theta,\theta^1)$ being below $(\gamma,\gamma')$. Three other 2-relations $\theta^s$ for $s=2,3,4$ are described below using the enumeration $\xi$. For each $s,t=1,\ldots,4$ we have $(\theta,\theta^s)\sim_2(\theta,\theta^t)$. 
\label{bat}
\end{prop}

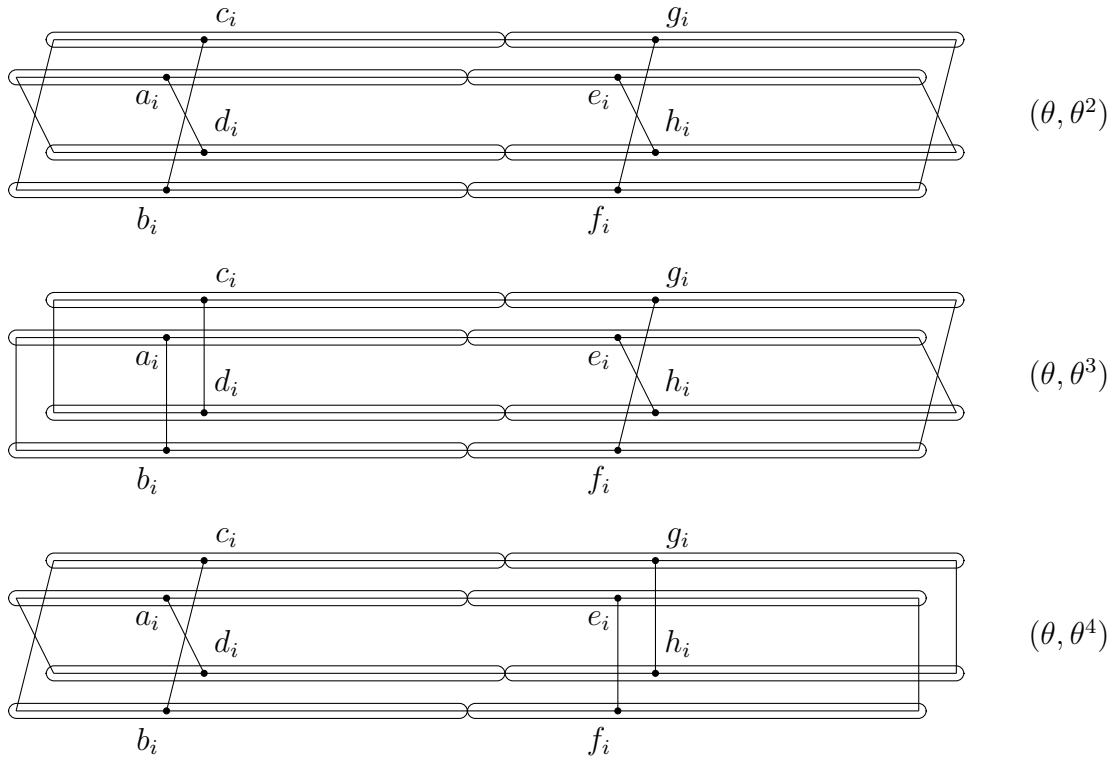
\begin{center}
\begin{tikzpicture}
\foreach \x in {0,2,8,12} {  \draw (\x,0) -- (\x+.5,2); \draw (\x+0.5,0.5) -- (\x,1.5); }
\foreach \x in {2,8} { \draw[fill] (\x, 0) circle [radius=0.04]; \draw[fill] (\x+0.5,0.5) circle [radius=0.04]; \draw[fill] (\x,1.5) circle [radius=0.04]; \draw[fill] (\x+0.5,2) circle [radius=0.04];}
\draw (0,0) -- (12,0); \draw (0,1.5) -- (12,1.5); \draw (0.5,0.5) -- (12.5,0.5); \draw (0.5,2) -- (12.5,2);
\draw[rounded corners=3pt] (-.1,-.1) rectangle ++(6.1,0.2); \draw[rounded corners=3pt] (6.0,-.1) rectangle ++(6.1,0.2);
\draw[rounded corners=3pt] (-.1,1.4) rectangle ++(6.1,0.2); \draw[rounded corners=3pt] (6.0,1.4) rectangle ++(6.1,0.2);
\draw[rounded corners=3pt] (.4,.4) rectangle ++(6.1,0.2); \draw[rounded corners=3pt] (6.5,.4) rectangle ++(6.1,0.2);
\draw[rounded corners=3pt] (.4,1.9) rectangle ++(6.1,0.2); \draw[rounded corners=3pt] (6.5,1.9) rectangle ++(6.1,0.2);
\node at (1.75,1.2) {$a_i$}; \node at (1.75,-.4) {$b_i$}; \node at (2.8,2.3) {$c_i$}; \node at (2.8,.9) {$d_i$};
\node at (7.75,1.2) {$e_i$}; \node at (7.75,-.4) {$f_i$}; \node at (8.8,2.3) {$g_i$}; \node at (8.8,.9) {$h_i$};
\node at (14,1) {$(\theta,\theta^2)$};
\end{tikzpicture}
\end{center}

\begin{center}
\begin{tikzpicture}
\foreach \x in {0,2} {  \draw (\x,0) -- (\x,1.5); \draw (\x+0.5,0.5) -- (\x+0.5,2); }
\foreach \x in {8,12} {  \draw (\x,0) -- (\x+.5,2); \draw (\x+0.5,0.5) -- (\x,1.5); }
\foreach \x in {2,8} { \draw[fill] (\x, 0) circle [radius=0.04]; \draw[fill] (\x+0.5,0.5) circle [radius=0.04]; \draw[fill] (\x,1.5) circle [radius=0.04]; \draw[fill] (\x+0.5,2) circle [radius=0.04];}
\draw (0,0) -- (12,0); \draw (0,1.5) -- (12,1.5); \draw (0.5,0.5) -- (12.5,0.5); \draw (0.5,2) -- (12.5,2);
\draw[rounded corners=3pt] (-.1,-.1) rectangle ++(6.1,0.2); \draw[rounded corners=3pt] (6.0,-.1) rectangle ++(6.1,0.2);
\draw[rounded corners=3pt] (-.1,1.4) rectangle ++(6.1,0.2); \draw[rounded corners=3pt] (6.0,1.4) rectangle ++(6.1,0.2);
\draw[rounded corners=3pt] (.4,.4) rectangle ++(6.1,0.2); \draw[rounded corners=3pt] (6.5,.4) rectangle ++(6.1,0.2);
\draw[rounded corners=3pt] (.4,1.9) rectangle ++(6.1,0.2); \draw[rounded corners=3pt] (6.5,1.9) rectangle ++(6.1,0.2);
\node at (1.75,1.2) {$a_i$}; \node at (1.75,-.4) {$b_i$}; \node at (2.8,2.3) {$c_i$}; \node at (2.8,.9) {$d_i$};
\node at (7.75,1.2) {$e_i$}; \node at (7.75,-.4) {$f_i$}; \node at (8.8,2.3) {$g_i$}; \node at (8.8,.9) {$h_i$};
\node at (14,1) {$(\theta,\theta^3)$};
\end{tikzpicture}
\end{center}

\begin{center}
\begin{tikzpicture}
\foreach \x in {8,12} {  \draw (\x,0) -- (\x,1.5); \draw (\x+0.5,0.5) -- (\x+0.5,2); }
\foreach \x in {0,2} {  \draw (\x,0) -- (\x+.5,2); \draw (\x+0.5,0.5) -- (\x,1.5); }
\foreach \x in {2,8} { \draw[fill] (\x, 0) circle [radius=0.04]; \draw[fill] (\x+0.5,0.5) circle [radius=0.04]; \draw[fill] (\x,1.5) circle [radius=0.04]; \draw[fill] (\x+0.5,2) circle [radius=0.04];}
\draw (0,0) -- (12,0); \draw (0,1.5) -- (12,1.5); \draw (0.5,0.5) -- (12.5,0.5); \draw (0.5,2) -- (12.5,2);
\draw[rounded corners=3pt] (-.1,-.1) rectangle ++(6.1,0.2); \draw[rounded corners=3pt] (6.0,-.1) rectangle ++(6.1,0.2);
\draw[rounded corners=3pt] (-.1,1.4) rectangle ++(6.1,0.2); \draw[rounded corners=3pt] (6.0,1.4) rectangle ++(6.1,0.2);
\draw[rounded corners=3pt] (.4,.4) rectangle ++(6.1,0.2); \draw[rounded corners=3pt] (6.5,.4) rectangle ++(6.1,0.2);
\draw[rounded corners=3pt] (.4,1.9) rectangle ++(6.1,0.2); \draw[rounded corners=3pt] (6.5,1.9) rectangle ++(6.1,0.2);
\node at (1.75,1.2) {$a_i$}; \node at (1.75,-.4) {$b_i$}; \node at (2.8,2.3) {$c_i$}; \node at (2.8,.9) {$d_i$};
\node at (7.75,1.2) {$e_i$}; \node at (7.75,-.4) {$f_i$}; \node at (8.8,2.3) {$g_i$}; \node at (8.8,.9) {$h_i$};
\node at (14,1) {$(\theta,\theta^4)$};
\end{tikzpicture}
\end{center}

\begin{proof}
Each equivalence class of each $\theta^s$ contains exactly one element of the bottom and one element of the top. So each $(\theta,\theta^s)$ is a factor pair that is a 2-atom. That $\gamma\subseteq\theta$ is already given. That $\theta^s\subseteq\gamma'$ follows since each block of $\theta^s$ is contained in one 8-element cube $\{a_i,\ldots,h_i\}$ that is a block of $\gamma'$. To see that $\gamma\circ\theta^s=\theta^s\circ\gamma$ we note that both are equivalence relations whose blocks are the union of two of the blocks of $\gamma$. Thus each $(\theta,\theta^s)\leq(\gamma,\gamma')$, and by the definition of $\sim_2$ we have $(\theta,\theta^s)\sim_2(\theta,\theta^t)$ for each $s,t=1,\ldots,4$.
\end{proof}

We now begin to use this result to investigate $\approx_2$.  

\begin{defn}
Let $\xi$ be a special 8-enumeration compatible with $(\theta,\theta')$ being below $(\gamma,\gamma')$ and let $I=J\cup K$ be a partition of $I$ with $|K|=|I|$. Define the $J$-twist of $\theta'$ with respect to $\xi$ to be the relation $\theta'_{J, \xi}$ whose blocks are as follows:
\begin{align*}
\{ a_k, b_k\}, \{ c_k, d_k\} & \mbox{ where }k\in K\\
\{ a_j, d_j\}, \{ c_j, b_j\}\, & \mbox{ where }j\in J\\
\{ e_i, f_i\},\, \{ g_i, h_i\}\, & \mbox{ where }i\in I
\end{align*}

\noindent Thus the $J$-twist agrees with $\theta'$ except the blocks $\{ a_j, b_j\},\{ c_j, d_j\}$ $(j\in J)$ of $\theta'$ are replaced with the blocks $\{ a_j, d_j\},\{ c_j, b_j\}$ $(j\in J)$ of its $J$-twist. 
\label{nail}
\end{defn}

We give an illustration of the $J$-twist of $\theta'$ with respect to $\xi$ below. It is more clear for our purposes to do so with a different style of picture, one that eliminates the role of $(\gamma,\gamma')$. The first diagram below shows $(\theta,\theta')$. We note that $\theta$ has two equivalence classes that are labelled $T$ and $B$ for top and bottom. The relation $\theta'$ has its equivalence classes described by Definition~\ref{special 8}. The pair $(\theta,\theta''_{J, \xi})$ is also seen to be a factor pair and it is shown in the second diagram. 
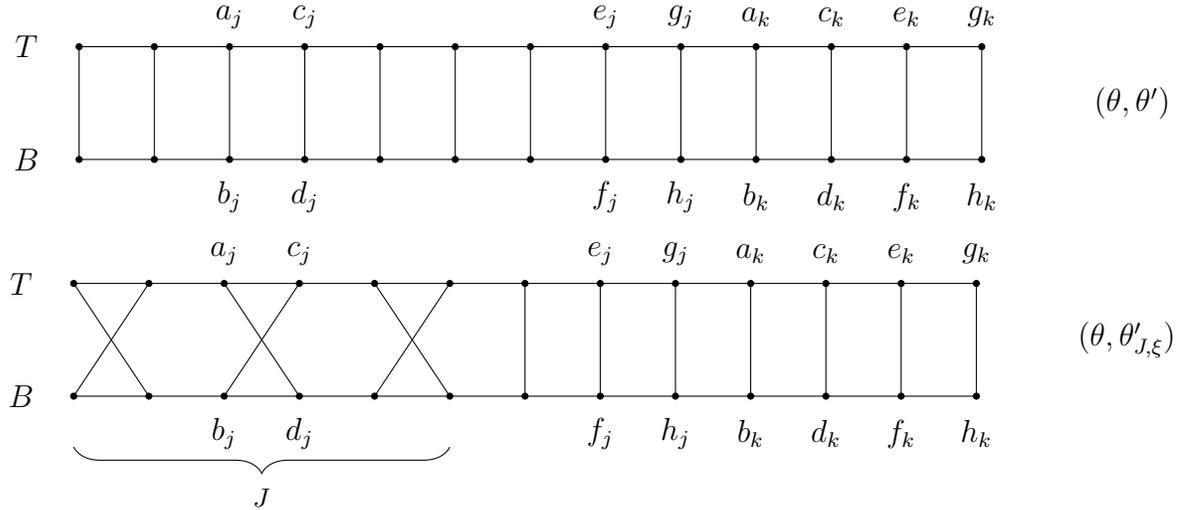
\vspace{1ex}

\begin{center}
\begin{tikzpicture}
\draw (0,0) -- (12,0); \draw (0,1.5) -- (12,1.5);
\foreach \x in {0,1,2,3,4,5,6,7,8,9,10,11,12} {  \draw (\x,0) -- (\x,1.5); }
\foreach \x in {0,1,2,3,4,5,6,7,8,9,10,11,12} { \draw[fill] (\x, 0) circle [radius=0.04]; \draw[fill] (\x ,1.5) circle [radius=0.04]; }
\node at (14,.75) {$(\theta,\theta')$};
\node at (-.7,1.5) {$T$}; \node at (-.7,0) {$B$};
\node at (2,1.9) {$a_j$}; \node at (2,-.5) {$b_j$}; \node at (3,1.9) {$c_j$}; \node at (3,-.5) {$d_j$};
\node at (7,1.9) {$e_j$}; \node at (7,-.5) {$f_j$}; \node at (8,1.9) {$g_j$}; \node at (8,-.5) {$h_j$};
\node at (9,1.9) {$a_k$}; \node at (9,-.5) {$b_k$}; \node at (10,1.9) {$c_k$}; \node at (10,-.5) {$d_k$};
\node at (11,1.9) {$e_k$}; \node at (11,-.5) {$f_k$}; \node at (12,1.9) {$g_k$}; \node at (12,-.5) {$h_k$};
\end{tikzpicture}
\end{center}

\begin{center}
\begin{tikzpicture}
\draw (0,0) -- (12,0); \draw (0,1.5) -- (12,1.5);
\foreach \x in {6,7,8,9,10,11,12} {  \draw (\x,0) -- (\x,1.5); }
\foreach \x in {0,1,2} {  \draw (2*\x,0) -- (2*\x+1,1.5); \draw (2*\x,1.5) -- (2*\x+1,0); }
\foreach \x in {0,1,2,3,4,5,6,7,8,9,10,11,12} { \draw[fill] (\x, 0) circle [radius=0.04]; \draw[fill] (\x ,1.5) circle [radius=0.04]; }
\node at (14,.75) {$(\theta,\theta'_{J, \xi})$};
\node at (-.7,1.5) {$T$}; \node at (-.7,0) {$B$};
\node at (2,1.9) {$a_j$}; \node at (2,-.5) {$b_j$}; \node at (3,1.9) {$c_j$}; \node at (3,-.5) {$d_j$};
\node at (7,1.9) {$e_j$}; \node at (7,-.5) {$f_j$}; \node at (8,1.9) {$g_j$}; \node at (8,-.5) {$h_j$};
\node at (9,1.9) {$a_k$}; \node at (9,-.5) {$b_k$}; \node at (10,1.9) {$c_k$}; \node at (10,-.5) {$d_k$};
\node at (11,1.9) {$e_k$}; \node at (11,-.5) {$f_k$}; \node at (12,1.9) {$g_k$}; \node at (12,-.5) {$h_k$};
\draw [decorate,decoration={brace,mirror,amplitude=10pt},xshift=0pt,yshift=-5pt](0,-.5) -- (5,-.5)node [black,midway,yshift=-19pt] {\footnotesize $J$};
\end{tikzpicture}
\end{center}

\begin{prop}
Let $\xi$ be a special 8-enumeration that is compatible with $(\theta,\theta')$ being below $(\gamma,\gamma')$ and let $I=J\cup K$ be a partition of $I$ with $|K|=|I|$. Then for the $J$-twist $\theta'_{J, \xi}$ we have 

\[(\theta,\theta')\approx_2(\theta,\theta'_{J, \xi})\]
\label{twist}
\end{prop}

\begin{proof}
We begin by modifying our picture of $(\theta,\theta')$ and $(\gamma,\gamma')$ of Figure~\ref{2-atom setup} by showing the elements whose indices lie in $J$ at the left of the ovals representing the $\gamma$ classes, and with the elements whose indices lie in $K$ at the right of these ovals. 
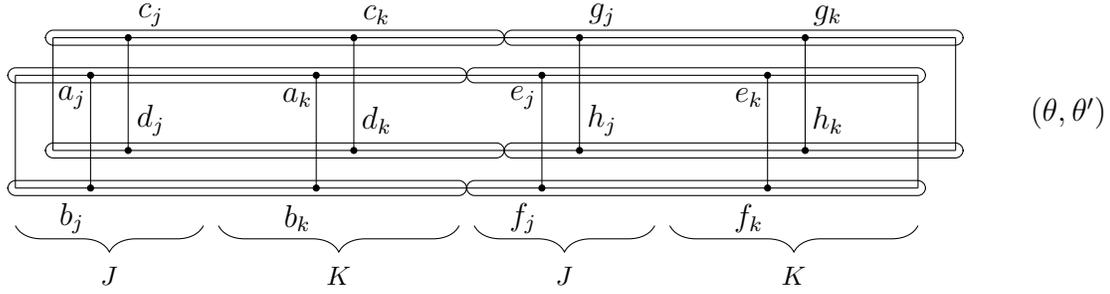
\vspace{2ex}

\begin{center}
\begin{tikzpicture}
\foreach \x in {0,1,4,7,10,12} {  \draw (\x,0) -- (\x,1.5); \draw (\x+0.5,0.5) -- (\x+0.5,2); }
\foreach \x in {1,4,7,10} { \draw[fill] (\x, 0) circle [radius=0.04]; \draw[fill] (\x+0.5,0.5) circle [radius=0.04]; \draw[fill] (\x,1.5) circle [radius=0.04]; \draw[fill] (\x+0.5,2) circle [radius=0.04];}
\draw (0,0) -- (12,0); \draw (0,1.5) -- (12,1.5); \draw (0.5,0.5) -- (12.5,0.5); \draw (0.5,2) -- (12.5,2);
\draw[rounded corners=3pt] (-.1,-.1) rectangle ++(6.1,0.2); \draw[rounded corners=3pt] (6.0,-.1) rectangle ++(6.1,0.2);
\draw[rounded corners=3pt] (-.1,1.4) rectangle ++(6.1,0.2); \draw[rounded corners=3pt] (6.0,1.4) rectangle ++(6.1,0.2);
\draw[rounded corners=3pt] (.4,.4) rectangle ++(6.1,0.2); \draw[rounded corners=3pt] (6.5,.4) rectangle ++(6.1,0.2);
\draw[rounded corners=3pt] (.4,1.9) rectangle ++(6.1,0.2); \draw[rounded corners=3pt] (6.5,1.9) rectangle ++(6.1,0.2);
\node at (0.75,1.2) {$a_j$}; \node at (0.75,-.4) {$b_j$}; \node at (1.8,2.3) {$c_j$}; \node at (1.8,.9) {$d_j$};
\node at (6.75,1.2) {$e_j$}; \node at (6.75,-.4) {$f_j$}; \node at (7.8,2.3) {$g_j$}; \node at (7.8,.9) {$h_j$};
\node at (3.75,1.2) {$a_k$}; \node at (3.75,-.4) {$b_k$}; \node at (4.8,2.3) {$c_k$}; \node at (4.8,.9) {$d_k$};
\node at (9.75,1.2) {$e_k$}; \node at (9.75,-.4) {$f_k$}; \node at (10.8,2.3) {$g_k$}; \node at (10.8,.9) {$h_k$};
\node at (14,1) {$(\theta,\theta')$};
\draw [decorate,decoration={brace,mirror,amplitude=10pt},xshift=0pt,yshift=0pt](0,-.5) -- (2.5,-.5)node [black,midway,yshift=-19pt] {\footnotesize $J$};
\draw [decorate,decoration={brace,mirror,amplitude=10pt},xshift=0pt,yshift=0pt](2.7,-.5) -- (5.9,-.5)node [black,midway,yshift=-19pt] {\footnotesize $K$};
\draw [decorate,decoration={brace,mirror,amplitude=10pt},xshift=0pt,yshift=0pt](6.1,-.5) -- (8.5,-.5)node [black,midway,yshift=-19pt] {\footnotesize $J$};
\draw [decorate,decoration={brace,mirror,amplitude=10pt},xshift=0pt,yshift=0pt](8.7,-.5) -- (12,-.5)node [black,midway,yshift=-19pt] {\footnotesize $K$};
\end{tikzpicture}
\end{center}
 
Apply Proposition~\ref{bat} to produce the 2-atom shown below that also lies beneath $(\gamma,\gamma')$ and satisfies $(\theta,\theta')\sim_2(\theta,\theta^2)$. 

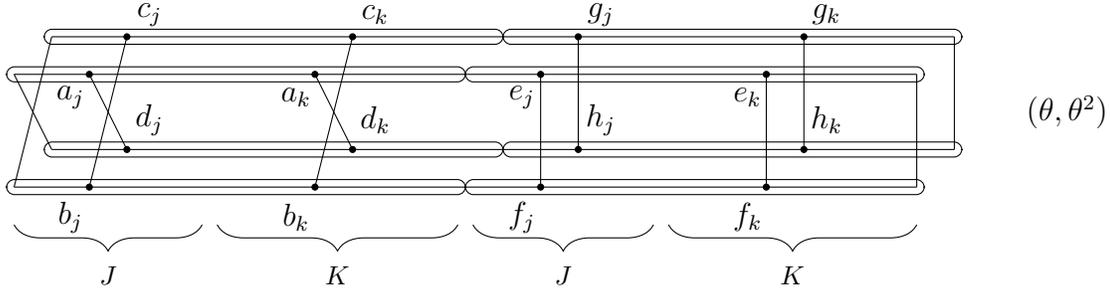
\begin{center}
\begin{tikzpicture}
\foreach \x in {7,10,12} {  \draw (\x,0) -- (\x,1.5); \draw (\x+0.5,0.5) -- (\x+0.5,2); }
\foreach \x in {0,1,4} {  \draw (\x,0) -- (\x+.5,2); \draw (\x+0.5,0.5) -- (\x,1.5); }
\foreach \x in {1,4,7,10} { \draw[fill] (\x, 0) circle [radius=0.04]; \draw[fill] (\x+0.5,0.5) circle [radius=0.04]; \draw[fill] (\x,1.5) circle [radius=0.04]; \draw[fill] (\x+0.5,2) circle [radius=0.04];}
\draw (0,0) -- (12,0); \draw (0,1.5) -- (12,1.5); \draw (0.5,0.5) -- (12.5,0.5); \draw (0.5,2) -- (12.5,2);
\draw[rounded corners=3pt] (-.1,-.1) rectangle ++(6.1,0.2); \draw[rounded corners=3pt] (6.0,-.1) rectangle ++(6.1,0.2);
\draw[rounded corners=3pt] (-.1,1.4) rectangle ++(6.1,0.2); \draw[rounded corners=3pt] (6.0,1.4) rectangle ++(6.1,0.2);
\draw[rounded corners=3pt] (.4,.4) rectangle ++(6.1,0.2); \draw[rounded corners=3pt] (6.5,.4) rectangle ++(6.1,0.2);
\draw[rounded corners=3pt] (.4,1.9) rectangle ++(6.1,0.2); \draw[rounded corners=3pt] (6.5,1.9) rectangle ++(6.1,0.2);
\node at (0.75,1.2) {$a_j$}; \node at (0.75,-.4) {$b_j$}; \node at (1.8,2.3) {$c_j$}; \node at (1.8,.9) {$d_j$};
\node at (6.75,1.2) {$e_j$}; \node at (6.75,-.4) {$f_j$}; \node at (7.8,2.3) {$g_j$}; \node at (7.8,.9) {$h_j$};
\node at (3.75,1.2) {$a_k$}; \node at (3.75,-.4) {$b_k$}; \node at (4.8,2.3) {$c_k$}; \node at (4.8,.9) {$d_k$};
\node at (9.75,1.2) {$e_k$}; \node at (9.75,-.4) {$f_k$}; \node at (10.8,2.3) {$g_k$}; \node at (10.8,.9) {$h_k$};
\node at (14,1) {$(\theta,\theta^2)$};
\draw [decorate,decoration={brace,mirror,amplitude=10pt},xshift=0pt,yshift=0pt](0,-.5) -- (2.5,-.5)node [black,midway,yshift=-19pt] {\footnotesize $J$};
\draw [decorate,decoration={brace,mirror,amplitude=10pt},xshift=0pt,yshift=0pt](2.7,-.5) -- (5.9,-.5)node [black,midway,yshift=-19pt] {\footnotesize $K$};
\draw [decorate,decoration={brace,mirror,amplitude=10pt},xshift=0pt,yshift=0pt](6.1,-.5) -- (8.5,-.5)node [black,midway,yshift=-19pt] {\footnotesize $J$};
\draw [decorate,decoration={brace,mirror,amplitude=10pt},xshift=0pt,yshift=0pt](8.7,-.5) -- (12,-.5)node [black,midway,yshift=-19pt] {\footnotesize $K$};
\end{tikzpicture}
\end{center}

Now we keep this same $(\theta,\theta^2)$, but choose a new $(\gamma_1,\gamma_1')$. In effect, we take the elements $\{a_j,b_j,c_j,d_j:j\in J\}$  from the leftmost ends of the 4 ovals at left, and move them to the leftmost ends of the 4 ovals at right. To draw this while retaining all the orientations of the lines indicating $\theta^2$ on the ovals at right we must interchange the positions of the $b_j$ and $d_j$. 

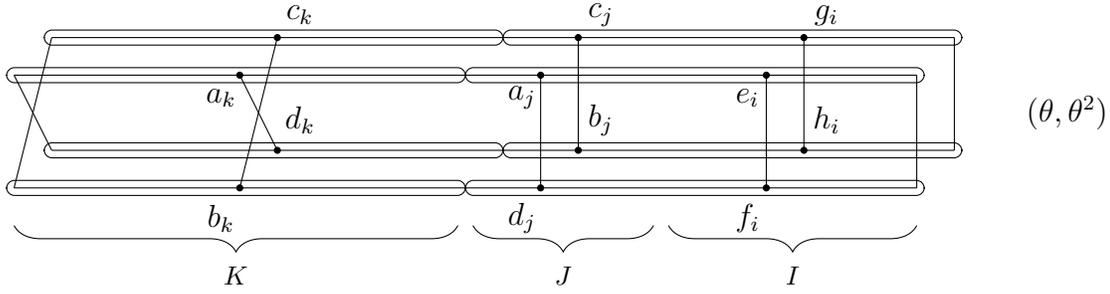
\begin{center}
\begin{tikzpicture}
\foreach \x in {7,10,12} {  \draw (\x,0) -- (\x,1.5); \draw (\x+0.5,0.5) -- (\x+0.5,2); }
\foreach \x in {0,3} {  \draw (\x,0) -- (\x+.5,2); \draw (\x+0.5,0.5) -- (\x,1.5); }
\foreach \x in {3,7,10} { \draw[fill] (\x, 0) circle [radius=0.04]; \draw[fill] (\x+0.5,0.5) circle [radius=0.04]; \draw[fill] (\x,1.5) circle [radius=0.04]; \draw[fill] (\x+0.5,2) circle [radius=0.04];}
\draw (0,0) -- (12,0); \draw (0,1.5) -- (12,1.5); \draw (0.5,0.5) -- (12.5,0.5); \draw (0.5,2) -- (12.5,2);
\draw[rounded corners=3pt] (-.1,-.1) rectangle ++(6.1,0.2); \draw[rounded corners=3pt] (6.0,-.1) rectangle ++(6.1,0.2);
\draw[rounded corners=3pt] (-.1,1.4) rectangle ++(6.1,0.2); \draw[rounded corners=3pt] (6.0,1.4) rectangle ++(6.1,0.2);
\draw[rounded corners=3pt] (.4,.4) rectangle ++(6.1,0.2); \draw[rounded corners=3pt] (6.5,.4) rectangle ++(6.1,0.2);
\draw[rounded corners=3pt] (.4,1.9) rectangle ++(6.1,0.2); \draw[rounded corners=3pt] (6.5,1.9) rectangle ++(6.1,0.2);

\node at (6.75,1.2) {$a_j$}; \node at (6.75,-.4) {$d_j$}; \node at (7.8,2.3) {$c_j$}; \node at (7.8,.9) {$b_j$};
\node at (2.75,1.2) {$a_k$}; \node at (2.75,-.4) {$b_k$}; \node at (3.8,2.3) {$c_k$}; \node at (3.8,.9) {$d_k$};
\node at (9.75,1.2) {$e_i$}; \node at (9.75,-.4) {$f_i$}; \node at (10.8,2.3) {$g_i$}; \node at (10.8,.9) {$h_i$};
\node at (14,1) {$(\theta,\theta^2)$};
\draw [decorate,decoration={brace,mirror,amplitude=10pt},xshift=0pt,yshift=0pt](0,-.5) -- (5.9,-.5)node [black,midway,yshift=-19pt] {\footnotesize $K$};
\draw [decorate,decoration={brace,mirror,amplitude=10pt},xshift=0pt,yshift=0pt](6.1,-.5) -- (8.5,-.5)node [black,midway,yshift=-19pt] {\footnotesize $J$};
\draw [decorate,decoration={brace,mirror,amplitude=10pt},xshift=0pt,yshift=0pt](8.7,-.5) -- (12,-.5)node [black,midway,yshift=-19pt] {\footnotesize $I$};
\end{tikzpicture}
\end{center}

To state this another way, let $J'=\{j':j\in J\}$ be a set in bijective correspondence with $J$ and disjoint from $I$. Since $K$ has the same cardinality as $I$ there is a bijection $\varphi:K\to J'\cup I$. We now define a new special 8-enumeration $\xi':\{A_k,\ldots,H_k:k\in K\}\to X$. Here, we recall our convention that $\xi(A_i)=a_i,\ldots,\xi(H_i)=h_i$ and $\xi'(A_k)=a_k',\ldots,\xi'(H_k)=h_k'$. We define $\xi'$ by setting $a_k'=a_k$, $b_k'=b_k$, $c_k'=c_k$ and $d_k'=d_k$ for all $k\in K$. In effect, this defines the four ovals at the left of the diagram above for $(\theta,\theta^2)$. We then set 

\[ e_k'=\begin{cases}a_j&\mbox{if }\varphi(k)=j'\\ e_i&\mbox{if }\varphi(k)=i\end{cases}\quad\quad\quad
   f_k'=\begin{cases}d_j&\mbox{if }\varphi(k)=j'\\ f_i&\mbox{if }\varphi(k)=i\end{cases} \]
   \vspace{-1ex}
  
\[ g_k'=\begin{cases}c_j&\mbox{if }\varphi(k)=j'\\ g_i&\mbox{if }\varphi(k)=i\end{cases}\quad\quad\quad
   h_k'=\begin{cases}b_j&\mbox{if }\varphi(k)=j'\\ h_i&\mbox{if }\varphi(k)=i\end{cases} \]
   \vspace{1ex}

\noindent These fill out the ovals on the right side of the diagram above for $(\theta,\theta^2)$. We note that $\xi'$ is not compatible with $(\theta,\theta^2)$ being below $(\gamma,\gamma')$ since the blocks of $\gamma$ and $\gamma'$ are not what Definition~\ref{special 8} prescribes for them using $\xi'$. Define a new pair $(\gamma_1,\gamma_1')$ as indicated below. One can then check that $\xi'$ is compatible with $(\theta,\theta^2)$ being below $(\gamma_1,\gamma_1')$. 
\vspace{-2ex}

\begin{align*}
\gamma_1&\mbox{ has blocks } \{a'_k:k\in K\}, \ldots, \{h_k':k\in K\}\\
\gamma_1' &\mbox{ has blocks } \{a_k',\ldots ,h_k'\} \mbox{ for } k\in K
\end{align*}
\vspace{-2ex}

Now apply Proposition~\ref{bat} to $(\theta,\theta^2)$ and $(\gamma_1,\gamma_1')$ using the enumeration $\xi'$ to produce $(\theta,\theta^3)$ shown below with $(\theta,\theta^2)\sim_2(\theta,\theta^3)$.

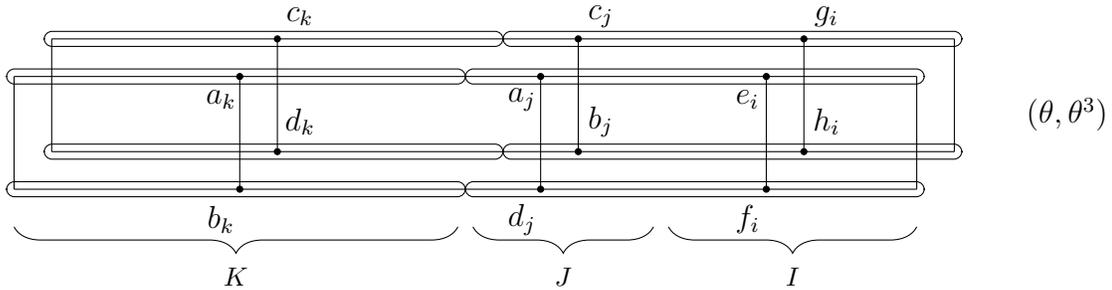
\begin{center}
\begin{tikzpicture}
\foreach \x in {0,3,7,10,12} {  \draw (\x,0) -- (\x,1.5); \draw (\x+0.5,0.5) -- (\x+0.5,2); }
\foreach \x in {3,7,10} { \draw[fill] (\x, 0) circle [radius=0.04]; \draw[fill] (\x+0.5,0.5) circle [radius=0.04]; \draw[fill] (\x,1.5) circle [radius=0.04]; \draw[fill] (\x+0.5,2) circle [radius=0.04];}
\draw (0,0) -- (12,0); \draw (0,1.5) -- (12,1.5); \draw (0.5,0.5) -- (12.5,0.5); \draw (0.5,2) -- (12.5,2);
\draw[rounded corners=3pt] (-.1,-.1) rectangle ++(6.1,0.2); \draw[rounded corners=3pt] (6.0,-.1) rectangle ++(6.1,0.2);
\draw[rounded corners=3pt] (-.1,1.4) rectangle ++(6.1,0.2); \draw[rounded corners=3pt] (6.0,1.4) rectangle ++(6.1,0.2);
\draw[rounded corners=3pt] (.4,.4) rectangle ++(6.1,0.2); \draw[rounded corners=3pt] (6.5,.4) rectangle ++(6.1,0.2);
\draw[rounded corners=3pt] (.4,1.9) rectangle ++(6.1,0.2); \draw[rounded corners=3pt] (6.5,1.9) rectangle ++(6.1,0.2);
\node at (6.75,1.2) {$a_j$}; \node at (6.75,-.4) {$d_j$}; \node at (7.8,2.3) {$c_j$}; \node at (7.8,.9) {$b_j$};
\node at (2.75,1.2) {$a_k$}; \node at (2.75,-.4) {$b_k$}; \node at (3.8,2.3) {$c_k$}; \node at (3.8,.9) {$d_k$};
\node at (9.75,1.2) {$e_i$}; \node at (9.75,-.4) {$f_i$}; \node at (10.8,2.3) {$g_i$}; \node at (10.8,.9) {$h_i$};
\node at (14,1) {$(\theta,\theta^3)$};
\draw [decorate,decoration={brace,mirror,amplitude=10pt},xshift=0pt,yshift=0pt](0,-.5) -- (5.9,-.5)node [black,midway,yshift=-19pt] {\footnotesize $K$};
\draw [decorate,decoration={brace,mirror,amplitude=10pt},xshift=0pt,yshift=0pt](6.1,-.5) -- (8.5,-.5)node [black,midway,yshift=-19pt] {\footnotesize $J$};
\draw [decorate,decoration={brace,mirror,amplitude=10pt},xshift=0pt,yshift=0pt](8.7,-.5) -- (12,-.5)node [black,midway,yshift=-19pt] {\footnotesize $I$};
\end{tikzpicture}
\end{center}
\vspace{1ex}

Comparing with the diagram after Definition~\ref{nail}, or directly with Definition~\ref{nail}, one sees that $\theta^3$ is exactly the $J$-twist $\theta'_{J, \xi}$ of $\theta'$ with resect to $\xi$. Since $$(\theta,\theta')\sim_2(\theta,\theta^2)\sim_2(\theta,\theta^3)=(\theta,\theta'_{J, \xi})$$\vspace{-2ex}

\noindent then $(\theta,\theta')\approx_2(\theta,\theta'_{J, \xi})$.
\end{proof}

\begin{rmk}{\em 
In this proof the switch from one enumeration $\xi$ to another $\xi'$ allows us to move from two 2-atoms $(\theta,\theta')$ and $(\theta,\theta^2)$ in the interval $[(\nabla,\Delta),(\gamma,\gamma')]$ where $|X/\gamma\,|=8$, to two 2-atoms $(\theta,\theta^2)$ and $(\theta,\theta^3)$ in a new interval $[(\nabla,\Delta),(\gamma_1,\gamma_1')]$ where $|X/\gamma_1'\,|=8$. In effect, we traverse from one 2-atom to some other via overlapping intervals of height 3. }
\end{rmk}

Proposition~\ref{twist} will be our key tool, but we must view it from a different perspective. Consider a regular equivalence relation $\theta$ with two blocks that we call $T$ and $B$. The relations $\theta'$ with $(\theta,\theta')$ a factor pair are exactly those 2-relations where each block of $\theta'$ has one element of $T$ and one element of $B$. In effect, these $\theta'$ are the various ways to pair elements of $T$ and $B$.

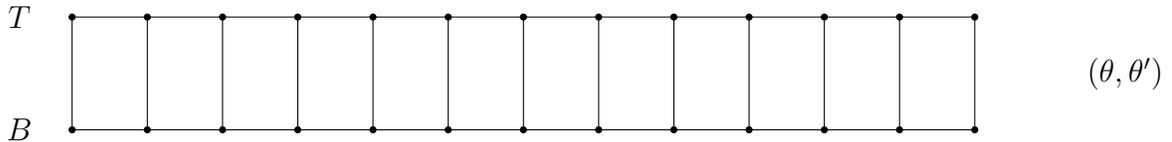
\begin{figure}[H]
\begin{center}
\begin{tikzpicture}
\draw (0,0) -- (12,0); \draw (0,1.5) -- (12,1.5);
\foreach \x in {0,1,2,3,4,5,6,7,8,9,10,11,12} {  \draw (\x,0) -- (\x,1.5); }
\foreach \x in {0,1,2,3,4,5,6,7,8,9,10,11,12} { \draw[fill] (\x, 0) circle [radius=0.04]; \draw[fill] (\x ,1.5) circle [radius=0.04]; }
\node at (14,.75) {$(\theta,\theta')$};
\node at (-.7,1.5) {$T$}; \node at (-.7,0) {$B$};
\end{tikzpicture}
\end{center}
\caption{A 2-atom}
\label{simple setup}
\end{figure}

\begin{defn}
Let $(\theta,\theta')$ be a 2-atom where the blocks of $\theta$ are $T$ and $B$. Then for each permutation $\sigma$ of $B$ define the 2-relation $\sigma(\theta')$ as follows. 
\[\sigma(\theta') \,\,\mbox{ has blocks $\,\{x,\sigma(y)\}\,$ where $\{x,y\}$ is a block of $\theta'$ and $y\in B$} \]
\end{defn}

The following is easily seen. 

\begin{prop}
Let $(\theta,\theta')$ be a 2-atom where the blocks of $\theta$ are $T$ and $B$. Then there is a bijective correspondence between $\{\phi:(\theta,\phi)\mbox{ is a 2-atom}\}$ and $\{\sigma(\theta'):\sigma\in\Perm(B)\}$.
\label{all}
\end{prop}

Our main task for 2-atoms is to show that if $(\theta,\theta')$ and $(\theta,\theta'')$ are 2-atoms with the same first coordinate, then $(\theta,\theta')\approx_2(\theta,\theta'')$. The key step is provided in the following. Here we recall that an involution on a set $B$ is a permutation $\delta$ of $B$ with $\delta^2=id$. So an involution is a permutation of $B$ that interchanges certain pairs of elements. 

\begin{prop}
Let $(\theta,\theta')$ be a 2-atom where the blocks of $\theta$ are $T$ and $B$. If $\delta$ is an involution of $B$ whose set of fixed points has the same cardinality as $B$, then $(\theta,\theta')\approx_2(\theta,\delta(\theta'))$. 
\label{big 2}
\end{prop}

\begin{proof}

Enumerate the 2-element orbits of $\delta$ with a set $J$, then label the elements of the $j^{th}$ orbit $b_j$ and $d_j$. So the 2-element orbits of $\delta$ are $\{b_j,d_j\}$ $(j\in J)$. For each $j\in J$, let $a_j$ and $c_j$ be the unique elements of $T$ with $\{a_j,b_j\}$ and $\{c_j,d_j\}$ blocks of $\theta'$. 

The elements of $B$ that do not belong to 2-element orbits are the fixed points of $\delta$ and the set of fixed points has the same cardinality as $B$. Partition the set of fixed points into two pieces, one of cardinality twice that of $J$ (we allow the possibility that $J$ is finite), and the other with the same cardinality as $B$. 

Put the elements of this first set in pairs and enumerate them as $f_j,h_j$ $(j\in J)$. Then for each $j\in J$ let $e_j$ and $g_j$ be the unique elements of $T$ with $\{e_j,f_j\}$ and $\{g_j,h_j\}$ blocks of $\theta'$. This second set of fixed points has the same infinite cardinality as $B$. Group this set into 4-tuples and enumerate them over a set $K$ as $b_k,d_k,f_k,h_k$ $(k\in K)$. For each $k\in K$ let $a_k,c_k,e_k,g_k$ be the unique elements of $T$ with $\{a_k,b_k\}$, $\{c_k,d_k\}$, $\{e_k,f_k\}$ and $\{g_k,h_k\}$ blocks of $\theta'$. 

Diagrams of $(\theta,\theta')$ and $(\theta,\delta(\theta'))$ are shown below. The diagram of $(\theta,\delta(\theta'))$ is obtained by interchanging the elements $b_j,d_j$ of the 2-element orbits of $\delta$. 

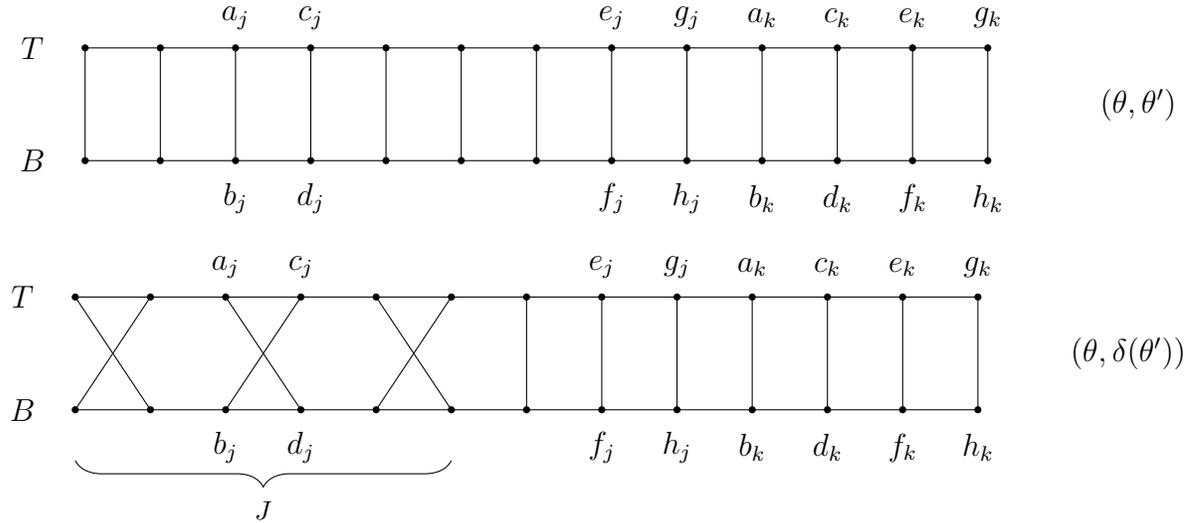
\begin{center}
\begin{tikzpicture}
\draw (0,0) -- (12,0); \draw (0,1.5) -- (12,1.5);
\foreach \x in {0,1,2,3,4,5,6,7,8,9,10,11,12} {  \draw (\x,0) -- (\x,1.5); }
\foreach \x in {0,1,2,3,4,5,6,7,8,9,10,11,12} { \draw[fill] (\x, 0) circle [radius=0.04]; \draw[fill] (\x ,1.5) circle [radius=0.04]; }
\node at (14,.75) {$(\theta,\theta')$};
\node at (-.7,1.5) {$T$}; \node at (-.7,0) {$B$};
\node at (2,1.9) {$a_j$}; \node at (2,-.5) {$b_j$}; \node at (3,1.9) {$c_j$}; \node at (3,-.5) {$d_j$};
\node at (7,1.9) {$e_j$}; \node at (7,-.5) {$f_j$}; \node at (8,1.9) {$g_j$}; \node at (8,-.5) {$h_j$};
\node at (9,1.9) {$a_k$}; \node at (9,-.5) {$b_k$}; \node at (10,1.9) {$c_k$}; \node at (10,-.5) {$d_k$};
\node at (11,1.9) {$e_k$}; \node at (11,-.5) {$f_k$}; \node at (12,1.9) {$g_k$}; \node at (12,-.5) {$h_k$};
\end{tikzpicture}
\end{center}

\begin{center}
\begin{tikzpicture}
\draw (0,0) -- (12,0); \draw (0,1.5) -- (12,1.5);
\foreach \x in {6,7,8,9,10,11,12} {  \draw (\x,0) -- (\x,1.5); }
\foreach \x in {0,1,2} {  \draw (2*\x,0) -- (2*\x+1,1.5); \draw (2*\x,1.5) -- (2*\x+1,0); }
\foreach \x in {0,1,2,3,4,5,6,7,8,9,10,11,12} { \draw[fill] (\x, 0) circle [radius=0.04]; \draw[fill] (\x ,1.5) circle [radius=0.04]; }
\node at (14,.75) {$(\theta,\delta(\theta'))$};
\node at (-.7,1.5) {$T$}; \node at (-.7,0) {$B$};
\node at (2,1.9) {$a_j$}; \node at (2,-.5) {$b_j$}; \node at (3,1.9) {$c_j$}; \node at (3,-.5) {$d_j$};
\node at (7,1.9) {$e_j$}; \node at (7,-.5) {$f_j$}; \node at (8,1.9) {$g_j$}; \node at (8,-.5) {$h_j$};
\node at (9,1.9) {$a_k$}; \node at (9,-.5) {$b_k$}; \node at (10,1.9) {$c_k$}; \node at (10,-.5) {$d_k$};
\node at (11,1.9) {$e_k$}; \node at (11,-.5) {$f_k$}; \node at (12,1.9) {$g_k$}; \node at (12,-.5) {$h_k$};
\draw [decorate,decoration={brace,mirror,amplitude=10pt},xshift=0pt,yshift=-5pt](0,-.5) -- (5,-.5)node [black,midway,yshift=-19pt] {\footnotesize $J$};
\end{tikzpicture}
\end{center}

Let $I=J\cup K$ and define $\xi:\{A_i,\ldots,H_i\}\to X$ by setting $\xi(A_i)=a_i,\ldots,\xi(H_i)=h_i$. Then define $(\gamma,\gamma')$ as follows 

\vspace{-2ex}

\begin{align*}
\gamma\,&\mbox{ has blocks } \{a_i:i\in I\}, \ldots, \{h_i:i\in I\}\\
\gamma' &\mbox{ has blocks } \{a_i,\ldots,h_i\} \mbox{ for } i\in I
\end{align*}
\vspace{-2ex}

\noindent Then it is easily seen that $\xi$ is a special 8-enumeration that is compatible with $(\theta,\theta')$ being below $(\gamma,\gamma')$. The situation is shown below. 

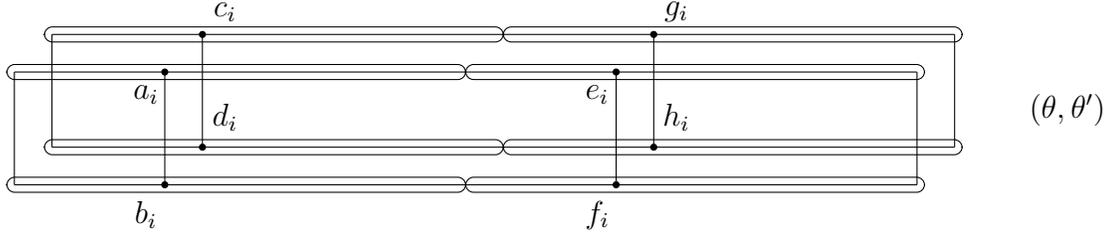
\begin{center}
\begin{tikzpicture}
\foreach \x in {0,2,8,12} {  \draw (\x,0) -- (\x,1.5); \draw (\x+0.5,0.5) -- (\x+0.5,2); }
\foreach \x in {2,8} { \draw[fill] (\x, 0) circle [radius=0.04]; \draw[fill] (\x+0.5,0.5) circle [radius=0.04]; \draw[fill] (\x,1.5) circle [radius=0.04]; \draw[fill] (\x+0.5,2) circle [radius=0.04];}
\draw (0,0) -- (12,0); \draw (0,1.5) -- (12,1.5); \draw (0.5,0.5) -- (12.5,0.5); \draw (0.5,2) -- (12.5,2);
\draw[rounded corners=3pt] (-.1,-.1) rectangle ++(6.1,0.2); \draw[rounded corners=3pt] (6.0,-.1) rectangle ++(6.1,0.2);
\draw[rounded corners=3pt] (-.1,1.4) rectangle ++(6.1,0.2); \draw[rounded corners=3pt] (6.0,1.4) rectangle ++(6.1,0.2);
\draw[rounded corners=3pt] (.4,.4) rectangle ++(6.1,0.2); \draw[rounded corners=3pt] (6.5,.4) rectangle ++(6.1,0.2);
\draw[rounded corners=3pt] (.4,1.9) rectangle ++(6.1,0.2); \draw[rounded corners=3pt] (6.5,1.9) rectangle ++(6.1,0.2);
\node at (1.75,1.2) {$a_i$}; \node at (1.75,-.4) {$b_i$}; \node at (2.8,2.3) {$c_i$}; \node at (2.8,.9) {$d_i$};
\node at (7.75,1.2) {$e_i$}; \node at (7.75,-.4) {$f_i$}; \node at (8.8,2.3) {$g_i$}; \node at (8.8,.9) {$h_i$};
\node at (14,1) {$(\theta,\theta')$};
\end{tikzpicture}
\end{center}

The partition $I=J\cup K$ has $|K|=|I|$. Using Definition~\ref{nail}, form the $J$-twist $\theta'_{J, \xi}$ of $\theta'$ with respect to $\xi$. From this definition, $\theta'_{J, \xi} = \sigma(\theta')$. Proposition~\ref{twist} gives $(\theta,\theta')\approx_2(\theta,\delta(\theta'))$. 
\end{proof}

Having established our result for certain special involutions having large sets of fixed points, we provide a simple observation that allows us to expand the utility of the result. 

\begin{prop}
Let $B$ be an infinite set. Then any involution $\tau$ of $B$ can be expressed

\[ \tau = \delta_2\circ\delta_1 \]
\vspace{-1ex}

\noindent  where $\delta_1$ and $\delta_2$ are involutions of $B$ where the cardinality of the set of fixed points of each is equal to the cardinality of $B$. 
\label{involutions}
\end{prop}

\begin{proof}
Enumerate the 2-element orbits of $\tau$ as $E_k$ $(k\in K)$ for some set $K$ of cardinality at most that of $B$. If $|K|<|B|$ then the set of fixed points of $\tau$ has cardinality equal to that of $B$ and we can express $\tau=\tau\circ id$. If $|K|=|B|$ we can partition $K$ as $K=K_1\cup K_2$ where $|K_1|=|K_2|=|B|$. Let $B_1=\bigcup_{K_1} E_k$ and $B_2=\bigcup_{K_2}E_k$. Then $B_1$ and $B_2$ are disjoint subsets of $B$ of cardinality equal to $B$. Their union need not be all of $B$, but this is not important. Since $B_1$ and $B_2$ are unions of orbits of $\tau$, we can define $\delta_1$ so that $\delta_1$ agrees with $\tau$ on $B_1$ and is equal to the identity otherwise, and $\delta_2$ so that $\delta_2$ agrees with $\tau$ on $B_2$ and is the identity otherwise. Then $\tau=\delta_2\circ\delta_1$ is a representation with the desired properties. 
\end{proof}

We now come to the desired result about the relation $\approx_2$. 

\begin{thm}
If $(\theta,\theta')$ and $(\theta,\theta'')$ are 2-atoms with the same first coordinate, then $$(\theta,\theta')\approx_2(\theta,\theta'')$$ \vspace{-2ex}
\label{tiger}
\end{thm}

\begin{proof}
Suppose that the blocks of $\theta$ are $T$ and $B$. Then by Proposition~\ref{all}, the set of all 2-atoms $(\theta,\theta'')$ that have the same first coordinate as $(\theta,\theta')$ is equal to the set of 2-atoms $(\theta,\sigma(\theta'))$ where $\sigma$ is a permutation of $B$. So it is enough to show that $S$ is equal to $\Perm(B)$ where
\[S=\{\sigma\in\Perm(B):(\theta,\theta')\approx_2(\theta,\sigma(\theta'))\}\]

Let $U$ be the set of all involutions $\delta$ of $B$ where the cardinality of the set of fixed points of $\delta$ is equal to the cardinality of $B$. We prove by induction on $n$ that if $\delta_1,\ldots,\delta_n\in U$ then the composite $\delta_n\circ\cdots\circ\delta_1$ is in $S$. The base case $n=1$ is given by Proposition~\ref{big 2}. For the inductive case, suppose $n>1$ and that $\delta_1,\ldots,\delta_n\in U$. Let $\sigma=\delta_{n-1}\circ\cdots\circ\delta_1$ and let $\pi=\delta_n\circ\cdots\circ\delta_1$. By the inductive hypothesis $\sigma\in S$, so $(\theta,\theta')\approx_2(\theta,\sigma(\theta'))$. Since $\theta$ remains unchanged and $\delta_n$ is a permutation of the block $B$ of $\theta$, we may apply Proposition~\ref{big 2} to $(\theta,\sigma(\theta'))$ to obtain that 

\[(\theta,\sigma(\theta'))\approx_2(\theta,\delta_n(\sigma(\theta')))\]
\vspace{-2ex}

\noindent Since $\delta_n(\sigma(\theta'))=\pi(\theta')$ and $\approx_2$ is by definition transitive, then $(\theta,\theta')\approx_2(\theta,\pi(\theta'))$. So $\pi$ is an element of $S$, concluding the inductive proof of the claim. 

Let $G$ be the subgroup of the permutation group $\Perm(B)$ that is generated by $U$, and let $N$ be the subgroup of $\Perm(B)$ that is generated by all of the involutions of $B$. Since each element of $U$ is an involution, it is its own inverse. So the subgroup $G$ generated by $U$ is obtained as the set of all composites $\delta_n\circ\cdots\circ\delta_1$ of members of $U$. Our claim then shows that $S$ contains~$G$. Proposition~\ref{involutions} shows that every involution is the composite of two members of $U$. So $S$ also contains the subgroup $N$ generated by all the involutions of $B$. 

For any involution $\tau$ and any $\lambda\in\Perm(B)$ we have $(\lambda\circ\tau\circ\lambda^{-1})^2=\lambda\circ\tau\circ\lambda^{-1}$ since $\tau$ is its own inverse. Therefore $\lambda\circ\tau\circ\lambda^{-1}$ is an involution of $B$. It follows from basic group theory that the subgroup $N$ of $\Perm(B)$ generated by the set of all involutions is a normal subgroup of $\Perm(B)$. The Baer-Schreier-Ulam Theorem \cite[p.~256]{Dixon-Mortimer} yields that the only normal subgroup of $\Perm(B)$ that has an element with no fixed points is the group $\Perm(B)$ itself. Thus $S=\Perm(B)$ as required. 
\end{proof}

This establishes the result we seek for the relation $\approx_2$. However, we need the corresponding result for $\approx_3$. We begin with our analog of a special 8-enumeration. Now we require an upper and lower case alphabet with 27 elements. We use $A,\ldots,Z,\Omega$ and $a,\ldots,z,\omega$. Again, we use the convention of letting $\xi(A_i)=a_i,\ldots,\xi(\Omega_i)=\omega_i$. 

\begin{defn}
A special 27-enumeration of $X$ consists of a set $I$ and a bijection 
\vspace{-1ex}

\[\xi:\{A_i,\ldots,Z_i,\Omega_i:i\in I\}\to X\]
\vspace{-2ex}

\noindent  We say $\xi$ is compatible with the factor pair $(\theta,\theta')$ being under $(\gamma,\gamma')$ if \vspace{-1ex}
\vspace{-1ex}

\begin{align*}
\theta\,\mbox{ has blocks }\, & \{a_i, c_i, e_i, g_i,i_i,j_i,k_i,x_i,y_i:i\in I\} \\
& \{ b_i, d_i, f_i, h_i, l_i, m_i,n_i, z_i,\omega_i:i\in I\}\\
& \{o_i,p_i,q_i,r_i,s_i,t_i,u_i,v_i,w_i:i\in I\}\\
\theta'\mbox{ has blocks }\,&\{ a_i, b_i,o_i\}, \{ c_i, d_i,r_i\}, \{ i_i, l_i,u_i\}\mbox{ for each }i\in I\\
& \{e_i,f_i,p_i\}, \{g_i,h_i,s_i\}, \{j_i,m_i,v_i\} \mbox{ for each }i\in I\\
& \{x_i,z_i,q_i\}, \{y_i,\omega_i,t_i\},\{k_i,n_i,w_i\} \mbox{ for each }i\in I\\
\gamma\,\mbox{ has blocks }\,&\{ a_i:i\in I\},\ldots,\{ \omega_i:i\in I\}\\
\gamma'\,\mbox{ has blocks }&\{ a_i,\ldots, \omega_i\}\,\mbox{ for each }\,i\in I
\end{align*}
\vspace{-2ex}

\noindent The situation is illustrated in Figure~\ref{3-atom setup}.
\label{special 27}
\end{defn}
\vspace{2ex}

\begin{figure}[h]
\begin{center}
\begin{tikzpicture}[scale=0.9, every node/.style={transform shape}]
\foreach \y in {0,1.5,3} { \draw [line width = .15mm] (0,\y) -- (1,1+\y); \draw [line width = .15mm]  (12,\y) -- (13,1+\y); }
\foreach \y in {0,1.5,3} { \foreach \h in {0,.5,1} \draw (\h,\y+\h) -- (\h+12,\y+\h);  }
\foreach \h in {0,.5,1} {\foreach \y in {0,1.5,3} {
	\draw[fill = white, opacity = .5, rounded corners=3pt] (-.1+\h,-.1+\y+\h) rectangle ++(4.05,0.2);
	\draw[fill = white, opacity = .5, rounded corners=3pt] (3.95+\h,-.1+\y+\h) rectangle ++(4.1,0.2);
	\draw[fill = white, opacity = .5, rounded corners=3pt] (8.05+\h,-.1+\y+\h) rectangle ++(4.1,0.2); } }
\foreach \x in {1.5,5.5,9.5} { \foreach \h in {0,.5,1} \draw (\x+\h,\h) -- (\x+\h,\h+3); }
\foreach \x in {0,4,8} { \foreach \y in {0,1.5,3} { \foreach \h in {0,.5,1} \draw[fill] (1.5+\x+\h,\y+\h) circle [radius=0.04]; }}

\draw[fill = gray, opacity = .35, rounded corners=3pt] (-.1,2.9) rectangle ++(4.05,0.2);
\draw[fill = gray, opacity = .35, rounded corners=3pt] (-.1,1.4) rectangle ++(4.05,0.2);
\draw[fill = gray, opacity = .35, rounded corners=3pt] (3.95,2.9) rectangle ++(4.1,0.2);
\draw[fill = gray, opacity = .35, rounded corners=3pt] (3.95,1.4) rectangle ++(4.1,0.2);
\draw[fill = gray, opacity = .35, rounded corners=3pt] (.4,3.4) rectangle ++(4.05,0.2);
\draw[fill = gray, opacity = .35, rounded corners=3pt] (.4,1.9) rectangle ++(4.05,0.2);
\draw[fill = gray, opacity = .35, rounded corners=3pt] (4.45,3.4) rectangle ++(4.1,0.2);
\draw[fill = gray, opacity = .35, rounded corners=3pt] (4.45,1.9) rectangle ++(4.1,0.2);
\node at (1.25,3) {$a_i$}; \node at (1.75,3.5) {$c_i$}; \node [darkgray!70] at (2.25,4) {$i_i$};
\node at (1.25,1.5) {$b_i$}; \node at (1.75,2) {$d_i$}; \node [darkgray!70] at (2.25,2.5) {$l_i$};
\node [darkgray!70] at (1.25,0) {$o_i$}; \node [darkgray!70] at (1.75,.5) {$r_i$}; \node [darkgray!70] at (2.25,1) {$u_i$};
\node at (5.25,3) {$e_i$}; \node at (5.75,3.5) {$g_i$}; \node [darkgray!70] at (6.25,4) {$j_i$};
\node at (5.25,1.5) {$f_i$}; \node at (5.75,2) {$h_i$}; \node [darkgray!70] at (6.25,2.5) {$m_i$};
\node [darkgray!70] at (5.25,0) {$p_i$}; \node [darkgray!70] at (5.75,.5) {$s_i$}; \node [darkgray!70] at (6.25,1) {$v_i$};
\node [darkgray!70] at (9.25,3) {$x_i$}; \node [darkgray!70] at (9.75,3.5) {$y_i$}; \node [darkgray!70] at (10.25,4) {$k_i$};
\node [darkgray!70] at (9.25,1.5) {$z_i$}; \node [darkgray!70] at (9.75,2) {$\omega_i$}; \node [darkgray!70] at (10.25,2.5) {$n_i$};
\node [darkgray!70] at (9.25,0) {$q_i$}; \node [darkgray!70] at (9.75,.5) {$t_i$}; \node [darkgray!70] at (10.25,1) {$w_i$};
\end{tikzpicture}
\end{center}
\caption{A special 27-enumeration $\xi$ compatible with $(\theta,\theta')$ being under $(\gamma,\gamma')$}
\label{3-atom setup}
\end{figure}
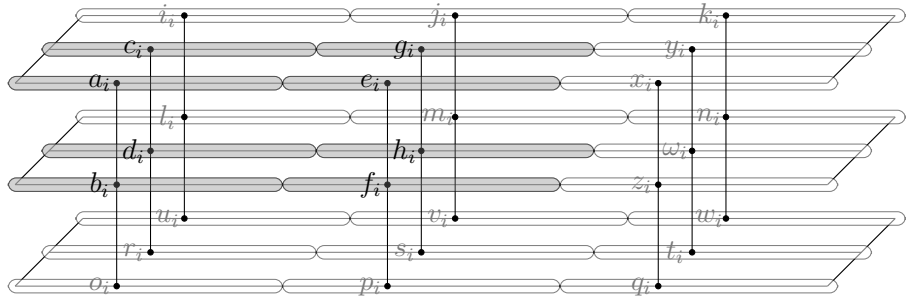
\vspace{2ex}

In this figure, $\theta$ has three blocks, the bottom level, middle level, and top level. The relation $\theta'$ has infinitely many blocks, each with three elements. They are the vertical lines. The 27 ovals are the blocks of $\gamma$, and the 27-element cubes are the blocks of $\gamma'$. The exact pattern to the labelling of the points is not of importance except to note that the shaded portion is labelled exactly as in the special 8-enumeration of Definition~\ref{special 8}. The proof of the following is essentially the same as that of Proposition~\ref{skunk} and we leave it to the reader. 

\begin{prop}
Let $(\theta,\theta')$ and $(\gamma,\gamma')$ be factor pairs. Then there is a special 27-relation $\xi$ that is compatible with $(\theta,\theta')$ being under $(\gamma,\gamma')$ iff the following conditions hold. 
\vspace{1ex}
\begin{enumerate}
\item $\theta$ has three blocks and $\theta'$ is a 3-relation
\item $\gamma$ has 27 blocks and $\gamma'$ is a 27-relation
\item $(\theta,\theta')\leq(\gamma,\gamma')$
\end{enumerate}
\label{skunk 3}
\end{prop}

Now an analog of Proposition~\ref{bat} for the 27-element setting. 

\begin{prop}
Let $\xi$ be a special 27-enumeration that is compatible with $(\theta,\theta^1)$ being below $(\gamma,\gamma')$. Three other 2-relations $\theta^s$ for $s=2,3,4$ are described below using the enumeration $\xi$. For each $s,t=1,\ldots,4$ we have $(\theta,\theta^s)\sim_3(\theta,\theta^t)$. 
\label{bat 3}
\end{prop}
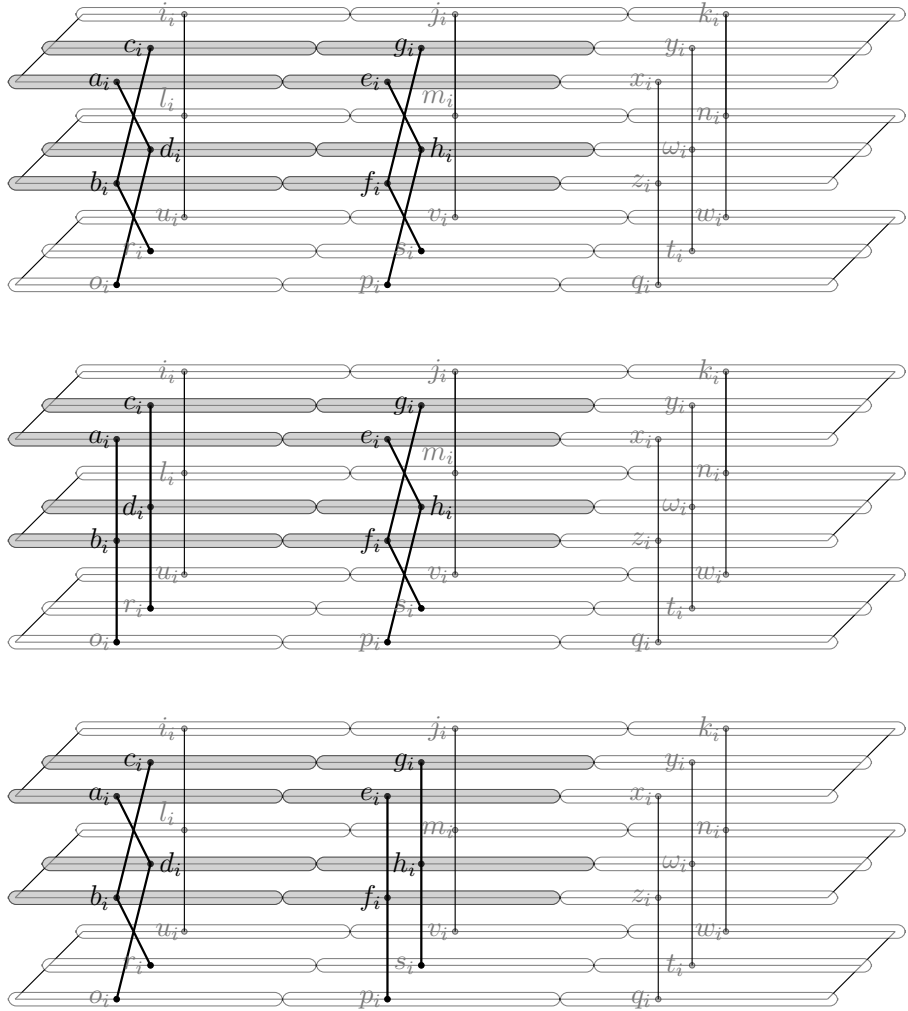
\vspace{2ex}

\begin{center}
\begin{tikzpicture}[scale=0.9, every node/.style={transform shape}]
\foreach \y in {0,1.5,3} { \draw [line width = .15mm] (0,\y) -- (1,1+\y); \draw [line width = .15mm]  (12,\y) -- (13,1+\y); }
\foreach \y in {0,1.5,3} { \foreach \h in {0,.5,1} \draw (\h,\y+\h) -- (\h+12,\y+\h);  }
\foreach \h in {0,.5,1} {\foreach \y in {0,1.5,3} {
	\draw[fill = white, opacity = .5, rounded corners=3pt] (-.1+\h,-.1+\y+\h) rectangle ++(4.05,0.2);
	\draw[fill = white, opacity = .5, rounded corners=3pt] (3.95+\h,-.1+\y+\h) rectangle ++(4.1,0.2);
	\draw[fill = white, opacity = .5, rounded corners=3pt] (8.05+\h,-.1+\y+\h) rectangle ++(4.1,0.2); } }
\foreach \x in {1.5,5.5,9.5} \draw [line width = .01mm] (\x+1,2.5) -- (\x+1,4); 
\foreach \x in {0,4,8} { \foreach \y in {0,1.5,3} { \foreach \h in {0,.5,1} \draw[fill = gray, opacity = .5] (1.5+\x+\h,\y+\h) circle [radius=0.04]; }}
\foreach \x in {0,4} { \foreach \y in {0,1.5,3} { \foreach \h in {0,.5} \draw[fill] (1.5+\x+\h,\y+\h) circle [radius=0.04]; }}
\draw[fill = gray, opacity = .35, rounded corners=3pt] (-.1,2.9) rectangle ++(4.05,0.2);
\draw[fill = gray, opacity = .35, rounded corners=3pt] (-.1,1.4) rectangle ++(4.05,0.2);
\draw[fill = gray, opacity = .35, rounded corners=3pt] (3.95,2.9) rectangle ++(4.1,0.2);
\draw[fill = gray, opacity = .35, rounded corners=3pt] (3.95,1.4) rectangle ++(4.1,0.2);
\draw[fill = gray, opacity = .35, rounded corners=3pt] (.4,3.4) rectangle ++(4.05,0.2);
\draw[fill = gray, opacity = .35, rounded corners=3pt] (.4,1.9) rectangle ++(4.05,0.2);
\draw[fill = gray, opacity = .35, rounded corners=3pt] (4.45,3.4) rectangle ++(4.1,0.2);
\draw[fill = gray, opacity = .35, rounded corners=3pt] (4.45,1.9) rectangle ++(4.1,0.2);
\foreach \x in {1.5,5.5,9.5} \draw [line width = .15mm] (\x+1,1) -- (\x+1,4); 
\foreach \h in {0,.5} \draw [line width = .15mm] (9.5+\h,\h) -- (9.5+\h,\h+3); 
\draw [line width = .3mm] (2,.5) -- (1.5,1.5) -- (2,3.5); 
\draw [line width = .3mm] (1.5,0) -- (2,2) -- (1.5,3); 
\draw [line width = .3mm] (6,.5) -- (5.5,1.5) -- (6,3.5);  
\draw [line width = .3mm] (5.5,0) -- (6,2) -- (5.5,3);
\node at (1.25,3) {$a_i$}; \node at (1.75,3.5) {$c_i$}; \node [darkgray!70] at (2.25,4) {$i_i$};
\node at (1.25,1.5) {$b_i$}; \node at (2.3,2.0) {$d_i$}; \node [darkgray!70] at (2.25,2.75) {$l_i$};
\node [darkgray!70] at (1.25,0) {$o_i$}; \node [darkgray!70] at (1.75,.5) {$r_i$}; \node [darkgray!70] at (2.25,1) {$u_i$};
\node at (5.25,3) {$e_i$}; \node at (5.75,3.5) {$g_i$}; \node [darkgray!70] at (6.25,4) {$j_i$};
\node at (5.25,1.5) {$f_i$}; \node at (6.3,2) {$h_i$}; \node [darkgray!70] at (6.25,2.75) {$m_i$};
\node [darkgray!70] at (5.25,0) {$p_i$}; \node [darkgray!70] at (5.75,.5) {$s_i$}; \node [darkgray!70] at (6.25,1) {$v_i$};
\node [darkgray!70] at (9.25,3) {$x_i$}; \node [darkgray!70] at (9.75,3.5) {$y_i$}; \node [darkgray!70] at (10.25,4) {$k_i$};
\node [darkgray!70] at (9.25,1.5) {$z_i$}; \node [darkgray!70] at (9.75,2) {$\omega_i$}; \node [darkgray!70] at (10.25,2.5) {$n_i$};
\node [darkgray!70] at (9.25,0) {$q_i$}; \node [darkgray!70] at (9.75,.5) {$t_i$}; \node [darkgray!70] at (10.25,1) {$w_i$};
\end{tikzpicture}
\end{center}
\vspace{2ex}

\begin{center}
\begin{tikzpicture}[scale=0.9, every node/.style={transform shape}]
\foreach \y in {0,1.5,3} { \draw [line width = .15mm] (0,\y) -- (1,1+\y); \draw [line width = .15mm]  (12,\y) -- (13,1+\y); }
\foreach \y in {0,1.5,3} { \foreach \h in {0,.5,1} \draw (\h,\y+\h) -- (\h+12,\y+\h);  }
\foreach \h in {0,.5,1} {\foreach \y in {0,1.5,3} {
	\draw[fill = white, opacity = .5, rounded corners=3pt] (-.1+\h,-.1+\y+\h) rectangle ++(4.05,0.2);
	\draw[fill = white, opacity = .5, rounded corners=3pt] (3.95+\h,-.1+\y+\h) rectangle ++(4.1,0.2);
	\draw[fill = white, opacity = .5, rounded corners=3pt] (8.05+\h,-.1+\y+\h) rectangle ++(4.1,0.2); } }
\foreach \x in {1.5,5.5,9.5} \draw [line width = .01mm] (\x+1,2.5) -- (\x+1,4); 
\foreach \x in {0,4,8} { \foreach \y in {0,1.5,3} { \foreach \h in {0,.5,1} \draw[fill = gray, opacity = .5] (1.5+\x+\h,\y+\h) circle [radius=0.04]; }}
\foreach \x in {0,4} { \foreach \y in {0,1.5,3} { \foreach \h in {0,.5} \draw[fill] (1.5+\x+\h,\y+\h) circle [radius=0.04]; }}
\draw[fill = gray, opacity = .35, rounded corners=3pt] (-.1,2.9) rectangle ++(4.05,0.2);
\draw[fill = gray, opacity = .35, rounded corners=3pt] (-.1,1.4) rectangle ++(4.05,0.2);
\draw[fill = gray, opacity = .35, rounded corners=3pt] (3.95,2.9) rectangle ++(4.1,0.2);
\draw[fill = gray, opacity = .35, rounded corners=3pt] (3.95,1.4) rectangle ++(4.1,0.2);
\draw[fill = gray, opacity = .35, rounded corners=3pt] (.4,3.4) rectangle ++(4.05,0.2);
\draw[fill = gray, opacity = .35, rounded corners=3pt] (.4,1.9) rectangle ++(4.05,0.2);
\draw[fill = gray, opacity = .35, rounded corners=3pt] (4.45,3.4) rectangle ++(4.1,0.2);
\draw[fill = gray, opacity = .35, rounded corners=3pt] (4.45,1.9) rectangle ++(4.1,0.2);
\foreach \x in {1.5,5.5,9.5} \draw [line width = .15mm] (\x+1,1) -- (\x+1,4); 
\foreach \h in {0,.5} \draw [line width = .15mm] (9.5+\h,\h) -- (9.5+\h,\h+3); 
\draw [line width = .3mm] (1.5,0) -- (1.5,1.5) -- (1.5,3); 
\draw [line width = .3mm] (2,.5) -- (2,2) -- (2,3.5); 
\draw [line width = .3mm] (6,.5) -- (5.5,1.5) -- (6,3.5);  
\draw [line width = .3mm] (5.5,0) -- (6,2) -- (5.5,3);
\node at (1.25,3) {$a_i$}; \node at (1.75,3.5) {$c_i$}; \node [darkgray!70] at (2.25,4) {$i_i$};
\node at (1.25,1.5) {$b_i$}; \node at (1.75,2.0) {$d_i$}; \node [darkgray!70] at (2.25,2.5) {$l_i$};
\node [darkgray!70] at (1.25,0) {$o_i$}; \node [darkgray!70] at (1.75,.5) {$r_i$}; \node [darkgray!70] at (2.25,1) {$u_i$};
\node at (5.25,3) {$e_i$}; \node at (5.75,3.5) {$g_i$}; \node [darkgray!70] at (6.25,4) {$j_i$};
\node at (5.25,1.5) {$f_i$}; \node at (6.3,2) {$h_i$}; \node [darkgray!70] at (6.25,2.75) {$m_i$};
\node [darkgray!70] at (5.25,0) {$p_i$}; \node [darkgray!70] at (5.75,.5) {$s_i$}; \node [darkgray!70] at (6.25,1) {$v_i$};
\node [darkgray!70] at (9.25,3) {$x_i$}; \node [darkgray!70] at (9.75,3.5) {$y_i$}; \node [darkgray!70] at (10.25,4) {$k_i$};
\node [darkgray!70] at (9.25,1.5) {$z_i$}; \node [darkgray!70] at (9.75,2) {$\omega_i$}; \node [darkgray!70] at (10.25,2.5) {$n_i$};
\node [darkgray!70] at (9.25,0) {$q_i$}; \node [darkgray!70] at (9.75,.5) {$t_i$}; \node [darkgray!70] at (10.25,1) {$w_i$};
\end{tikzpicture}
\end{center}
\vspace{2ex}

\begin{center}
\begin{tikzpicture}[scale=0.9, every node/.style={transform shape}]
\foreach \y in {0,1.5,3} { \draw [line width = .15mm] (0,\y) -- (1,1+\y); \draw [line width = .15mm]  (12,\y) -- (13,1+\y); }
\foreach \y in {0,1.5,3} { \foreach \h in {0,.5,1} \draw (\h,\y+\h) -- (\h+12,\y+\h);  }
\foreach \h in {0,.5,1} {\foreach \y in {0,1.5,3} {
	\draw[fill = white, opacity = .5, rounded corners=3pt] (-.1+\h,-.1+\y+\h) rectangle ++(4.05,0.2);
	\draw[fill = white, opacity = .5, rounded corners=3pt] (3.95+\h,-.1+\y+\h) rectangle ++(4.1,0.2);
	\draw[fill = white, opacity = .5, rounded corners=3pt] (8.05+\h,-.1+\y+\h) rectangle ++(4.1,0.2); } }
\foreach \x in {0,4,8} { \foreach \y in {0,1.5,3} { \foreach \h in {0,.5,1} \draw[fill = gray, opacity = .5] (1.5+\x+\h,\y+\h) circle [radius=0.04]; }}
\foreach \x in {0,4} { \foreach \y in {0,1.5,3} { \foreach \h in {0,.5} \draw[fill] (1.5+\x+\h,\y+\h) circle [radius=0.04]; }}
\draw[fill = gray, opacity = .35, rounded corners=3pt] (-.1,2.9) rectangle ++(4.05,0.2);
\draw[fill = gray, opacity = .35, rounded corners=3pt] (-.1,1.4) rectangle ++(4.05,0.2);
\draw[fill = gray, opacity = .35, rounded corners=3pt] (3.95,2.9) rectangle ++(4.1,0.2);
\draw[fill = gray, opacity = .35, rounded corners=3pt] (3.95,1.4) rectangle ++(4.1,0.2);
\draw[fill = gray, opacity = .35, rounded corners=3pt] (.4,3.4) rectangle ++(4.05,0.2);
\draw[fill = gray, opacity = .35, rounded corners=3pt] (.4,1.9) rectangle ++(4.05,0.2);
\draw[fill = gray, opacity = .35, rounded corners=3pt] (4.45,3.4) rectangle ++(4.1,0.2);
\draw[fill = gray, opacity = .35, rounded corners=3pt] (4.45,1.9) rectangle ++(4.1,0.2);
\foreach \x in {1.5,5.5,9.5} \draw [line width = .15mm] (\x+1,1) -- (\x+1,4); 
\foreach \h in {0,.5} \draw [line width = .15mm] (9.5+\h,\h) -- (9.5+\h,\h+3); 
\draw [line width = .3mm] (2,.5) -- (1.5,1.5) -- (2,3.5); 
\draw [line width = .3mm] (1.5,0) -- (2,2) -- (1.5,3); 
\draw [line width = .3mm] (5.5,0) -- (5.5,1.5) -- (5.5,3); 
\draw [line width = .3mm] (6,.5) -- (6,2) -- (6,3.5);
\node at (1.25,3) {$a_i$}; \node at (1.75,3.5) {$c_i$}; \node [darkgray!70] at (2.25,4) {$i_i$};
\node at (1.25,1.5) {$b_i$}; \node at (2.3,2.0) {$d_i$}; \node [darkgray!70] at (2.25,2.75) {$l_i$};
\node [darkgray!70] at (1.25,0) {$o_i$}; \node [darkgray!70] at (1.75,.5) {$r_i$}; \node [darkgray!70] at (2.25,1) {$u_i$};
\node at (5.25,3) {$e_i$}; \node at (5.75,3.5) {$g_i$}; \node [darkgray!70] at (6.25,4) {$j_i$};
\node at (5.25,1.5) {$f_i$}; \node at (5.75,2) {$h_i$}; \node [darkgray!70] at (6.25,2.5) {$m_i$};
\node [darkgray!70] at (5.25,0) {$p_i$}; \node [darkgray!70] at (5.75,.5) {$s_i$}; \node [darkgray!70] at (6.25,1) {$v_i$};
\node [darkgray!70] at (9.25,3) {$x_i$}; \node [darkgray!70] at (9.75,3.5) {$y_i$}; \node [darkgray!70] at (10.25,4) {$k_i$};
\node [darkgray!70] at (9.25,1.5) {$z_i$}; \node [darkgray!70] at (9.75,2) {$\omega_i$}; \node [darkgray!70] at (10.25,2.5) {$n_i$};
\node [darkgray!70] at (9.25,0) {$q_i$}; \node [darkgray!70] at (9.75,.5) {$t_i$}; \node [darkgray!70] at (10.25,1) {$w_i$};
\end{tikzpicture}
\end{center}
\vspace{3ex}

\begin{proof}
It is a matter of showing that $(\theta,\theta^s)\leq(\gamma,\gamma')$ for $s=1,\ldots,4$. That $\gamma\subseteq\theta$ follows since the blocks of $\gamma$ are the ovals, and these are contained in the three levels that are blocks of $\theta$. For $s=1,\ldots,4$ each, possibly bent, vertical line in the diagram for $(\theta,\theta^s)$ is a subset of $\{a_i,\ldots,\omega_i:i\in I\}$ for some $i\in I$. Therefore the possibly bent vertical lines that are blocks of $\theta^s$ are subsets of the 27-cubes that are blocks of $\gamma'$, so $\theta^s\subseteq\gamma'$. It remains to show that $\gamma\circ\theta^s=\theta^s\circ\gamma$. This follows since each composite is an equivalence relation whose blocks are the unions of three ovals connected by a possibly bent vertical line. 
\end{proof}

Next is an analog of Definition~\ref{nail}.

\begin{defn}
Let $\xi$ be a special 27-enumeration compatible with $(\theta,\theta')$ being below $(\gamma,\gamma')$ and let $I=J\cup K$ be a partition of $I$ with $|K|=|I|$. Define the $J$-twist of $\theta'$ with respect to $\xi$ to be the relation $\theta'_{J, \xi}$ whose blocks are as follows:

\begin{align*}
\mbox{ for }k\in K & \quad \{ a_k, b_k,o_k\}, \{ c_k, d_k,r_k\} \\
\mbox{ for }j\in J & \quad \{ a_j, d_j,o_j\}, \{ c_j, b_j,r_k\} \\
\mbox{ for }i\in I & \quad\{ e_i, f_i,p_i\},\, \{ g_i, h_i,s_i\},\, \{j_i,m_i,v_i\},\,\{ x_i,z_i,q_i\},\, \{y_i,\omega_i,t_i\},\, \{k_i,n_i,w_i\},\,\{i_i,l_i,u_i\} 
\end{align*}
\vspace{0ex}

\noindent Thus the $J$-twist agrees with $\theta'$ except the blocks $\{ a_j, b_j,o_j\},\{ c_j, d_j,r_j\}$ of $\theta$ are replaced with the blocks $\{ a_j, d_j,o_j\},\{ c_j, b_j,r_j\}$ of its $J$-twist. 
\label{nail 27}
\end{defn}

The following is an analog of Proposition~\ref{twist}. 

\begin{prop}
Let $\xi$ be a special 27-enumeration that is compatible with $(\theta,\theta')$ being below $(\gamma,\gamma')$ and let $I=J\cup K$ be a partition of $I$ with $|K|=|I|$. Then for the $J$-twist $\theta'_{J, \xi}$ we have 
\[(\theta,\theta')\approx_3(\theta,\theta'_{J, \xi})\]
\label{twist 27}
\end{prop}

\begin{proof}
Below is a diagram showing $\xi$ being compatible with $(\theta,\theta')$ being below $(\gamma,\gamma')$. We note that it is assumed that $\xi:\{A_i,\ldots,\Omega_i:i\in I\}\to X$ as described in Definition~\ref{special 27} and that the diagram is labelled as in Figure~\ref{3-atom setup}. 

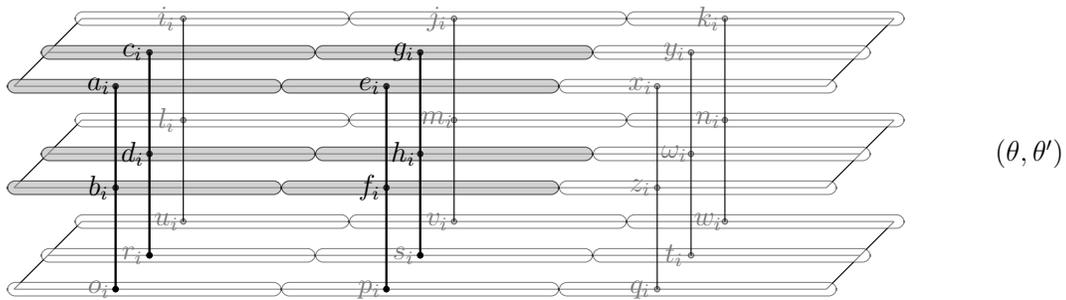
\begin{center}
\begin{tikzpicture}[scale=0.9, every node/.style={transform shape}]
\foreach \y in {0,1.5,3} { \draw [line width = .15mm] (0,\y) -- (1,1+\y); \draw [line width = .15mm]  (12,\y) -- (13,1+\y); }
\foreach \y in {0,1.5,3} { \foreach \h in {0,.5,1} \draw (\h,\y+\h) -- (\h+12,\y+\h);  }
\foreach \h in {0,.5,1} {\foreach \y in {0,1.5,3} {
	\draw[fill = white, opacity = .5, rounded corners=3pt] (-.1+\h,-.1+\y+\h) rectangle ++(4.05,0.2);
	\draw[fill = white, opacity = .5, rounded corners=3pt] (3.95+\h,-.1+\y+\h) rectangle ++(4.1,0.2);
	\draw[fill = white, opacity = .5, rounded corners=3pt] (8.05+\h,-.1+\y+\h) rectangle ++(4.1,0.2); } }
\foreach \x in {0,4,8} { \foreach \y in {0,1.5,3} { \foreach \h in {0,.5,1} \draw[fill = gray, opacity = .5] (1.5+\x+\h,\y+\h) circle [radius=0.04]; }}
\foreach \x in {0,4} { \foreach \y in {0,1.5,3} { \foreach \h in {0,.5} \draw[fill] (1.5+\x+\h,\y+\h) circle [radius=0.04]; }}
\draw[fill = gray, opacity = .35, rounded corners=3pt] (-.1,2.9) rectangle ++(4.05,0.2);
\draw[fill = gray, opacity = .35, rounded corners=3pt] (-.1,1.4) rectangle ++(4.05,0.2);
\draw[fill = gray, opacity = .35, rounded corners=3pt] (3.95,2.9) rectangle ++(4.1,0.2);
\draw[fill = gray, opacity = .35, rounded corners=3pt] (3.95,1.4) rectangle ++(4.1,0.2);
\draw[fill = gray, opacity = .35, rounded corners=3pt] (.4,3.4) rectangle ++(4.05,0.2);
\draw[fill = gray, opacity = .35, rounded corners=3pt] (.4,1.9) rectangle ++(4.05,0.2);
\draw[fill = gray, opacity = .35, rounded corners=3pt] (4.45,3.4) rectangle ++(4.1,0.2);
\draw[fill = gray, opacity = .35, rounded corners=3pt] (4.45,1.9) rectangle ++(4.1,0.2);
\foreach \x in {1.5,5.5,9.5} \draw [line width = .15mm] (\x+1,1) -- (\x+1,4); 
\foreach \h in {0,.5} \draw [line width = .15mm] (9.5+\h,\h) -- (9.5+\h,\h+3); 
\draw [line width = .3mm] (1.5,0) -- (1.5,3); 
\draw [line width = .3mm] (2,.5) -- (2,3.5); 
\draw [line width = .3mm] (5.5,0) -- (5.5,3); 
\draw [line width = .3mm] (6,.5) -- (6,3.5);
\node at (1.25,3) {$a_i$}; \node at (1.75,3.5) {$c_i$}; \node [darkgray!70] at (2.25,4) {$i_i$};
\node at (1.25,1.5) {$b_i$}; \node at (1.75,2) {$d_i$}; \node [darkgray!70] at (2.25,2.5) {$l_i$};
\node [darkgray!70] at (1.25,0) {$o_i$}; \node [darkgray!70] at (1.75,.5) {$r_i$}; \node [darkgray!70] at (2.25,1) {$u_i$};
\node at (5.25,3) {$e_i$}; \node at (5.75,3.5) {$g_i$}; \node [darkgray!70] at (6.25,4) {$j_i$};
\node at (5.25,1.5) {$f_i$}; \node at (5.75,2) {$h_i$}; \node [darkgray!70] at (6.25,2.5) {$m_i$};
\node [darkgray!70] at (5.25,0) {$p_i$}; \node [darkgray!70] at (5.75,.5) {$s_i$}; \node [darkgray!70] at (6.25,1) {$v_i$};
\node [darkgray!70] at (9.25,3) {$x_i$}; \node [darkgray!70] at (9.75,3.5) {$y_i$}; \node [darkgray!70] at (10.25,4) {$k_i$};
\node [darkgray!70] at (9.25,1.5) {$z_i$}; \node [darkgray!70] at (9.75,2) {$\omega_i$}; \node [darkgray!70] at (10.25,2.5) {$n_i$};
\node [darkgray!70] at (9.25,0) {$q_i$}; \node [darkgray!70] at (9.75,.5) {$t_i$}; \node [darkgray!70] at (10.25,1) {$w_i$};
\node at (15,2) {$(\theta,\theta')$};
\end{tikzpicture}
\end{center}
\vspace{2ex}

Use Proposition~\ref{bat 3} to produce $(\theta,\theta^2)$ shown below with $(\theta,\theta')\sim_3(\theta,\theta^2)$. The relation $\theta^2$ agrees with $\theta'$ except the blocks $\{a_i,b_i,o_i\}$, $\{c_i,d_i,r_i\}$ for $i\in I$ of $\theta'$ are replaced by the blocks $\{a_i,d_i,o_i\}$, $\{c_i,b_i,r_i\}$ for $i\in I$ of $\theta^2$. On this diagram we emphasized that each of the $\gamma$ blocks, the ovals, has its elements indexed indirectly by $I$. 
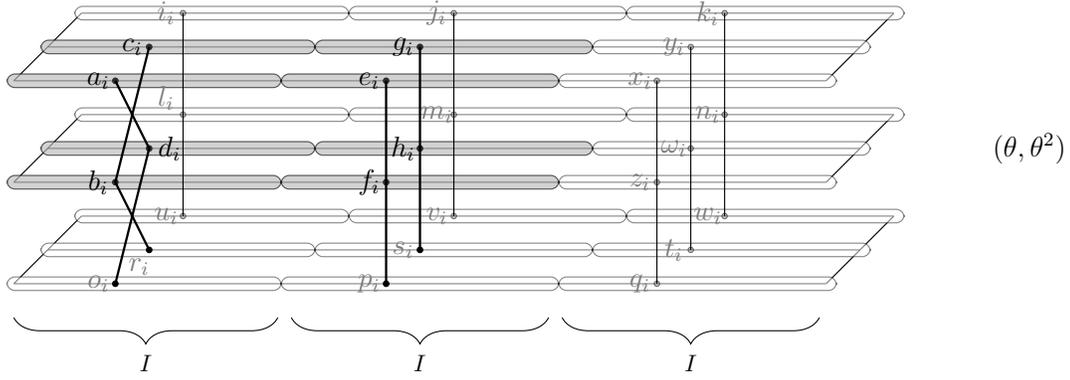
\vspace{2ex}

\begin{center}
\begin{tikzpicture}[scale=0.9, every node/.style={transform shape}]
\foreach \y in {0,1.5,3} { \draw [line width = .15mm] (0,\y) -- (1,1+\y); \draw [line width = .15mm]  (12,\y) -- (13,1+\y); }
\foreach \y in {0,1.5,3} { \foreach \h in {0,.5,1} \draw (\h,\y+\h) -- (\h+12,\y+\h);  }
\foreach \h in {0,.5,1} {\foreach \y in {0,1.5,3} {
	\draw[fill = white, opacity = .5, rounded corners=3pt] (-.1+\h,-.1+\y+\h) rectangle ++(4.05,0.2);
	\draw[fill = white, opacity = .5, rounded corners=3pt] (3.95+\h,-.1+\y+\h) rectangle ++(4.1,0.2);
	\draw[fill = white, opacity = .5, rounded corners=3pt] (8.05+\h,-.1+\y+\h) rectangle ++(4.1,0.2); } }
\foreach \x in {0,4,8} { \foreach \y in {0,1.5,3} { \foreach \h in {0,.5,1} \draw[fill = gray, opacity = .5] (1.5+\x+\h,\y+\h) circle [radius=0.04]; }}
\foreach \x in {0,4} { \foreach \y in {0,1.5,3} { \foreach \h in {0,.5} \draw[fill] (1.5+\x+\h,\y+\h) circle [radius=0.04]; }}
\draw[fill = gray, opacity = .35, rounded corners=3pt] (-.1,2.9) rectangle ++(4.05,0.2);
\draw[fill = gray, opacity = .35, rounded corners=3pt] (-.1,1.4) rectangle ++(4.05,0.2);
\draw[fill = gray, opacity = .35, rounded corners=3pt] (3.95,2.9) rectangle ++(4.1,0.2);
\draw[fill = gray, opacity = .35, rounded corners=3pt] (3.95,1.4) rectangle ++(4.1,0.2);
\draw[fill = gray, opacity = .35, rounded corners=3pt] (.4,3.4) rectangle ++(4.05,0.2);
\draw[fill = gray, opacity = .35, rounded corners=3pt] (.4,1.9) rectangle ++(4.05,0.2);
\draw[fill = gray, opacity = .35, rounded corners=3pt] (4.45,3.4) rectangle ++(4.1,0.2);
\draw[fill = gray, opacity = .35, rounded corners=3pt] (4.45,1.9) rectangle ++(4.1,0.2);
\foreach \x in {1.5,5.5,9.5} \draw [line width = .15mm] (\x+1,1) -- (\x+1,4); 
\foreach \h in {0,.5} \draw [line width = .15mm] (9.5+\h,\h) -- (9.5+\h,\h+3); 
\draw [line width = .3mm] (2,.5) -- (1.5,1.5) -- (2,3.5); 
\draw [line width = .3mm] (1.5,0) -- (2,2) -- (1.5,3); 
\draw [line width = .3mm] (5.5,0) -- (5.5,1.5) -- (5.5,3); 
\draw [line width = .3mm] (6,.5) -- (6,2) -- (6,3.5);
\node at (1.25,3) {$a_i$}; \node at (1.75,3.5) {$c_i$}; \node [darkgray!70] at (2.25,4) {$i_i$};
\node at (1.25,1.5) {$b_i$}; \node at (2.3,2.0) {$d_i$}; \node [darkgray!70] at (2.25,2.75) {$l_i$};
\node [darkgray!70] at (1.25,0) {$o_i$}; \node [darkgray!70] at (1.85,.25) {$r_i$}; \node [darkgray!70] at (2.25,1) {$u_i$};
\node at (5.25,3) {$e_i$}; \node at (5.75,3.5) {$g_i$}; \node [darkgray!70] at (6.25,4) {$j_i$};
\node at (5.25,1.5) {$f_i$}; \node at (5.75,2) {$h_i$}; \node [darkgray!70] at (6.25,2.5) {$m_i$};
\node [darkgray!70] at (5.25,0) {$p_i$}; \node [darkgray!70] at (5.75,.5) {$s_i$}; \node [darkgray!70] at (6.25,1) {$v_i$};
\node [darkgray!70] at (9.25,3) {$x_i$}; \node [darkgray!70] at (9.75,3.5) {$y_i$}; \node [darkgray!70] at (10.25,4) {$k_i$};
\node [darkgray!70] at (9.25,1.5) {$z_i$}; \node [darkgray!70] at (9.75,2) {$\omega_i$}; \node [darkgray!70] at (10.25,2.5) {$n_i$};
\node [darkgray!70] at (9.25,0) {$q_i$}; \node [darkgray!70] at (9.75,.5) {$t_i$}; \node [darkgray!70] at (10.25,1) {$w_i$};
\draw [decorate,decoration={brace,mirror,amplitude=10pt},xshift=0pt,yshift=0pt](0,-.5) -- (3.9,-.5)node [black,midway,yshift=-19pt] {\footnotesize $I$};
\draw [decorate,decoration={brace,mirror,amplitude=10pt},xshift=0pt,yshift=0pt](4.1,-.5) -- (7.9,-.5)node [black,midway,yshift=-19pt] {\footnotesize $I$};
\draw [decorate,decoration={brace,mirror,amplitude=10pt},xshift=0pt,yshift=0pt](8.1,-.5) -- (11.9,-.5)node [black,midway,yshift=-19pt] {\footnotesize $I$};
\node at (15,2) {$(\theta,\theta^2)$};
\end{tikzpicture}
\end{center}
\vspace{2ex}

We now make use of the partition $I=J\cup K$ to create a new factor pair $(\gamma_1,\gamma_1')$ and a new special 27-enumeration $\xi'$ that realizes $(\theta,\theta^2)$ being under $(\gamma_1,\gamma_1')$. In effect, we move the elements of the nine leftmost ovals whose coordinates have an index $j\in J$ into the corresponding nine middle ovals. This is shown in the diagram below, and a detailed description follows. An important note is that to retain a presentation of the $\theta^2$-classes of the middle ovals, the placements of $b_j$ and $d_j$ for $j\in J$ are interchanged. This is key to producing the $J$-twist. 
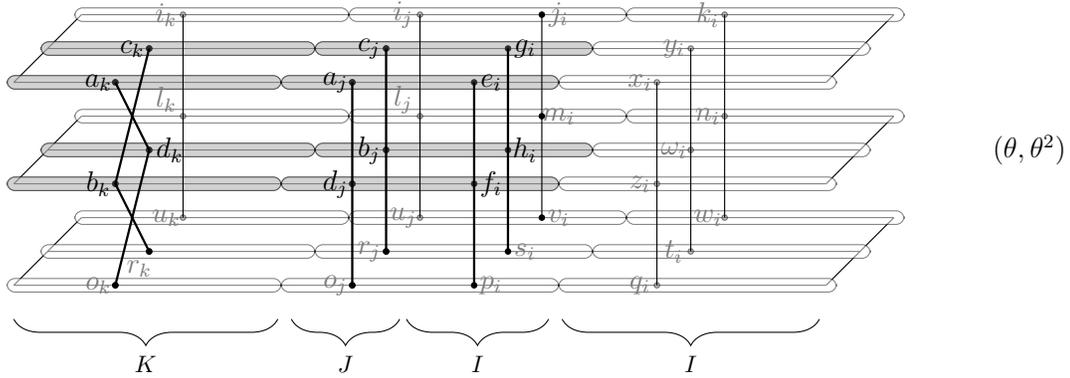
\vspace{2ex}

\begin{center}
\begin{tikzpicture}[scale=0.9, every node/.style={transform shape}]
\foreach \y in {0,1.5,3} { \draw [line width = .15mm] (0,\y) -- (1,1+\y); \draw [line width = .15mm]  (12,\y) -- (13,1+\y); }
\foreach \y in {0,1.5,3} { \foreach \h in {0,.5,1} \draw (\h,\y+\h) -- (\h+12,\y+\h);  }
\foreach \h in {0,.5,1} {\foreach \y in {0,1.5,3} {
	\draw[fill = white, opacity = .5, rounded corners=3pt] (-.1+\h,-.1+\y+\h) rectangle ++(4.05,0.2);
	\draw[fill = white, opacity = .5, rounded corners=3pt] (3.95+\h,-.1+\y+\h) rectangle ++(4.1,0.2);
	\draw[fill = white, opacity = .5, rounded corners=3pt] (8.05+\h,-.1+\y+\h) rectangle ++(4.1,0.2); } }
\foreach \y in {0,1.5,3} { \draw[fill] (6.8,\y) circle [radius=0.04]; } 
\foreach \y in {0,1.5,3} { \draw[fill] (7.3,.5+\y) circle [radius=0.04]; } 
\foreach \y in {0,1.5,3} { \draw[fill] (7.8,1+\y) circle [radius=0.04]; } 
\foreach \x in {0,3.5} { \foreach \y in {0,1.5,3} { \foreach \h in {0,.5} \draw[fill = gray, opacity = .5] (1.5+\x+\h,\y+\h) circle [radius=0.04]; }}
\foreach \x in {0,3.5,8} { \foreach \y in {0,1.5,3} { \foreach \h in {0,.5,1} \draw[fill = gray, opacity = .5] (1.5+\x+\h,\y+\h) circle [radius=0.04]; }} 
\foreach \x in {0,3.5} { \foreach \y in {0,1.5,3} { \foreach \h in {0,.5} \draw[fill] (1.5+\x+\h,\y+\h) circle [radius=0.04]; }}
\draw[fill = gray, opacity = .35, rounded corners=3pt] (-.1,2.9) rectangle ++(4.05,0.2);
\draw[fill = gray, opacity = .35, rounded corners=3pt] (-.1,1.4) rectangle ++(4.05,0.2);
\draw[fill = gray, opacity = .35, rounded corners=3pt] (3.95,2.9) rectangle ++(4.1,0.2);
\draw[fill = gray, opacity = .35, rounded corners=3pt] (3.95,1.4) rectangle ++(4.1,0.2);
\draw[fill = gray, opacity = .35, rounded corners=3pt] (.4,3.4) rectangle ++(4.05,0.2);
\draw[fill = gray, opacity = .35, rounded corners=3pt] (.4,1.9) rectangle ++(4.05,0.2);
\draw[fill = gray, opacity = .35, rounded corners=3pt] (4.45,3.4) rectangle ++(4.1,0.2);
\draw[fill = gray, opacity = .35, rounded corners=3pt] (4.45,1.9) rectangle ++(4.1,0.2);
\draw [line width = .3mm] (2,.5) -- (1.5,1.5) -- (2,3.5); 
\draw [line width = .3mm] (1.5,0) -- (2,2) -- (1.5,3); 
\foreach \x in {1.5,5,9.5} \draw [line width = .15mm] (\x+1,1) -- (\x+1,4); 
\foreach \h in {0,.5} \draw [line width = .15mm] (9.5+\h,\h) -- (9.5+\h,\h+3); 
\draw [line width = .3mm] (5,0) -- (5,1.5) -- (5,3); 
\draw [line width = .3mm] (5.5,.5) -- (5.5,2) -- (5.5,3.5);
\draw [line width = .3mm] (6.8,0) -- (6.8,1.5) -- (6.8,3); 
\draw [line width = .3mm] (7.3,.5) -- (7.3,2) -- (7.3,3.5); 
\draw [line width = .15mm] (7.8,1) -- (7.8,2.5) -- (7.8,4); 
\node at (1.25,3) {$a_k$}; \node at (1.75,3.5) {$c_k$}; \node [darkgray!70] at (2.25,4) {$i_k$};
\node at (1.25,1.5) {$b_k$}; \node at (2.3,2.0) {$d_k$}; \node [darkgray!70] at (2.25,2.75) {$l_k$};
\node [darkgray!70] at (1.25,0) {$o_k$}; \node [darkgray!70] at (1.85,.25) {$r_k$}; \node [darkgray!70] at (2.25,1) {$u_k$};
\node at (4.75,3) {$a_j$}; \node at (5.25,3.5) {$c_j$}; \node [darkgray!70] at (5.75,4) {$i_j$};
\node at (4.75,1.5) {$d_j$}; \node at (5.25,2.0) {$b_j$}; \node [darkgray!70] at (5.75,2.75) {$l_j$};
\node [darkgray!70] at (4.75,0) {$o_j$}; \node [darkgray!70] at (5.25,.5) {$r_j$}; \node [darkgray!70] at (5.75,1) {$u_j$};
\node at (7.05,3) {$e_i$}; \node at (7.55,3.5) {$g_i$}; \node [darkgray!70] at (8.05,4) {$j_i$};
\node at (7.05,1.5) {$f_i$}; \node at (7.55,2) {$h_i$}; \node [darkgray!70] at (8.05,2.5) {$m_i$};
\node [darkgray!70] at (7.05,0) {$p_i$}; \node [darkgray!70] at (7.55,.5) {$s_i$}; \node [darkgray!70] at (8.05,1) {$v_i$};
\node [darkgray!70] at (9.25,3) {$x_i$}; \node [darkgray!70] at (9.75,3.5) {$y_i$}; \node [darkgray!70] at (10.25,4) {$k_i$};
\node [darkgray!70] at (9.25,1.5) {$z_i$}; \node [darkgray!70] at (9.75,2) {$\omega_i$}; \node [darkgray!70] at (10.25,2.5) {$n_i$};
\node [darkgray!70] at (9.25,0) {$q_i$}; \node [darkgray!70] at (9.75,.5) {$t_i$}; \node [darkgray!70] at (10.25,1) {$w_i$};
\draw [decorate,decoration={brace,mirror,amplitude=10pt},xshift=0pt,yshift=0pt](0,-.5) -- (3.9,-.5)node [black,midway,yshift=-19pt] {\footnotesize $K$};
\draw [decorate,decoration={brace,mirror,amplitude=10pt},xshift=0pt,yshift=0pt](4.1,-.5) -- (5.7,-.5)node [black,midway,yshift=-19pt] {\footnotesize $J$};
\draw [decorate,decoration={brace,mirror,amplitude=10pt},xshift=0pt,yshift=0pt](5.8,-.5) -- (7.9,-.5)node [black,midway,yshift=-19pt] {\footnotesize $I$};
\draw [decorate,decoration={brace,mirror,amplitude=10pt},xshift=0pt,yshift=0pt](8.1,-.5) -- (11.9,-.5)node [black,midway,yshift=-19pt] {\footnotesize $I$};
\node at (15,2) {$(\theta,\theta^2)$};
\end{tikzpicture}
\end{center}
\vspace{2ex}

For an explicit description of this construction, we first let $J'=\{j':j\in J\}$ be a set in bijective correspondence with $J$ and disjoint from $I$. Since $K$ has the same cardinality as $I$ there are bijections $\varphi_1:K\to J'\cup I$ and $\varphi_2:K\to I$. We define a special 27-enumeration $\xi':\{A_k,\ldots,\Omega_k:k\in K\}\to X$. To do so, recall the convention that $\xi(A_i)=a_i,\ldots,\xi(\Omega_i)=\omega_i$ and that $\xi'(A_k)=a'_k,\ldots,\xi'(\Omega_k)=\omega'_k$. Then set for each $k\in K$ 

\[ a_k'=a_k, b_k'=b_k,o_k'=o_k,c_k'=c_k,d_k'=d_k,r_k'=r_k,i_k'=i_k,l_k'=l_k,u_k'=u_k\]
\vspace{-1ex}

\noindent So the nine ovals on the diagram above are labelled as expected. To describe the middle nine ovals of the diagram above we must describe the images under $\xi'$ of $E_k,F_k,P_k,G_k,H_k,S_k,J_k,M_k,V_k$. We do this only for $E_k$, the others follow the same pattern. 

\[ e_k' = \begin{cases} a_j & \mbox{ if } \varphi_1(k)=j' \\ e_i & \mbox{ if }\varphi_1(k)=i \end{cases} \]
\vspace{0ex}

\noindent The final nine ovals at right are given by the images under $\xi'$ of $X_k,Z_k,Q_k,Y_k,\Omega_k,T_k,K_k,N_k,W_k$. We set $x_k' = x_i$ where $\varphi_2(k)=i$, and similarly with the others. Then define $\gamma_1$ to have as blocks $\{a_k':k\in K\},\ldots,\{\omega_k':k\in K\}$ and $\gamma_1'$ to have as blocks $\{a_k',\ldots,\omega_k'\}$ for $k\in K$. It is then a simple but time consuming matter to verify that $\xi$ is a special 27-enumeration compatible with $(\theta,\theta^2)$ being below $(\gamma_1,\gamma_1')$ and given by the diagram above. 

Now use Proposition~\ref{bat 3} to produce from the special 27-enumeration $\xi'$ compatible with $(\theta,\theta^2)$ being under $(\gamma_1,\gamma_1')$ the 3-atom $(\theta,\theta^3)$ shown below with $(\theta,\theta^2)\sim_3(\theta,\theta^3)$. 
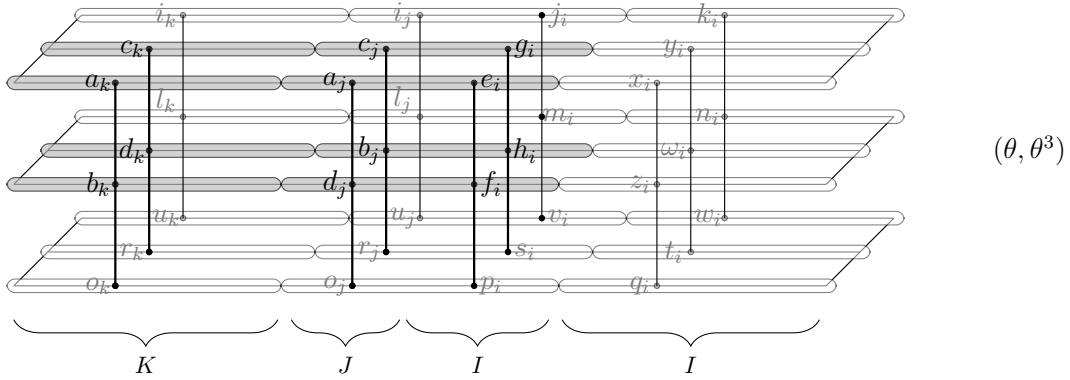
\vspace{2ex}

\begin{center}
\begin{tikzpicture}[scale=0.9, every node/.style={transform shape}]
\foreach \y in {0,1.5,3} { \draw [line width = .15mm] (0,\y) -- (1,1+\y); \draw [line width = .15mm]  (12,\y) -- (13,1+\y); }
\foreach \y in {0,1.5,3} { \foreach \h in {0,.5,1} \draw (\h,\y+\h) -- (\h+12,\y+\h);  }
\foreach \h in {0,.5,1} {\foreach \y in {0,1.5,3} {
	\draw[fill = white, opacity = .5, rounded corners=3pt] (-.1+\h,-.1+\y+\h) rectangle ++(4.05,0.2);
	\draw[fill = white, opacity = .5, rounded corners=3pt] (3.95+\h,-.1+\y+\h) rectangle ++(4.1,0.2);
	\draw[fill = white, opacity = .5, rounded corners=3pt] (8.05+\h,-.1+\y+\h) rectangle ++(4.1,0.2); } }
\foreach \y in {0,1.5,3} { \draw[fill] (6.8,\y) circle [radius=0.04]; } 
\foreach \y in {0,1.5,3} { \draw[fill] (7.3,.5+\y) circle [radius=0.04]; } 
\foreach \y in {0,1.5,3} { \draw[fill] (7.8,1+\y) circle [radius=0.04]; } 
\foreach \x in {0,3.5} { \foreach \y in {0,1.5,3} { \foreach \h in {0,.5} \draw[fill = gray, opacity = .5] (1.5+\x+\h,\y+\h) circle [radius=0.04]; }}%
\foreach \x in {0,3.5,8} { \foreach \y in {0,1.5,3} { \foreach \h in {0,.5,1} \draw[fill = gray, opacity = .5] (1.5+\x+\h,\y+\h) circle [radius=0.04]; }} %
\foreach \x in {0,3.5} { \foreach \y in {0,1.5,3} { \foreach \h in {0,.5} \draw[fill] (1.5+\x+\h,\y+\h) circle [radius=0.04]; }}%
\draw[fill = gray, opacity = .35, rounded corners=3pt] (-.1,2.9) rectangle ++(4.05,0.2);
\draw[fill = gray, opacity = .35, rounded corners=3pt] (-.1,1.4) rectangle ++(4.05,0.2);
\draw[fill = gray, opacity = .35, rounded corners=3pt] (3.95,2.9) rectangle ++(4.1,0.2);
\draw[fill = gray, opacity = .35, rounded corners=3pt] (3.95,1.4) rectangle ++(4.1,0.2);
\draw[fill = gray, opacity = .35, rounded corners=3pt] (.4,3.4) rectangle ++(4.05,0.2);
\draw[fill = gray, opacity = .35, rounded corners=3pt] (.4,1.9) rectangle ++(4.05,0.2);
\draw[fill = gray, opacity = .35, rounded corners=3pt] (4.45,3.4) rectangle ++(4.1,0.2);
\draw[fill = gray, opacity = .35, rounded corners=3pt] (4.45,1.9) rectangle ++(4.1,0.2);
\draw [line width = .3mm] (2,.5) -- (2,2) -- (2,3.5); 
\draw [line width = .3mm] (1.5,0) -- (1.5,1.5) -- (1.5,3); 
\foreach \x in {1.5,5,9.5} \draw [line width = .15mm] (\x+1,1) -- (\x+1,4); 
\foreach \h in {0,.5} \draw [line width = .15mm] (9.5+\h,\h) -- (9.5+\h,\h+3); 
\draw [line width = .3mm] (5,0) -- (5,1.5) -- (5,3); 
\draw [line width = .3mm] (5.5,.5) -- (5.5,2) -- (5.5,3.5);
\draw [line width = .3mm] (6.8,0) -- (6.8,1.5) -- (6.8,3); 
\draw [line width = .3mm] (7.3,.5) -- (7.3,2) -- (7.3,3.5); 
\draw [line width = .15mm] (7.8,1) -- (7.8,2.5) -- (7.8,4); 
\node at (1.25,3) {$a_k$}; \node at (1.75,3.5) {$c_k$}; \node [darkgray!70] at (2.25,4) {$i_k$};
\node at (1.25,1.5) {$b_k$}; \node at (1.75,2.0) {$d_k$}; \node [darkgray!70] at (2.25,2.75) {$l_k$};
\node [darkgray!70] at (1.25,0) {$o_k$}; \node [darkgray!70] at (1.75,.5) {$r_k$}; \node [darkgray!70] at (2.25,1) {$u_k$};
\node at (4.75,3) {$a_j$}; \node at (5.25,3.5) {$c_j$}; \node [darkgray!70] at (5.75,4) {$i_j$};
\node at (4.75,1.5) {$d_j$}; \node at (5.25,2.0) {$b_j$}; \node [darkgray!70] at (5.75,2.75) {$l_j$};
\node [darkgray!70] at (4.75,0) {$o_j$}; \node [darkgray!70] at (5.25,.5) {$r_j$}; \node [darkgray!70] at (5.75,1) {$u_j$};
\node at (7.05,3) {$e_i$}; \node at (7.55,3.5) {$g_i$}; \node [darkgray!70] at (8.05,4) {$j_i$};
\node at (7.05,1.5) {$f_i$}; \node at (7.55,2) {$h_i$}; \node [darkgray!70] at (8.05,2.5) {$m_i$};
\node [darkgray!70] at (7.05,0) {$p_i$}; \node [darkgray!70] at (7.55,.5) {$s_i$}; \node [darkgray!70] at (8.05,1) {$v_i$};
\node [darkgray!70] at (9.25,3) {$x_i$}; \node [darkgray!70] at (9.75,3.5) {$y_i$}; \node [darkgray!70] at (10.25,4) {$k_i$};
\node [darkgray!70] at (9.25,1.5) {$z_i$}; \node [darkgray!70] at (9.75,2) {$\omega_i$}; \node [darkgray!70] at (10.25,2.5) {$n_i$};
\node [darkgray!70] at (9.25,0) {$q_i$}; \node [darkgray!70] at (9.75,.5) {$t_i$}; \node [darkgray!70] at (10.25,1) {$w_i$};
\draw [decorate,decoration={brace,mirror,amplitude=10pt},xshift=0pt,yshift=0pt](0,-.5) -- (3.9,-.5)node [black,midway,yshift=-19pt] {\footnotesize $K$};
\draw [decorate,decoration={brace,mirror,amplitude=10pt},xshift=0pt,yshift=0pt](4.1,-.5) -- (5.7,-.5)node [black,midway,yshift=-19pt] {\footnotesize $J$};
\draw [decorate,decoration={brace,mirror,amplitude=10pt},xshift=0pt,yshift=0pt](5.8,-.5) -- (7.9,-.5)node [black,midway,yshift=-19pt] {\footnotesize $I$};
\draw [decorate,decoration={brace,mirror,amplitude=10pt},xshift=0pt,yshift=0pt](8.1,-.5) -- (11.9,-.5)node [black,midway,yshift=-19pt] {\footnotesize $I$};
\node at (15,2) {$(\theta,\theta^3)$};
\end{tikzpicture}
\end{center}
\vspace{2ex}

\noindent Then $\theta^3$ agrees with $\theta'$ except that the blocks $\{a_j,b_j,o_j\},\{c_j,d_j,r_j\}$ for $j\in J$ of $\theta'$ have been replaced by the blocks $\{a_j,d_j,o_j\}, \{c_j,b_j,r_j\}$ for $j\in J$ of $\theta^3$. So $\theta^3$ is the $J$-twist $\theta'_{J, \xi}$. Since $(\theta,\theta')\sim_3(\theta,\theta^2)$ and $(\theta,\theta^2)\sim_3(\theta,\theta^3)$, we have $(\theta,\theta')\approx_3(\theta,\theta'_{J, \xi})$ as required. 
\end{proof}

We continue as in the 2-atom setting. In the following, we assume that $\theta$ is a fixed regular relation that has three blocks that we call $T$, $M$, and $B$ for top, middle and bottom. The relations $\theta'$ with $(\theta,\theta')$ a factor pair are exactly those 3-relations where each block of $\theta'$ has one element of $T$, one element of $M$ and one element of $B$. A typical situation is shown below. 

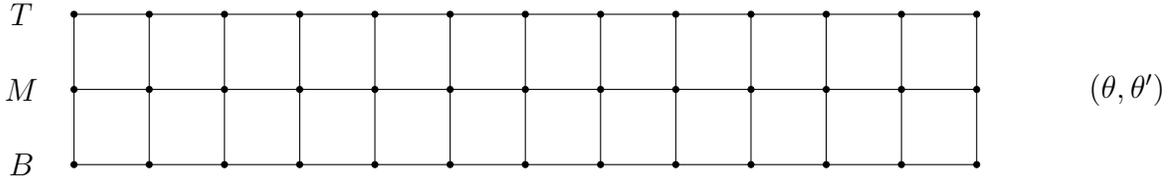
\begin{center}
\begin{tikzpicture}
\draw (0,0) -- (12,0); \draw (0,1) -- (12,1); \draw (0,2) -- (12,2);
\foreach \x in {0,1,2,3,4,5,6,7,8,9,10,11,12} { \draw (\x,0) -- (\x,2); }
\foreach \x in {0,1,2,3,4,5,6,7,8,9,10,11,12} { \draw[fill] (\x, 0) circle [radius=0.04]; \draw[fill] (\x ,1) circle [radius=0.04]; \draw[fill] (\x ,2) circle [radius=0.04];}
\node at (14,1) {$(\theta,\theta')$};
\node at (-.7,2) {$T$}; \node at (-.7,0) {$B$}; \node at (-.7,1) {$M$};
\end{tikzpicture}
\end{center}
\vspace{1ex}

\begin{defn}
Let $(\theta,\theta')$ be a 3-atom where the blocks of $\theta$ are $T$, $M$, and $B$. Then for each permutation $\sigma$ of $M$ and each permutation $\rho$ of $B$ define the 3-relations $\sigma(\theta')$ and $\rho(\theta')$ as follows. 
\begin{align*}
\sigma(\theta') \,\,\mbox{ has blocks }&\mbox{ $\,\{x,\sigma(y),z\}\,$ where $\{x,y,z\}$ is a block of $\theta'$ and $y\in M$} \\
\rho(\theta') \,\,\mbox{ has blocks }&\mbox{ $\,\{x,y,\rho(z)\}\,$ where $\{x,y,z\}$ is a block of $\theta'$ and $z\in B$}
\end{align*}
\vspace{-1ex}
\end{defn}

Next we establish an analog of Proposition~\ref{big 2}. 

\begin{prop}
Let $(\theta,\theta')$ be a 3-atom where the blocks of $\theta$ are $T,M,B$. If $\delta$ is an involution of $M$ whose set of fixed points has the same cardinality as $M$, then $(\theta,\theta')\approx(\theta,\delta(\theta'))$. 
\label{big 3}
\end{prop}

\begin{proof}
The proof follows that of Proposition~\ref{big 2}. Enumerate the 2-element orbits of $\delta$ with a set $J$, then label the elements of the $j^{th}$ orbit as $b_j$ and $d_j$. The remaining elements of $M$ are fixed points of $\delta$. Split these into two sets, one of cardinality 7 times that of $J$, and the other infinite. Enumerate the elements of this first set as $f_j,h_j,l_j,m_j,n_j,z_j,\omega_j$ $(j\in J)$. Split this second set of fixed points into infinitely many 9-tuples, enumerate these 9-tuples over some set $K$ that is disjoint from $J$, and then enumerate the elements of the $k^{th}$ 9-tuple as $b_k,d_k,f_k,h_k,l_k,m_k,n_k,z_k,\omega_k$. Set $I=J\cup K$. Then the elements of $M$ are enumerated as 

\[ b_i,d_i,f_i,h_i,l_i,m_i,n_i,z_i,\omega_i \quad (i\in I)\]
\vspace{-1ex}

Enumerate the elements of $T$ as $a_i,c_i,e_i,g_i,i_i,j_i,k_i,x_i,y_i$ $(i\in I)$ and enumerate the elements of $B$ as $o_i,r_i,p_i,s_i,u_i,v_i,w_i,q_i,t_i$ $(i\in I)$ in the unique way possible with $\{a_i,b_i,o_i\},\ldots,\{y_i,\omega_i,t_i\}$ blocks of $\theta'$. Define $\gamma$ and $\gamma'$ so that the blocks of $\gamma$ are $\{a_i:i\in I\},\ldots,\{\omega_i:i\in I\}$ and the blocks of $\gamma'$ are $\{a_i,\ldots,\omega_i:i\in I\}$. 

Define $\xi:\{A_i,\ldots,\Omega_i:i\in I\}\to X$ by setting $\xi(A_i)=a_i,\ldots,\xi(\Omega_i)=\omega_i$. The constructions provided yield that $\xi$ is a special 27-enumeration compatible with $(\theta,\theta')$ being under $(\gamma,\gamma')$ with the labeling exactly as in Figure~\ref{3-atom setup}. Further, the $J$-twist $\theta'_{J, \xi}$ is exactly $\delta(\theta')$. The result then follows by Proposition~\ref{twist 27}. 
\end{proof}

We come now to the aim of this section.

\begin{thm}
For any two 3-atoms $(\theta,\theta')$ and $(\theta,\theta'')$ with the same first component, 

$$(\theta,\theta')\approx_3(\theta,\theta'')$$
\label{david}
\end{thm}

\begin{proof}
We follow the proof of Theorem~\ref{tiger}. Suppose $(\theta,\theta')$ is a 3-atom and set 

\[ S = \{\sigma\in\Perm(M):(\theta,\theta')\approx_3(\theta,\sigma(\theta'))\}\]
\vspace{-1ex}

\noindent Let $U$ be the set of all involutions of $M$ whose set of fixed points has cardinality equal to that of $M$. Then as in the proof of Theorem~\ref{tiger}, Proposition~\ref{big 3} shows that $S$ contains all composites of members of $U$. Therefore, as in the proof of Theorem~\ref{tiger}, $S$ contains the subgroup of $\Perm(M)$ that is generated by $U$, and therefore $S=\Perm(M)$. 

Now fix any $\sigma\in\Perm(M)$ and consider 

\[ T_\sigma = \{\rho\in\Perm(B):(\theta,\sigma(\theta'))\approx_3(\theta,\rho(\sigma(\theta')))\}\]
\vspace{-1ex}

\noindent Our arguments apply also to the bottom level $B$ of $\theta$, and so $T_\sigma=\Perm(B)$. 
Then for any $\theta''$ with $(\theta,\theta'')$ a 3-atom, there is a permutation $\sigma$ of $M$ and a permutation $\rho$ of $B$ with $\rho(\sigma(\theta'))=\theta''$. Since $\sigma\in S$ then $(\theta,\theta')\approx_3(\theta,\sigma(\theta'))$, and since $\rho\in T_\sigma$, then $(\theta,\sigma(\theta'))\approx_3(\theta,\rho(\sigma(\theta')))$. Since $\approx_3$ is by definition transitive, $(\theta,\theta')\approx(\theta,\theta'')$. 
\end{proof}


\section{From automorphisms of $\Fact(X)$ to automorphisms of $\Eq*(X)$}

In this section we show that if $\alpha$ is an automorphism of $\Fact(X)$ and $(\theta,\theta')$, $(\theta,\theta'')$ are factor pairs with the same first component and $X/\theta$ infinite, then $\alpha(\theta,\theta')$ and $\alpha(\theta,\theta'')$ also have the same first component. This is Proposition~\ref{f}, restated as Proposition~\ref{larry} below. It is used to define from $\alpha$ an automorphism $\beta$ of $\Eq*(X)$. This is Proposition~\ref{g}, restated as Proposition~\ref{g'} below. 

\begin{defn}
Let $M$ and $N$ be sets where the elements of $M$ are ordered pairs and the elements of $N$ are ordered pairs. A map $f:M\to N$ is said to preserve first coordinates if whenever $(m,m')$ and $(m,m'')$ are elements of $M$ with the same first coordinate, then their images $f(m,m')$ and $f(m,m'')$ also have the same first coordinate. 
\end{defn}

In Proposition~\ref{dopey} it was shown that for each permutation $\sigma\in\Perm(X)$ there is an automorphism $\Gamma(\sigma)$ of $\Fact(X)$ given by $\Gamma(\sigma)(\theta,\theta') = (\sigma(\theta),\sigma(\theta'))$. Clearly each $\Gamma(\sigma)$ preserves first coordinates. We will required the following result that is of interest in its own right as it shows that a fragment of transitivity holds for the automorphism group of $\Fact(X)$. 

\begin{prop}
Let $(\theta,\theta')$ and $(\phi,\phi')$ be factor pairs of $X$ and let $\kappa$ and $\lambda$ be cardinals where
\vspace{1ex}

\begin{enumerate}
\item both $\theta$ and $\phi$ have $\kappa$ equivalence classes
\item both $\theta'$ and $\phi'$ have $\lambda$ equivalence classes
\end{enumerate}
\vspace{1ex}

\noindent Then there is a permutation $\sigma\in\Perm(X)$ with $\Gamma(\sigma)$ mapping $(\theta,\theta')$ to $(\phi,\phi')$. Further, any two factor pairs $(\theta,\theta')$ and $(\theta,\theta'')$ with the same first coordinate satisfy these two conditions, and in this case $\sigma$ can be chosen so that $(x,\sigma(x))\in\theta$ for each $x\in X$. 
\label{move}
\end{prop}

\begin{proof}
Enumerate the classes of $\theta$ as $F_j$ $(j\in \lambda)$, the classes of $\theta'$ as $E_i$ $(i\in\kappa)$, the classes of $\phi$ as $H_j$ $(j\in\lambda)$, and the classes of $\phi'$ as $G_i$ $(i\in\kappa)$. The elements of $X$ can be enumerated as $x_{i,j}$ where $x_{i,j}$ is the unique element of $E_i\cap F_j$. They can also be enumerated as $y_{i,j}$ where $y_{i,j}$ is the unique element of $G_i\cap H_j$. Define $\sigma(x_{i,j})=y_{i,j}$. Then $\sigma$ is a permutation of $X$. Note 
\begin{align*}
\,E_i=\{x_{i,j}:j\in\lambda\}\quad & \quad \,F_j=\{x_{i,j}:i\in\kappa\}\\
G_i=\{y_{i,j}:j\in\lambda\}\quad & \quad H_j=\{x_{i,j}:i\in\kappa\}
\end{align*}
So $\sigma$ maps $E_i$ to $G_i$ and $F_j$ to $H_j$. This implies that $\sigma(\theta)=\phi$ and $\sigma(\theta')=\phi'$, hence $\Gamma(\sigma)$ maps $(\theta,\theta')$ to $(\phi,\phi')$. If we begin with $(\theta,\theta')$ and $(\theta,\theta'')$, then the number of blocks of both $\theta'$ and $\theta''$ is equal to the cardinality of each block of $\theta$, so the two conditions are satisfied. Then in the proof above, taking $F_j=H_j$, we have that $x_{i,j}$ and $\sigma(x_{i,j})=y_{i,j}$ both belong to the equivalence class $F_j=H_j$ of $\theta=\phi$. It follows that $(x_{i,j},\sigma(x_{i,j}))\in\theta$. 
\end{proof}

We now come to the crucial result that required the efforts of the previous section. 

\begin{prop}
If $\alpha$ is an automorphism of $\Fact(X)$ and $(\theta,\theta')$ and $(\theta,\theta'')$ are 3-atoms with the same first coordinate, then $\alpha(\theta,\theta')$ and $\alpha(\theta,\theta'')$ have the same first coordinate. 
\label{nm}
\end{prop}

\begin{proof}
By Theorem~\ref{david}, $(\theta,\theta')\approx_3(\theta,\theta'')$. Since $\approx_3$ is the transitive closure of $\sim_3$, it is enough to show that if $(\theta,\theta')$ and $(\theta,\theta'')$ are 3-atoms with $(\theta,\theta')\sim_3(\theta,\theta'')$, then $\alpha(\theta,\theta')$ and $\alpha(\theta,\theta'')$ have the same first coordinate. Suppose that $(\theta,\theta')\sim_3(\theta,\theta'')$. Then by definition there is a factor pair $(\gamma,\gamma')$ with $|X/\gamma\,|=27$ and both $(\theta,\theta')$ and $(\theta,\theta'')$ belonging to the interval $[(\nabla,\Delta),(\gamma,\gamma')]$. 

The proof of Theorem~\ref{interval} in \cite{Taewon} gives mutually inverse isomorphisms $\Sigma$ and $\Phi$ 
\begin{center}
\begin{tikzpicture}[scale = 2.0]
\node at (0,0) {$[(\nabla,\Delta),(\gamma,\gamma')]$};
\node at (2,0) {$\Fact(X/\gamma)$};
\draw [->] (0.8,0.04) to (1.4,0.04);
\draw [->] (1.4,-0.02) to (0.8,-0.02);
\node at (1.1,0.2) {$\Sigma$};
\node at (1.1,-0.17) {$\Phi$};
\end{tikzpicture}
\end{center}
\noindent In the setting of sets, these isomorphisms are as follows. 
\vspace{-1ex}

\begin{align*}
\Sigma(\theta,\theta') &= (\theta/\gamma,(\theta'\circ\gamma)/\gamma)\\ 
\Phi(\mu,\nu) &= (\hat{\mu},\hat{\nu}\cap\gamma')
\end{align*}
\vspace{-1ex}

\noindent Here $\theta/\gamma=\{(x/\gamma,y/\gamma):(x,y)\in\theta\}$ and $\hat{\mu}=\{(x,y):(x/\gamma,y/\gamma)\in\mu\}$. The key point is that both $\Sigma$ and $\Phi$ preserve first coordinates. 

Let $\alpha(\gamma,\gamma') = (\gamma_1,\gamma_1')$. Since $X/\gamma$ has 27 elements, $\gamma'$ is a 27-relation. So by Proposition~\ref{e'}, $\gamma_1'$ is a 27-relation, hence $X/\gamma_1$ has 27 elements. So $\gamma$ and $\gamma_1$ both have 27 equivalence classes, and the cardinalities of the sets of equivalence classes of $\gamma'$ and $\gamma_1'$ both equal the cardinality of $X$. So By Proposition~\ref{move} there is a permutation $\sigma$ of $X$ with $\Gamma(\sigma)(\gamma_1,\gamma_1')=(\gamma,\gamma')$. 

Set 
\begin{align*}
\delta & \,\,=\,\, \Gamma(\sigma)\circ\alpha\\
I&\,\,=\,\, [(\nabla,\Delta),(\gamma,\gamma')] 
\end{align*}
\vspace{-1ex}

\noindent Then $\delta$ is an automorphism of $\Fact(X)$ that maps $(\gamma,\gamma')$ to itself. Since $I$ consists of all elements of $\Fact(X)$ beneath $(\gamma,\gamma')$, we have that the restriction $\delta|I$ is an automorphism of the \ts{omp} $I$.  Using the mutually inverse isomorphisms between $I$ and $\Fact(X/\gamma)$ described above, let $\beta$ be the automorphism of $\Fact(X/\gamma)$ given by 

\[\beta = \Sigma\,\circ\,(\delta|I)\,\circ\,\Phi\]
\vspace{-1ex}

Since $X/\gamma$ is a 27-element set, Theorem~\ref{27} shows that $\beta$ is given by a permutation of $X/\gamma$, that is, there is a permutation $\rho$ of $X/\gamma$ with $\beta=\Gamma(\rho)$. In particular, $\beta$ preserves first coordinates. Then $\delta|I=\Phi\circ\beta\circ\Sigma$ is a composite of maps that preserve first coordinates, hence $\delta|I$ preserves first coordinates. Since $(\theta,\theta')$ and $(\theta,\theta'')$ belong to $I$, then $\delta(\theta,\theta')$ and $\delta(\theta,\theta'')$ have the same first coordinates, and since $\Gamma(\sigma^{-1})$ preserves first coordinates, then $(\Gamma(\sigma^{-1})\circ\delta)(\theta,\theta')$ and $(\Gamma(\sigma^{-1})\circ\delta)(\theta,\theta'')$ have the same first coordinates. But $\Gamma(\sigma^{-1})$ is the inverse of $\Gamma(\sigma)$, hence $\Gamma(\sigma^{-1})\circ\delta = \alpha$. This provides the result. 
\end{proof}

\begin{prop}
If $\alpha$ is an automorphism $\alpha$ of $\Fact(X)$ and $(\theta,\theta')$ and $(\theta,\theta'')$ are elements of $\Fact(X)$ with the same first coordinate and $X/\theta$ infinite, then $\alpha(\theta,\theta')$ and $\alpha(\theta,\theta'')$ have the same first coordinate. 
\label{larry}
\end{prop}

\begin{proof}
By Proposition~\ref{move} there is $\sigma\in\Perm(X)$ with $\Gamma(\sigma)(\theta,\theta')=(\theta,\theta'')$ and $(x,\sigma(x))\in\theta$ for each $x\in X$. We use this to show that for any equivalence relation $\phi$

\begin{equation}
\tag{A}
\label{zigger}
\mbox{there is $\phi'$ with $(\phi,\phi')\leq(\theta,\theta')$ iff there is $\phi''$ with $(\phi,\phi'')\leq(\theta,\theta'')$}
\end{equation}
\vspace{-1ex}

Suppose that $\phi'$ is a relation with $(\phi,\phi')\leq (\theta,\theta')$. Then $\Gamma(\sigma)(\phi,\phi')\leq\Gamma(\sigma)(\theta,\theta')$, giving that $(\sigma(\phi),\sigma(\phi')) \leq (\theta,\theta'')$. Since $(\phi,\phi')\leq(\theta,\theta')$, we have $\theta\subseteq\phi$. We claim that $\sigma(\phi)=\phi$. Suppose $(x,y)\in\phi$. Using the property that $(x,\sigma(x))\in\theta$ for each $x\in X$ we have 

\[\sigma (x)\, \theta\, x\, \phi\, y\, \theta\, \sigma (y)\] 
\vspace{-1ex}

\noindent Then since $\theta\subseteq\phi$, this shows that $(\sigma(x),\sigma(y)) \in \phi$. So $\sigma(\phi)\subseteq\phi$. Note that $\sigma^{-1}$ satisfies $(\sigma^{-1}(y),y)\in\theta$ for each $y\in X$, as can be seen by substituting $y=\sigma(x)$ into the corresponding property for $\sigma$. Then if $(x,y)\in\phi$, 
\[ \sigma^{-1}(x)\,\theta\, x\,\phi\, y\,\theta\,\sigma^{-1}(y) \]
\vspace{-1ex}

\noindent Since $\theta\subseteq\phi$, this shows that $(\sigma^{-1}(x),\sigma^{-1}(y))\in\phi$, and hence that $(x,y)\in\sigma(\phi)$. So $\phi\subseteq\sigma(\phi)$, showing equality. Since $\phi=\sigma(\phi)$ and $(\sigma(\phi),\sigma(\phi'))\leq(\theta,\theta'')$, there exists $\phi''=\sigma(\phi')$ with $(\phi,\phi'')\leq(\theta,\theta'')$. This gives one implication in (\ref{zigger}) and symmetry gives the other. 

We now show that $\alpha(\theta,\theta')$ and $\alpha(\theta,\theta'')$ have the same first coordinate. Suppose that $\alpha(\theta,\theta') = (\gamma,\gamma')$ and $\alpha(\theta,\theta'')=(\mu,\mu')$. Since $X/\theta$ is infinite, $\theta$ has infinitely many blocks. So the blocks of both $\theta'$ and $\theta''$ must contain infinitely many elements. So neither $\theta,\theta'$ is an $n$-relation for any natural number $n$. Proposition~\ref{e'} gives that neither $\gamma',\mu'$ is an $n$-relation for any $n$, and therefore that both $X/\gamma$ and $X/\mu$ are infinite. 

Since both $X/\gamma$ and $X/\mu$ are infinite, we may apply Proposition~\ref{b'} in the case of 3-atoms. This gives 
\[\gamma = \bigcap\{\xi:\mbox{ there is a 3-atom }(\xi,\xi')\leq(\gamma,\gamma')\}\]
\[\mu = \bigcap\{\epsilon:\mbox{ there is a 3-atom }(\epsilon,\epsilon')\leq(\mu,\mu')\}\]
\vspace{0ex}

Suppose there is a 3-atom $(\xi,\xi')\leq(\gamma,\gamma')$ and let $\alpha^{-1}(\xi,\xi')=(\phi,\phi')$. Then $(\phi,\phi')\leq(\theta,\theta')$. So by (\ref{zigger}) there is $\phi''$ with $(\phi,\phi'')\leq(\theta,\theta'')$. Since $(\xi,\xi')$ is a 3-atom and $\alpha^{-1}$ is an automorphism, Proposition~\ref{crud} yields that $(\phi,\phi')$ is a 3-atom. So $X/\phi$ has 3 elements, and therefore $(\phi,\phi'')$ is also a 3-atom. Since the 3-atoms $(\phi,\phi')$ and $(\phi,\phi'')$ have the same first coordinate, by Proposition~\ref{nm} $\alpha(\phi,\phi')$ and $\alpha(\phi,\phi'')$ have the same first coordinate. But $\alpha(\phi,\phi')=(\xi,\xi')$, so its first coordinate is $\xi$, hence the first coordinate of $\alpha(\phi,\phi'')$ is also $\xi$. So there is $\xi''$ with $\alpha(\phi,\phi'')=(\xi,\xi'')$. By Proposition~\ref{crud} $(\xi,\xi'')$ is a 3-atom, and since $(\phi,\phi'')\leq(\theta,\theta'')$ we have $(\xi,\xi'')\leq(\mu,\mu')$. So each $\xi$ that occurs in the meet for $\gamma$ also occurs in the meet for $\mu$. So $\mu\leq\gamma$, and $\gamma\leq\mu$ by symmetry. 
\end{proof}

We now come to the culmination of our considerations of the structure of $\Fact(X)$, the following result which is a restatement of Theorem~\ref{g}. Here we recall from Definition~\ref{fool} that $\Eq*(X)$ is the poset set of all finite regular equivalence relations on $X$.

\begin{thm}
For any automorphism $\alpha$ of $\Fact(X)$ there is an automorphism $\beta$ of $\Eq*(X)$ with the following properties.
\vspace{1ex}
\begin{enumerate}
\item $\theta$ is an $n$-relation iff $\beta(\theta)$ is an $n$-relation
\item if $(\theta,\theta')\in\Fact(X)$ with $\theta\in\Eq*(X)$, then $\beta(\theta)$ is the first component of $\alpha(\theta,\theta')$
\end{enumerate}
\label{g'}
\end{thm}

\begin{proof}
The relations $\theta\in\Eq*(X)$ are those where there is a natural number $n$ so that every block has $n$ elements. Such $\theta$ must have infinitely many blocks, so have $X/\theta$ infinite. As with every regular equivalence relation, there is an equivalence relation $\theta'$ with $(\theta,\theta')$ a factor pair. Theorem~\ref{larry} shows that if $\theta'$ and $\theta''$ are such that $(\theta,\theta')$ and $(\theta,\theta'')$ are factor pairs, then $\alpha(\theta,\theta')$ and $\alpha(\theta,\theta'')$ have the same first component. Proposition~\ref{e'} shows that if $\theta$ is an $n$-relation, then so is the first component of $\alpha(\theta,\theta')$. Therefore we can define a map $\beta$ from $\Eq*(X)$ to itself by setting $\beta(\theta)$ to be the first component of any $\alpha(\theta,\theta')$, and this map $\beta$ has the property that if $\theta$ is an $n$-relation, then so is $\beta(\theta)$. The same argument applied to the automorphism $\alpha^{-1}$ produces another map from $\Eq*(X)$ to itself that is the inverse of $\beta$, so $\beta$ is a bijection from $\Eq*(X)$ to itself. 

It remains to show that $\beta$ is order preserving. Suppose that $\theta,\phi\in\Eq*(X)$ with $\theta\subseteq\phi$. Then there are natural numbers $m,n$ so that $\theta$ is an $n$-relation and $\phi$ is an $m$-relation. Since $\theta\subseteq\phi$ each block of $\phi$ is a union of blocks of $\theta$, and therefore $m$ is a multiple of $n$. Say $m=pn$. Enumerate the blocks of $\phi$ as $F_i$ $(i\in\kappa)$ for some infinite cardinal $\kappa$. Each block $F_i$ for $i\in\kappa$ will contain exactly $p$ blocks of $\theta$, enumerate these blocks of $\theta$ as $E_{i,j}$ for $1\leq j\leq p$. For $i\in\kappa$ and $1\leq j\leq p$ the block $E_{i,j}$ of $\theta$ has $n$ elements. Enumerate these as $x_{i,j,k}$ where $1\leq k\leq n$. Now define $\theta'$ and $\phi'$ by requiring 
\begin{align*}
\phi'\mbox{ has blocks }&\quad \{x_{i,j,k}:i\in\kappa\}\mbox{ for each }1\leq j\leq p, 1\leq k\leq n\\
\theta'\mbox{ has blocks }&\quad \{x_{i,j,k}:i\in\kappa,1\leq j\leq p\}\mbox{ for each } 1\leq k\leq n
\end{align*}

Since the intersection of blocks of $\theta$ and $\theta'$ contains exactly one element, and the intersection of blocks of $\phi$ and $\phi'$ contains exactly one element, $(\theta,\theta')$ and $(\phi,\phi')$ are factor pairs. We have $\theta\subseteq\phi$ and each block of $\phi'$ is contained in a block of $\theta'$ so $\phi'\subseteq\theta'$. That $\theta\circ\phi'=\phi'\circ\theta$ follows since both expressions give the equivalence relation whose blocks are $\{x_{i,j,k}:i\in\kappa,1\leq k\leq n\}$ for each $1\leq j\leq p$. So $(\phi,\phi')\leq(\theta,\theta')$. Then $\alpha(\phi,\phi')\leq\alpha(\theta,\theta')$. This implies that the first component $\beta(\theta)$ of $\alpha(\theta,\theta')$ is contained in the first component of $\alpha(\phi,\phi')$. So $\beta(\theta)\subseteq\beta(\phi)$, giving that $\beta$ is order preserving, hence an automorphism of $\Eq*(X)$. 
\end{proof}

\section{From automorphisms of $\Eq*(X)$ to maps of subsets of $X$}

We now begin the second half of the program. The eventual aim is to begin with an automorphism $\beta$ of $\Eq*(X)$ that takes $n$-relations to $n$-relations, and to produce from this a permutation $\sigma$ of $X$ so that $\beta(\theta)=\sigma(\theta)$ for each $\theta\in\Eq*(X)$. In this section we take the first steps and construct from $\beta$ and any 2-relation $\pi$ a bijection from the subsets of $X$ that are the union of blocks of $\pi$ to the subsets of $X$ that are the union of blocks of $\beta(\pi)$. 

\begin{defn} 
For a set $Y$, a subset $S$ of $\Eq*(Y)$ and $k,n$ natural numbers, let 
\vspace{1ex}
\begin{align*}
\U_k(S) &= \{\phi\in\Eq*(Y):\mbox{$\phi$ is an $k$-relation and $\theta\subseteq\phi$ for each }\theta\in S\}\\
\LL_n(S) &= \{\phi\in\Eq*(Y):\mbox{$\phi$ is an $n$-relation and $\phi\subseteq\theta$ for each }\theta\in S\}
\end{align*}
\end{defn}

A few remarks about usage are worthwhile. When $S$ is a single element, or small finite set of relations, we often use shorthand such as $\U_k(\theta)$ or $\U_k(\theta_1,\theta_2)$ with obvious meaning. This notion will be applied not only to $\Eq*(X)$ for our given infinite set $X$, but also to subsets $Y$ of~$X$. We will not explicitly mention the ambient set $X$ or $Y$ when using $\U_k$ or $\LL_n$ since this is clear from the context. 

\begin{prop} 
For natural numbers $k,n$ 
\vspace{1ex}
\begin{enumerate}
\item $S\subseteq T\Rightarrow \U_k(T)\subseteq\U_k(S)$
\item $S\subseteq T\Rightarrow \LL_n(T)\subseteq\LL_n(S)$
\item if $S$ is a set of $n$-relations, then $S\subseteq\LL_n\!\U_k(S)$
\item if $S$ is a set of $k$-relations, then $S\subseteq\U_k \LL_n(S)$
\item if $S$ is a set of $n$-relations, then $\U_k\LL_n\!\U_k(S)=\U_k(S)$
\item if $S$ is a set of $k$-relations, then $\,\LL_n\!\U_k \LL_n(S)=\LL_n(S)$
\item $S=\LL_n\!\U_k(S) \Leftrightarrow$ there exists a set $T$ of $k$-relations with $S=\LL_n(T)$
\item $S=\U_k \LL_n(S)\Leftrightarrow$ there exists $T$ with $S=\U_k(T)$
\item $\U_k(\bigcup_I S_i) = \bigcap_I \U_k(S_i)$
\item $\LL_n(\bigcup_I S_i) = \bigcap_I\LL_n(S_i)$
\end{enumerate}
\label{n}
\end{prop}

\begin{proof}
(1) and (2) are obvious from the form of the definition. For (3), if $\theta\in S$, then by definition any $\phi\in\U_k(S)$ contains $\theta$, and since every relation in $S$ is an $n$-relation, then $\theta$ is an $n$-relation contained in every member of $\U_k(S)$, hence $\theta\in\LL_n\!\U_k(S)$. Then (4) is by symmetry with (3). For (5), part (3) gives that $S\subseteq\LL_n\!\U_k(S)$, and then part (1) gives $\U_k\LL_n\U_k(S)\subseteq\U_k(S)$. But by definition $\U_k(S)$ is a set of $k$-relations, so by part (4) $\U_k(S)\subseteq\U_k\LL_n\U_k(S)$. This establishes (5) and (6) is by symmetry. For (7) ``$\,\Rightarrow$'' if $S=U_k\LL_n(S)$, put $T=\U_k(S)$. For ``$\Leftarrow$'' if $S=\LL_n(T)$ then $\LL_n\!\U_k(S)=\LL_n\!\U_k\LL_n(T)$, and by (6) $\LL_n\!\U_k\LL_n(T) = \LL_n(T) = S$. This establishes (7), and~(8) is by symmetry. Finally (9) and (10) are simple consequences of the definitions. 
\end{proof}

These considerations should be familiar from the construction of MacNeille completions. We extend the notion of normal ideals to this setting in the following. 

\begin{defn}
A subset $S\subseteq\Eq*(X)$ is an $(n,k)$-normal ideal if $S=\LL_n\!\U_k(S)$. 
\end{defn}

\begin{prop}
For $n$, $k$ natural numbers and $S$ a set of $n$-relations, $\LL_n\!\U_k(S)$ is the smallest $(n,k)$-normal ideal that contains $S$. Call this the  $(n,k)$-normal ideal generated by $S$. 
\label{p}
\end{prop}

\begin{proof}
There is a set $T=\U_k(S)$ of $k$-relations with $\LL_n\!\U_k(S) = \LL_n(T)$, so by Proposition~\ref{n} part~(7) $\LL_n\!\U_k(S)$ is an $(n,k)$-normal ideal. Since $S$ is a set of $n$-relations, then by part (3) of Proposition~\ref{n} this $(n,k)$-normal ideal $\LL_n\!\U_k(S)$ contains $S$. If $N$ is another such, then $S\subseteq N$ implies by Proposition~\ref{n} parts (1) and (2) that $\LL_n\!\U_k(S)\subseteq\LL_n\!\U_k(N) = N$. 
\end{proof}

\begin{prop}
For any 2-relation $\theta$ on an infinite set $X$ we have $\LL_2\!\U_4(\theta)=\{\theta\}$. 
\label{q}
\end{prop}

\begin{proof}
By Proposition~\ref{p} we have $\theta\in\LL_2\!\U_4(\theta)$. Suppose that $\phi$ is a 2-relation that is not equal to $\theta$. Then there are $x,y$ with $(x,y)\in\theta$ and $(x,y)\not\in\phi$. Let $z$ be distinct from $x$ with $(x,z)\in\phi$. Choose two elements $a,b$ distinct from $x,y,z$ with $(a,b)\in\theta$. Then $\{x,y\}$ and $\{a,b\}$ are blocks of $\theta$. Enumerate the blocks of $\theta$ as $E_i$ $(i\in\kappa)$ with $E_0=\{x,y\}$ and $E_1=\{a,b\}$. Form a 4-relation $\gamma$ whose blocks are $F_i$ $(i\in\kappa)$ where $F_i=E_{2i}\cup E_{2i+1}$. Then $\theta\subseteq\gamma$ so $\gamma\in\U_4(\theta)$. But $x,z$ are in different blocks of $\gamma$, so $\phi$ is not a subset of $\gamma$. Thus $\phi\not\in\LL_i\!\U_4(\theta)$. 
\end{proof}

\begin{defn}
For a set $X$, an equivalence relation $\pi$ on $X$ and a natural number $n$, let 
\vspace{1ex}

\begin{enumerate}
\item $\mathcal{P}(n,\pi)$ be the set of all $n$-element subsets of $X$ that are the union of blocks of $\pi$
\item $\mathcal{P}(\pi)$ be the set of all subsets of $X$ that are the union of blocks of $\pi$
\item $\mathcal{P}(n)$ be the set of all $n$-element subsets of $X$
\end{enumerate}
\label{kilnk}
\end{defn}

\begin{defn}
We say that distinct 2-relations $\theta_1,\theta_2$ agree except on the 4-set $A$ if $A$ is a 4-element set and 
\vspace{1ex}

\begin{enumerate}
\item $A$ is the union of two blocks of $\theta_1$ and $A$ is the union of two blocks of $\theta_2$
\item $\theta_1$ and $\theta_2$ are not equal on $A$
\item $\theta_1$ and $\theta_2$ are equal on $X\setminus A$
\end{enumerate}
\end{defn}

The following observation lies at the heart of producing a property that is preserved by an automorphism $\beta$ of $\Eq*(X)$ that connects 2-relations to 4-elements subsets of $X$. This proof of this proposition is self evident from its statement. 

\begin{prop}
Let $\theta_1$ be a 2-relation and $\{a,b\},\{c,d\}$ be blocks of $\theta_1$. Then there are exactly two 2-relations $\theta_2,\theta_3$ that agree with $\theta_1$ except on $A=\{a,b,c,d\}$. These relations agree with $\theta_1$ on $X\setminus A$ and partition $A$ as follows:
\begin{align*}
\theta_2&\quad\mbox{has blocks }\{a,c\},\{b,d\}\\
\theta_3&\quad\mbox{has blocks }\{a,d\},\{b,c\}
\end{align*}
Further, any two of $\theta_1,\theta_2,\theta_3$ agree except on $A$, and $A$ is the unique set for which this is true. 
\label{tater}
\end{prop}

\begin{prop}
Let $\theta_1,\theta_2$ be distinct 2-relations on the infinite set $X$ and let $N=\LL_2\!\U_4(\theta_1,\theta_2)$ be the $(2,4)$-normal ideal they generate. Then the following are equivalent. 
\vspace{1ex}

\begin{enumerate}
\item $\theta_1,\theta_2$ agree except on a 4-set $A$
\item $N$ has exactly three elements $\theta_1,\theta_2,\theta_3$
\end{enumerate}
\vspace{1ex}

\noindent When these conditions are satisfied $\theta_3$ is the unique 2-relation so that any two of $\theta_1,\theta_2,\theta_3$ agree except on $A$, and the normal ideal $N$ is generated by any two of $\theta_1,\theta_2,\theta_3$. 
\label{iti}
\end{prop}

\begin{proof}
(1) ``$\,\Rightarrow\,$'' (2) 
Suppose $A=\{a,b,c,d\}$ where $\{a,b\},\{c,d\}$ are blocks of $\theta_1$ and $\{a,c\},\{b,d\}$ are blocks of $\theta_2$. Let $\theta_3$ be the third relation that agrees with $\theta_1$ except on $A$ as given in Proposition~\ref{tater}. Any 4-relation $\gamma$ that is an upper bound of $\theta_1,\theta_2$ has $A$ as one of its blocks and therefore $\theta_3\subseteq\gamma$. So $\theta_1,\theta_2,\theta_3$ all belong to $N$. Suppose that $\phi$ is a 2-relation distinct from $\theta_1,\theta_2,\theta_3$. Then $\phi$ must differ from $\theta_1$ on $X\setminus A$. So there are $x,y\in X\setminus A$ with $(x,y)\in\theta_i$ for $i=1,2,3$ and $(x,y)\not\in\phi$. An obvious modification of the proof of Proposition~\ref{q} produces a 4-relation $\gamma$ that is an upper bound of $\theta_1,\theta_2,\theta_3$ but does not contain $\phi$. 

(2) ``$\,\Rightarrow\,$'' (1) To have $N$ not be the set of all 2-relations, there must be a 4-relation $\gamma$ that is an upper bound of $\theta_1$ and $\theta_2$. Then each block of $\gamma$ is the union of two blocks of $\theta_1$ and the union of two blocks of $\theta_2$. Since $\theta_1\neq\theta_2$ there is a block $\{a,b\}$ of $\theta_1$ that is not a block of $\theta_2$. Suppose $\{a,c\}$ is the block of $\theta_2$ that contains $a$ and that $A=\{a,b,c,d\}$ is the block of $\gamma$ that contains $a$. Then $\{a,b\},\{c,d\}$ are blocks of $\theta_1$ and $\{a,c\},\{b,d\}$ are blocks of $\theta_2$. If $\theta_1$ and $\theta_2$ do not agree on $X\setminus A$, we can repeat this process to find blocks $\{e,f\},\{g,h\}$ of $\theta_1$ that are disjoint from $A$ with $\{e,g\},\{f,h\}$ blocks of $\theta_2$. Set $B=\{e,f,g,h\}$. Then every 4-relation that contains $\theta_1$ and $\theta_2$ has both $A$ and $B$ as blocks. There are nine 2-relations that agree with $\theta_1$ on $X\setminus (A\cup B)$ and have $A$ as the union of two of their blocks and $B$ as the union of two of their blocks. All of these relations belong to $N$, a contradiction. So $\theta_1$ and $\theta_2$ agree on $X\setminus A$. 

The further comments are obvious from what has been shown. 
\end{proof}

We now connect these notions to the automorphism $\beta$ of $\Eq*(X)$ produced by Theorem~\ref{g'}. We introduce some terminology to make statements more concise. 

\begin{defn}
An automorphism of $\Eq*(X)$ is size preserving if $\theta$ being an $n$-relation implies that $\beta(\theta)$ is an $n$-relation. 
\end{defn}

The poset $\Eq*(X)$ is a labelled poset, meaning that to each element is assigned a natural number $n$. A size preserving automorphism of $\Eq*(X)$ is an automorphism of this labelled poset. It is clear that if $\beta$ is size preserving, then its inverse $\beta^{-1}$ is also size preserving. The automorphism $\beta$ produced by Theorem~\ref{g'} is size preserving. 

\begin{prop}
For a size preserving automorphism $\beta$ of $\Eq*(X)$ and a 2-relation $\pi$ there are bijections
\begin{align*}
\sigma^\pi_4&:\mathcal{P}(4,\pi)\to\mathcal{P}(4,\beta(\pi))\\
\mu^\pi_4&:\mathcal{P}(4,\pi)\to\mathcal{P}(4,\beta^{-1}(\pi))
\end{align*}
\vspace{-1ex}

\noindent with $\sigma^\pi_4$ and $\mu^{\beta(\pi)}_4$ inverses. These maps satisfy for any 2-relation $\theta$ and any $A\in\mc{P}(4,\pi)$ 
\vspace{1.5ex}

\begin{enumerate}
\item $\pi$ and $\theta$ agree except on $A\,\,$ iff $\,\,\beta(\pi)$ and $\beta(\theta)$ agree except on $\sigma^\pi_4(A)$ \vspace{1ex}
\item $\pi$ and $\theta$ agree except on $A\,\,$ iff $\,\,\beta^{-1}(\pi)$ and $\beta^{-1}(\theta)$ agree except on $\mu^\pi_4(A)$
\end{enumerate}
\label{s}
\end{prop}

\begin{proof}
Let $A\in\mc{P}(4,\pi)$. By Proposition~\ref{tater} there are exactly two 2-relations $\theta_1$ and $\theta_2$ that agree with $\pi$ except on $A$, and by Proposition~\ref{iti} $N=\{\pi,\theta_1,\theta_2\}$ is the $(2,4)$-normal ideal generated by any two of its elements. Since $\beta$ is size preserving and an order isomorphism, it follows that $\beta(N) = \{\beta(\pi),\beta(\theta_1),\beta(\theta_2)\}$ is also the $(2,4)$-normal ideal generated by any two of its elements. So by Proposition~\ref{iti} may define 

\[\sigma^\pi_4(A) = B \quad \mbox{ where any two of $\beta(\pi), \beta(\theta_1),\beta(\theta_2)$ agree except on $B$}\]
\vspace{-1ex} 

\noindent This defines a map $\sigma^\pi_4:\mc{P}(4,\pi)\to\mc{P}(4,\beta(\pi))$. Since $\theta_1,\theta_2$ are the only 2-relations that agree with $\pi$ except on $A$ and $\beta(\theta_1),\beta(\theta_2)$ are the only 2-relations that agree with $\beta(\pi)$ except on $\sigma^\pi_4(A)$, property (1) holds. Identical reasoning using the size preserving automorphism $\beta^{-1}$ produces $\mu^\pi_4$ that satisfies (2). Finally, properties (1) and (2) applied to $\sigma^\pi_4$ and $\mu^{\beta(\pi)}_4$ show they are mutual inverses. 
\end{proof}

\begin{prop}
Let $\pi$ be a 2-relation and $A$ and $B$ be distinct sets in $\mc{P}(4,\pi)$. Let $\theta_1,\theta_2$ be the two 2-relations that agree with $\pi$ except on $A$, and let $\phi_1,\phi_2$ be the two 2-relations that agree with $\pi$ except on $B$. Then 
\[\mbox{ $A\cap B=\emptyset\quad$ iff $\quad\U_4(\{\pi,\theta_1,\theta_2,\phi_1,\phi_2\})\neq\emptyset$}\]
\vspace{-3ex}
\label{t}
\end{prop}

\begin{proof}
If $A$ and $B$ are disjoint we can enumerate the blocks of $\pi$ as $E_i$ $(i\in\kappa)$ so that $E_0\cup E_1=A$ and $E_2\cup E_3 = B$. Create a 4-relation $\gamma$ with blocks $F_i$ $(i\in\kappa)$ where $F_0=A$, $F_1=B$, and $F_i=E_{2i}\cup E_{2i+1}$ for each $i\geq 4$. Then $\gamma$ is an upper bound of $\pi,\theta_1,\theta_2,\phi_1,\phi_2$. Conversely, suppose there is a 4-relation $\gamma$ that contains $\pi,\theta_1,\theta_2,\phi_1,\phi_2$. Then $\gamma$ has both $A$ and $B$ as blocks. So $A,B$ must either be equal or disjoint, and we have assumed they are distinct. 
\end{proof}

\begin{cor}
For $A,B\in\mc{P}(4,\pi)$ we have $A\cap B=\emptyset\,$ iff $\,\sigma^\pi_4(A)\cap\sigma^\pi_4(B)=\emptyset$. 
\label{u}
\end{cor}

\begin{proof}
We use the same notation as was used in Proposition~\ref{t}. Suppose that $A\cap B=\emptyset$. By Proposition~\ref{t} there is a 4-relation $\gamma$ that is an upper bound of $\pi,\theta_1,\theta_2,\phi_1,\phi_2$. Then $\beta(\gamma)$ is an upper bound of $\beta(\pi),\beta(\theta_1),\beta(\theta_2),\beta(\phi_1),\beta(\phi_2)$. Since $\beta(\pi),\beta(\theta_1),\beta(\theta_2)$ agree except on $\sigma^\pi_4(A)$ and $\beta(\pi),\beta(\phi_1),\beta(\phi_2)$ agree except on $\sigma^\pi_4(B)$, then by Proposition~\ref{t} $\sigma^\pi_4(A)\cap\sigma^\pi_4(B)\neq\emptyset$. For the converse, apply the same reasoning using $\mu^{\beta(\pi)}$. 
\end{proof}

\begin{prop}
Suppose $A$, $B$ and $C$ are distinct elements of $\mc{P}(4,\pi)$ such that each pair of them intersects non-trivially. Let $\theta_1,\theta_2$ be the two 2-relations that agree with $\pi$ except on $A$, let $\phi_1,\phi_2$ be the two 2-relations that agree with $\pi$ except on $B$, and let $\chi_1,\chi_2$ be the two 2-relations that agree with $\pi$ except on $C$. Then

\[\mbox{ $A\cap B\cap C=\emptyset\quad$ iff $\quad\U_6(\{\pi,\theta_1,\theta_2,\phi_1,\phi_2,\chi_1,\chi_2\})\neq\emptyset$}\]
\vspace{0ex}

\noindent Therefore $A\cap B\cap C=\emptyset\,$ iff $\,\sigma_4^\pi(A)\cap\sigma_4^\pi(B)\cap\sigma_4^\pi(C)=\emptyset$.
\label{v1}
\end{prop}

\begin{proof}
Since each of $A,B$ and $C$ is a union of two blocks of $\pi$, they are distinct, and any two intersect non-trivially, there are two possibilities: either $A,B,C$ all contain the same block of $\pi$ or $A\cap B\cap C=\emptyset$. In the first case $A\cup B\cup C$ has eight elements. In the second case there must be blocks $E_1,E_2,E_3$ of $\pi$ with $A=E_1\cup E_2$, $B=E_1\cup E_3$ and $C=E_2\cup E_3$, so $A\cup B\cup C$ has six elements. Any 6-relation that contains $\pi,\theta_1,\theta_2,\phi_1,\phi_2,\chi_1,\chi_2$ must have a block that contains $A$, a block that contains $B$, and a block that contains $C$. If the first case happens then there is no 6-relation that is an upper bound of these elements since $A,B,C$ overlap in a 2-element set and their union has eight elements so cannot be contained in a 6-element block. In the second case, form a 6-relation $\gamma$ by letting $A\cup B\cup C$ be one of its blocks, then put the infinitely many blocks of $\pi$ not contained in $A\cup B\cup C$ into groups of three to form the other 6-element blocks of $\gamma$. Then $\gamma$ is an upper bound of $\pi,\theta_1,\theta_2,\phi_1,\phi_2,\chi_1,\chi_2$.
\end{proof}

\begin{prop}
If $A$, $B$ and $C$ are distinct elements of $\mc{P}(4,\pi)$ such that each pair of them intersects non-trivially, then 
\[\mbox{$A\cap B\cap C=\emptyset\quad$ iff $\quad\sigma_4^\pi(A)\cap\sigma_4^\pi(B)\cap\sigma_4^\pi(C)=\emptyset$}\]
\vspace{-3ex}
\label{v}
\end{prop}

\begin{proof}
Let $\theta_1,\theta_2$ be the 2-relations that agree with $\pi$ except on $A$; $\phi_1,\phi_2$ agree with $\pi$ except on $B$; and $\chi_1,\chi_2$ agree with $\pi$ except on $C$. Then by Proposition~\ref{s}
\vspace{-1ex}

\begin{align*}
\beta(\theta_1),\beta(\theta_2)&\mbox{ agree with $\beta(\pi)$ except on $\sigma^\pi_4(A)$}\\
\beta(\phi_1),\beta(\phi_2)&\mbox{ agree with $\beta(\pi)$ except on $\sigma^\pi_4(B)$}\\
\beta(\chi_1),\beta(\chi_2)&\mbox{ agree with $\beta(\pi)$ except on $\sigma^\pi_4(C)$}
\end{align*} 
\vspace{-1ex}

\noindent Since $\sigma^\pi_4$ is a bijection and $A,B,C$ are distinct, then $\sigma^\pi_4(A),\sigma^\pi_4(B)$, $\sigma^\pi_4(C)$ are distinct. Since any two of $A,B,C$ intersect non-trivially, by Corollary~\ref{u} any two of $\sigma^\pi_4(A),\sigma^\pi_4(B)$, $\sigma^\pi_4(C)$ intersect non-trivially. So $\sigma^\pi_4(A),\sigma^\pi_4(B),\sigma^\pi_4(C)$ satisfy the assumptions of Proposition~\ref{v1}.

Suppose $A\cap B\cap C=\emptyset$. Applying Proposition~\ref{v1} to $A,B,C$ there is a 6-relation $\gamma$ that is an upper bound of $\pi,\theta_1,\theta_2,\phi_1,\phi_2,\chi_1,\chi_2$. So $\beta(\gamma)$ is a 6-relation that is an upper bound of $\beta(\pi),\beta(\theta_1),\beta(\theta_2),\beta(\phi_1),\beta(\phi_2),\beta(\chi_1),\beta(\chi_2)$. Applying Proposition~\ref{v1} to $\sigma^\pi_4(A),\sigma^\pi_4(B),\sigma^\pi_4(C)$ gives that $\sigma^\pi_4(A)\cap \sigma^\pi_4(B)\cap \sigma^\pi_4(C)\neq\emptyset$. This provides one direction of the statement. The other is provided by the same argument using the map $\mu^{\beta(\pi)}_4$. 
\end{proof}

We require one more technical fact about the maps $\sigma^\pi_4$. 

\begin{prop}
For $\pi$ a 2-relation, $P\in\mc{P}(2,\pi)$, and $A,B,C,D\in\mc{P}(4,\pi)$, 
\vspace{1ex}

\begin{enumerate}
\item if $A\cap B = P$ then $\sigma^\pi_4(A)\cap\sigma^\pi_4(B)$ is an element of $\mc{P}(2,\beta(\pi))$\vspace{1ex}
\item if $A\cap B=P\!$ and $\,C\cap D=P$ then $\sigma^\pi_4(A)\,\cap\,\sigma^\pi_4(B)\,=\,\sigma^\pi_4(C)\,\cap\,\sigma^\pi_4(D)$
\end{enumerate}
\label{M}
\end{prop}

\begin{proof}
(1) Let $\sigma^\pi_4(A)\cap\sigma^\pi_4(B) = Q$. Since $Q$ is the intersection of two sets belonging to $\mc{P}(4,\beta(\pi))$ it is either (i) empty, (ii) is the union of two blocks of $\beta(\pi)$, or (iii) is a block of $\beta(\pi)$. Since $A\cap B = P$ we have $A\cap B\neq\emptyset$ and $A\neq B$. So Corollary~\ref{u} excludes the first possibility, and the fact that $\sigma^\pi_4$ is a bijection excludes the second. 

(2) Let $\sigma^\pi_4(C)\cap\sigma^\pi_4(D)=R$, and note that (1) also implies that $R$ is a block of $\beta(\pi)$. To prove (2) we first note that if $\{A,B\}=\{C,D\}$ there is nothing to prove. We next consider the case that exactly one of $A,B$ is equal to one of $C,D$. For this case, we may assume without loss of generality that $A=D$. Since $A\cap B = P$ and $P\subseteq C$, we have $A\cap B\cap C=P$. Then Proposition~\ref{v} gives 
\[\sigma_4^\pi(A)\cap\sigma_4^\pi(B)\cap\sigma_4^\pi(C)\neq\emptyset\]
\vspace{-1ex}

\noindent This set is equal to $Q\cap R$, and since it is a non-empty intersection of two blocks of $\beta(\pi)$, we must have $Q=R$ as required. 

It remains to consider the case when $A,B,C,D$ are all distinct. Since $A\cap B = P$ and $C\cap D=P$ and all four sets are distinct, then $A\cap C=P$. We can apply the case proved above twice, once to the situation that $A\cap B = P = A\cap C$, and again to the situation that $C\cap A = P = C\cap D$. These yield 

\[\sigma_4^\pi(A)\cap\sigma_4^\pi(B) = \sigma_4^\pi(A)\cap\sigma_4^\pi(C)\quad\mbox{ and }\quad \sigma_4^\pi(C)\cap\sigma_4^\pi(A)=\sigma_4^\pi(C)\cap\sigma_4^\pi(D)\]
\vspace{-1ex}

\noindent So $\sigma_4^\pi(A)\cap\sigma_4^\pi(B) = \sigma_4^\pi(C)\cap\sigma_4^\pi(D)$ as required. 
\end{proof}

We next move from bijections between certain 4-element subsets of $X$ to bijections between certain 2-element subsets of $X$. 

\begin{prop}
For any 2-relation $\pi$ there are bijections 
\vspace{-1ex}

\begin{align*}
\sigma_2^\pi&:\mc{P}(2,\pi)\to\mc{P}(2,\beta(\pi))\\
\mu^\pi_2&:\mc{P}(2,\pi)\to\mc{P}(2,\beta^{-1}(\pi))
\end{align*}
\vspace{-1ex}

\noindent with $\sigma^\pi_2$ and $\mu^{\beta(\pi)}_2$ inverses. These maps satisfy for any $P\in\mc{P}(2,\pi)$ and $A,B\in\mc{P}(4,\pi)$
\vspace{1ex}

\begin{enumerate}
\item $A\cap B = P\,\,$ iff $\,\,\sigma^\pi_4(A)\cap\sigma^\pi_4(B)=\sigma^\pi_2(P)$\vspace{1ex}
\item $A\cap B = P\,\,$ iff $\,\,\mu^\pi_4(A)\cap\mu^\pi_4(B)=\mu^\pi_2(P)$
\end{enumerate}
\label{w}
\end{prop}

\begin{proof}
For any $P\in\mc{P}(2,\pi)$ there are sets $A,B\in\mc{P}(4,\pi)$ with $A\cap B=P$. So we may define $\sigma^\pi_2$ by setting 
\[\sigma^\pi_2(P)=\sigma^\pi_4(A)\cap \sigma^\pi_4(B)\quad\mbox{ if }A,B\in\mc{P}(4,\pi) \mbox{ and }A\cap B = P \]
\vspace{-1ex}

\noindent Proposition~\ref{M} shows that this map is well defined, and it therefore satisfies condition~(1). Identical reasoning applied to the size preserving automorphism $\beta^{-1}$ of $\Eq*(X)$ provides a map $\mu^\pi_2$ that satisfies~(2). Then conditions~(1) and (2) imply that $\sigma^\pi_2$ and $\mu^{\beta(\pi)}_2$ are inverses of one another, and so both are bijections. 
\end{proof}

\begin{prop}
If $A$ is the union of distinct blocks $P$ and $Q$ of $\pi$. Then 

\[\sigma_4^\pi(A)=\sigma_2^\pi(P)\cup\sigma_2^\pi(Q)\]
\vspace{-2ex}

\noindent A similar statement holds for $\mu^\pi_4$.
\label{x}
\end{prop}

\begin{proof}
We know that $\sigma_4^\pi(A)$ is a 4-element set that is the union of two blocks of $\beta(\pi)$. Since $P$ is a 2-element block of $\pi$ that is contained in $A$ there is a 4-element set $B$ that is a union of two blocks of $\pi$ with $A\cap B=P$. By the definition of $\sigma_2^\pi$ we have that $\sigma_2^\pi(P)=\sigma_4^\pi(A)\cap\sigma_4^\pi(B)$, so $\sigma_2^\pi(P)\subseteq\sigma_4^\pi(A)$. Similarly $\sigma_2^\pi(Q)\subseteq\sigma_4^\pi(A)$. Since $\sigma_2^\pi$ is a bijection, we have $\sigma_2^\pi(P)\neq\sigma_2^\pi(Q)$. Since $\sigma^\pi_2(P)$ and $\sigma^\pi_2(Q)$ are distinct blocks of $\beta(\pi)$ that are contained in $\sigma^\pi_4(A)$, we therefore have that $\sigma_4^\pi(A)=\sigma_2^\pi(P)\cup\sigma_2^\pi(Q)$. Similar reasoning applies to $\mu^\pi_2$. 
\end{proof}

Since we can now reconstruct $\sigma_4^\pi$ from $\sigma_2^\pi$, it is not worthwhile keeping both maps. In fact, we define a new map that works on any subset of $X$ that is the union of blocks of $\pi$. Both of the maps $\sigma^\pi_4$ and $\sigma^\pi_2$ are restrictions of this more general map. 

\begin{defn}
For a 2-relation $\pi$ define 
\vspace{-1ex}

\begin{align*}
\sigma^\pi&:\mc{P}(\pi)\to\mc{P}(\beta(\pi))\\
\mu^\pi&:\mc{P}(\pi)\to\mc{P}(\beta^{-1}(\pi))
\end{align*}
\vspace{-2ex}

\noindent by setting $\sigma^\pi(A) = \bigcup\{\sigma^\pi_2(P):P\mbox{ is a block of $\pi$ with $P\subseteq A$}\}$, and similarly for $\mu^\pi(A)$. 
\label{y}
\end{defn}

\section{Patching together bijections between sets of subsets of $X$}

In this section we combine the maps $\sigma^\pi$ for different 2-relations $\pi$ and patch them together to produce a an automorphism of the 2-element subsets of $X$. To do this requires a closer study of the relationship between 2-relations. 

\pagebreak[3]

\begin{defn}
Let $\pi$ and $\lambda$ be 2-relations. We say 
\vspace{1ex}

\begin{enumerate}
\item $\pi$ and $\lambda$ overlap in 4-sets if they have at least one 4-relation upper bound
\item $\pi$ and $\lambda$ have full overlap in 4-sets if they have exactly one 4-relation upper bound
\item $\pi$ and $\lambda$ have non-full overlap in 4-sets if they have at least two 4-relation upper bounds
\item $\pi$ and $\lambda$ have small overlap in 4-sets if they have infinitely many 4-relation upper bounds
\end{enumerate}
\vspace{1ex}

\noindent If $\pi$ and $\lambda$ are overlap in 4-sets, a 4-element set $A$ is an overlap of $\pi$ and $\lambda$ if $A$ is the union of two blocks of $\pi$, $A$ is the union of two blocks of $\lambda$, and if $\pi$ and $\lambda$ do not agree on $A$. 
\label{aa}
\end{defn}

If $\beta$ is a size preserving automorphism, then a 4-relation $\gamma$ is an upper bound of two \mbox{2-relations} $\pi$ and $\lambda$ iff the 4-relation $\beta(\gamma)$ is a an upper bound of the 2-relations $\beta(\pi)$ and $\beta(\lambda)$. The following is then clear. 

\begin{prop}
If $\beta$ is a size preserving automorphism of $\Eq*(X)$ and $\pi$, $\lambda$ are 2-relations, then $\pi$ and $\lambda$ overlap in 4-sets iff $\beta(\pi)$ and $\beta(\lambda)$ overlap in 4-sets, and similarly for full overlap, non-full overlap, and small overlap. 
\label{zinkle}
\end{prop}

We next provide alternate characterizations of these notions that are easier to manipulate. Note that if $\pi$ and $\lambda$ are 2-relations, if $E=\{x,y\}$ is a block of $\pi$ that is not a block of $\lambda$, then there is at most one block $E'$ of $\pi$ with $E\cup E'$ being equal to the union of two blocks of $\lambda$. The only candidate for $E'$ is $E'=\{x',y'\}$ where $x'$ and $y'$ are the unique elements of $X$ with $\{x,x'\}$ and $\{y,y'\}$ blocks of $\lambda$. Of course, there is no guarantee that this $E'$ will be a block of $\pi$. 

\begin{prop}
The 2-relations $\pi$ and $\lambda$ overlap in 4-sets iff the following conditions are satisfied. 
\vspace{1ex}

\begin{enumerate}
\item For each block $E$ of $\pi$ that is not a block of $\lambda$, there is a block $E'$ of $\pi$ with $E\cup E'$ the union of two blocks of $\lambda$. This $E'$ is necessarily unique, and is also not a block of $\lambda$. \vspace{1ex}
\item There are an infinite number, or even finite number, of blocks of $\pi$ that are blocks of $\lambda$. 
\end{enumerate}
\vspace{1ex}

\noindent When these conditions are satisfied, the 4-relations $\gamma$ that are upper bounds of $\pi$ and $\lambda$ correspond to pairings of the blocks of $\pi$ that are blocks of $\lambda$. For such a pairing construct blocks of $\gamma$ from these pairs and take the $E\cup E'$ where $E$ is a block of $\pi$ that is not a block of $\lambda$ as the remaining blocks of $\gamma$. 
\label{bbbb}
\end{prop}

\begin{proof}
If $\pi$ and $\lambda$ overlap in 4-sets, then there is a 4-relation $\gamma$ that is an upper bound of both. This naturally provides a pairing of all blocks of $\pi$ with two blocks paired if they belong to the same block of $\gamma$. Under this pairing, a block $E$ of $\pi$ is a block of $\lambda$ iff its pair is a block of $\lambda$. This shows that conditions (1) and (2) hold. For the converse, suppose that conditions (1) and (2) hold. Condition (2) shows that the blocks of $\pi$ that are also blocks of $\lambda$ may be put into pairs. For this pairing, construct $\gamma$ as described in the statement of the result. This is obviously an upper bound of $\pi$ and $\lambda$.
\end{proof}

The following is obvious from Definition~\ref{aa} and Proposition~\ref{bbbb}.

\begin{prop}
Let $\pi$ and $\lambda$ be 2-relations that overlap in 4-sets. Then 
\vspace{1ex}

\begin{enumerate}
\item $\pi$ and $\lambda$ have full overlap in 4-sets iff they have at most two blocks in common
\item $\pi$ and $\lambda$ have non-full overlap in 4-sets iff they have at least four blocks in common
\item $\pi$ and $\lambda$ have small overlap in 4-sets iff they have infinitely many blocks in common
\end{enumerate}
\label{aaaa}
\end{prop}

We next give an alternate description of the overlaps $A$ of $\pi$ and $\lambda$ from Definition~\ref{aa}. 

\begin{prop}
Let $\pi$ and $\lambda$ be 2-relations that overlap in 4-sets and let $A$ be a 4-element set. Then (1) implies (2).  
\vspace{1ex}
\begin{enumerate}
\item $A$ is an overlap of $\pi$ and $\lambda$
\item every 4-relation that contains $\pi$ and $\lambda$ has $A$ as a block. 
\end{enumerate}
\vspace{1ex}

\noindent Further, if $\pi$ and $\gamma$ have non-full overlap, then (2) implies (1). 
\label{bb}
\end{prop}

\begin{proof}
(1) $\Rightarrow$ (2) Let $\gamma$ be a 4-relation that contains $\pi$ and $\lambda$. By the definition of an overlap, $A$ is the union of two blocks of $\pi$, $A$ is the union of two blocks of $\lambda$, and $\pi$ and $\lambda$ do not~agree~on~$A$. So $A=\{a,b,c,d\}$ with $\{a,b\},\{c,d\}$ are blocks of $\pi$ and $\{a,c\},\{b,d\}$ are blocks of $\lambda$. Since $\pi$ and $\lambda$ are contained in $\gamma$, it must be that $A$ is a block of $\gamma$. 

We now show that (2) $\Rightarrow$ (1) under the assumption that $\pi$ and $\lambda$ have non-full overlap. Since $\pi$ and $\lambda$ overlap in 4-sets they have a 4-relation $\gamma$ as an upper bound. Then by (2)  $A$ is a block of $\gamma$. As is the case with every block of $\gamma$, $A$ is the union of two blocks $P,Q$ of $\pi$ and the union of two blocks of $\lambda$. It remains only to show that $\pi$ and $\lambda$ do not agree on $A$, or equivalently that $P$ and $Q$ are not blocks of $\lambda$. Since $\pi$ and $\lambda$ have non-full overlap, by Proposition~\ref{aaaa} they have at least four blocks in common. If $P$ and $Q$ were two of them, then there would be a pairing of the common blocks of $\pi$ and $\lambda$ that does not pair $P$ and $Q$. By Proposition~\ref{bbbb} this pairing would produce an upper bound $\gamma'$ of $\pi$ and $\lambda$ that did not have $P$ and $Q$ contained in the same block, hence does not have $A$ as a block. 
\end{proof}

Now we come to the key moment when we can start comparing maps $\sigma^\pi$ and $\sigma^\lambda$ built from a size preserving automorphism $\beta$ of $\Eq*(X)$ using 2-relations $\pi$ and $\lambda$. Similar results hold for the maps $\mu^\pi$ and $\mu^\lambda$ built from $\beta^{-1}$, but we do not state them. 

\begin{prop}
Suppose that $\pi$ and $\lambda$ have non-full overlap in 4-sets. Then for a 4-element set $A$ these are equivalent. 
\vspace{1ex}
\begin{enumerate}
\item $A$ is an overlap of $\pi$ and $\lambda$
\item $\sigma^\pi(A)$ is an overlap of $\beta(\pi)$ and $\beta(\lambda)$
\item $\sigma^\lambda(A)$ is an overlap of $\beta(\pi)$ and $\beta(\lambda)$
\end{enumerate}
\label{ff}
\end{prop}

\begin{proof}
(1) $\Rightarrow$ (2) Proposition~\ref{zinkle} gives that $\beta(\pi)$ and $\beta(\lambda)$ have non-full overlap in 4-sets. Let $\gamma'$ be an upper bound of $\beta(\pi)$ and $\beta(\lambda)$. Since $\beta$ is a size preserving automorphism, $\gamma'=\beta(\gamma)$ for some 4-relation $\gamma$ that is an upper bound of $\pi$ and $\lambda$. Since $A$ is an overlap of $\pi$ and $\lambda$, $A$ is a block of $\gamma$. So $\gamma$ contains the two relations that agree with $\pi$ except on $A$. Let $\theta$ be one of these. Then $\gamma'$ contains $\beta(\pi)$ and $\beta(\theta)$. By Proposition~\ref{s}, $\beta(\pi)$ and $\beta(\theta)$ agree except on $\sigma^\pi(A)$. So $\sigma^\pi(A)$ is a block of $\gamma'$. Since every 4-relation that contains $\beta(\pi)$ and $\beta(\lambda)$ has $\sigma^\pi(A)$ as a block, by Proposition~\ref{bb} $\sigma^\pi(A)$ is an overlap of $\beta(\pi)$ and $\beta(\lambda)$. 

(2) $\Rightarrow$ (1) This follows using the corresponding version of (1) $\Rightarrow$ (2) for $\beta^{-1}$ and $\mu^\pi$. Thus (1) and (2) are equivalent. The equivalence of (1) and (3) is similar. 
\end{proof}

Note that this result has not shown that $\sigma^\pi(A)$ and $\sigma^\lambda(A)$ are equal. We next begin a series successively stronger results that do give this, and in a more general setting than in the result above. The first step is the following. 

\begin{prop}
Suppose that $\pi$ and $\lambda$ are 2-relations that have small overlap in 4-sets and that $A$ is the union of two blocks of $\pi$ that are also blocks of $\lambda$. Then $\sigma^\pi(A)=\sigma^\lambda(A)$. 
\label{hhh}
\end{prop}

\begin{proof}
Suppose that the overlaps of $\pi$ and $\lambda$ are $K_i$ $(i\in I)$. Since $\pi$ and $\lambda$ agree on $A$, we have that $A$ is not an overlap. Let $\theta_1,\theta_2$ be the 2-relations that agree with $\pi$ except on $A$, and let $\phi_1,\phi_2$ be the 2-relations that agree with $\lambda$ except on $A$. The restrictions of $\theta_1,\theta_2$ to $A$ partition $A$ into two blocks of two in the two possible ways that differ from $\pi$, and the same is true of the restrictions of $\phi_1,\phi_2$. We assume that the numbering of $\theta_1,\theta_2$ and $\phi_1,\phi_2$ is such that $\theta_1,\phi_1$ agree on $A$. 
Since $\pi$ and $\lambda$ have small overlap, by Proposition~\ref{aaaa} they have infinitely many blocks in common. The same is true of any pair from $\theta_1,\theta_2,\pi,\phi_1,\phi_2,\lambda$, so by Proposition~\ref{aaaa} any pair of these relations also has small overlap in 4-sets. 

The overlaps of $\pi$ and $\lambda$ are $K_i$ $(i\in I)$. By Proposition~\ref{ff} the overlaps of $\beta(\pi)$ and $\beta(\lambda)$ are $\{\sigma^\pi(K_i):i\in I\}$, and the overlaps of $\beta(\pi)$ and $\beta(\lambda)$ are $\{\sigma^\lambda(K_i):i\in I\}$. Thus 
\vspace{-.75ex}

\[\{\sigma^\pi(K_i):i\in I\} = \{\sigma^\lambda(K_i):i\in I\}\]
\vspace{-1ex}

\noindent Since $\phi_1$ agrees with $\lambda$ on $X\setminus A$ and differs from $\lambda$ on $X\setminus A$, the overlaps of $\pi$ and $\phi_1$ are $\{K_i:i\in I\}\cup\{A\}$. Applying Proposition~\ref{ff} with similar reasoning to that above, we have 
\vspace{-.75ex}

\[\{\sigma^\pi(K_i):i\in I\}\cup\{\sigma^\pi(A)\} = \{\sigma^{\phi_1}(K_i):i\in I\}\cup\{\sigma^{\phi_1}(A)\}\]
\vspace{-1ex}

\noindent Combining these we have 
\vspace{-.75ex}

\[\{\sigma^\lambda(K_i):i\in I\}\cup\{\sigma^\pi(A)\}=\{\sigma^{\phi_1}(K_i):i\in I\}\cup\{\sigma^{\phi_1}(A)\}\]
\vspace{-1ex}

\noindent Apply the same argument involving Proposition~\ref{ff} again, this time to $\phi_1$ and $\lambda$. These overlap in exactly $A$. So we have that $\sigma^{\phi_1}(A)=\sigma^\lambda(A)$. Therefore $\sigma^\lambda(A)$ occurs on the right side of the expression above, so it must also occur on the left side of this expression. We know that $A\neq K_i$ for each $i\in I$, and since $\sigma^\lambda$ is a bijection from $\mc{P}(4,\lambda)$ to $\mc{P}(4,\beta(\lambda))$ then $\sigma^\lambda(A)\neq\sigma^\lambda(K_i)$ for each $i\in I$. Thus $\sigma^\lambda(A)=\sigma^\pi(A)$ as required. 
\end{proof}

We next extend the previous result. 

\begin{prop}
Suppose that $\pi$ and $\lambda$ are 2-relations that have small overlap in 4-sets and that $A$ is the union of two blocks of $\pi$ and the union of two blocks of $\lambda$. Then $\sigma^\pi(A)=\sigma^\lambda(A)$. 
\label{ii}
\end{prop}

\begin{proof}
If $\pi$ and $\lambda$ agree on $A$, the result is given by Proposition~\ref{hhh}. Assume that $\pi$ and $\lambda$ differ on $A$. Let $\theta$ agree with $\pi$ on $X\setminus A$ and agree with $\lambda$ on $A$. Then $\pi$ and $\theta$ have small overlap in 4-sets and their unique overlap is $A$. By Proposition~\ref{ff} $\sigma^\pi(A)$ is the unique overlap of $\beta(\pi)$ and $\beta(\theta)$, and by Proposition~\ref{ff} $\sigma^\theta(A)$ is the unique overlap of $\beta(\pi)$ and $\beta(\theta)$. So $\sigma^\pi(A)=\sigma^\theta(A)$. 

Since $\pi$ and $\lambda$ have small overlap, by Proposition~\ref{aaaa} they have infinitely many blocks in common. Since $\theta$ agrees with $\pi$ except on $A$, then $\theta$ and $\lambda$ have infinitely many blocks in common, and since they overlap in 4-sets, Proposition~\ref{aaaa} shows that $\theta$ and $\lambda$ have small overlap in 4-sets. Since $A$ is the union of two blocks of $\theta$ that are also blocks of $\lambda$, Proposition~\ref{hhh} gives that $\sigma^\theta(A)=\sigma^\lambda(A)$. Thus $\sigma^\pi(A)=\sigma^\lambda(A)$ as required. 
\end{proof}

We next extend this result further by removing the assumption that the overlap of $\pi$ and $\lambda$ in 4-sets is small. 

\begin{prop}
Suppose that $\pi$ and $\lambda$ are 2-relations that overlap in 4-sets and $A$ is the union of two blocks of $\pi$ and $A$ is the union of two blocks of $\lambda$. Then $\sigma^\pi(A)=\sigma^\lambda(A)$. 
\label{jj}
\end{prop}

\begin{proof}
If $\pi$ and $\lambda$ have infinitely many blocks in common, then by Proposition~\ref{aaaa}, $\pi$ and $\lambda$ have small overlap in 4-sets, and the result is given by Proposition~\ref{ii}. Assume that $\pi$ and $\lambda$ have only finitely many blocks in common. Since $\pi$ and $\lambda$ overlap in 4-sets there is a 4-relation $\gamma$ that is an upper bound of both. Each block $F$ of $\gamma$ is the union of two blocks of $\pi$ and is the union of two blocks of $\lambda$. The block $F$ is an overlap of $\pi$ and $\lambda$ iff the two blocks of $\pi$ contained in it are not blocks of $\lambda$. It follows that $\pi$ and $\lambda$ must have infinitely many overlaps. Enumerate the set of overlaps of $\pi$ and $\lambda$ as $K_i$ $(i\in I)$ where $I$ is an infinite set. 

Partition $I$ into two infinite pieces $I=P\cup Q$. Define a new 2-relation $\theta$ so that $\theta$ agrees with $\pi$ on $X\setminus \bigcup_PK_i$ and $\theta$ agrees with $\lambda$ on $\bigcup_PK_i$. Since each block of $\theta$ is either a block of $\pi$ or a block of $\lambda$, the 4-relation $\gamma$ that is an upper bound of $\pi$ and $\lambda$ also contains $\theta$. So $\pi$ and $\theta$ overlap in 4-sets, and $\theta$ and $\lambda$ overlap in 4-sets. 

Since $A$ is the union of two blocks of $\pi$ and the union of two blocks of $\lambda$, either $A$ is an overlap of $\pi$ and $\lambda$ or $A$ is disjoint from each overlap of $\pi$ and $\lambda$. It follows that $A$ is the union of two blocks of $\theta$. Since $\pi$ and $\theta$ agree on $X\setminus\bigcup_PK_i$, which contains the infinite set $\bigcup_QK_i$, it follows that $\pi$ and $\theta$ have infinitely many blocks in common. Since $\theta$ and $\lambda$ agree on the infinite set $\bigcup_PK_i$, they have infinitely many blocks in common. So by Proposition~\ref{aaaa}, $\pi$ and $\theta$ have small overlap in 4-sets, and $\theta$ and $\lambda$ have small overlap in 4-sets. Proposition~\ref{ii} then gives that $\sigma^\pi(A)=\sigma^\theta(A)$ and $\sigma^\theta(A)=\sigma^\lambda(A)$. Thus $\sigma^\pi(A)=\sigma^\lambda(A)$ as required. 
\end{proof}

Our final step in this section is to remove all aspects of the overlap condition. For this, we require some preliminary results. 

\begin{defn}
Let $\phi$ and $\rho$ be 2-relations on a set $Y$. For a natural number $n\geq 1$ we say an indexed family $a_i,b_i$ $0\leq i<n$ of distinct elements is an $n$-cycle if with addition modulo $n$, for each $0\leq i<n$ \vspace{-2ex}
\begin{align*}
\{a_i,b_i\}\,\,\,&\mbox{ is a block of }\phi\\
\{b_i,a_{i+1}\}&\mbox{ is a block of }\rho
\end{align*}
\vspace{-2ex}

\noindent We say that an indexed family $a_i,b_i$ $(i\in\mathbb{Z})$ is a $\mathbb{Z}$-cycle if for each $i\in\mathbb{Z}$ we have that $\{a_i,b_i\}$ is a block of $\phi$ and $\{b_i,a_{i+1}\}$ is a block of $\rho$. 
\end{defn}

\begin{prop}
Let $\phi$ and $\rho$ be 2-relations on $Y$ and let $\theta$ be the transitive closure of $\phi\cup\rho$. Then each block of $\theta$ can be indexed to form an $n$-cycle for some $n\geq 1$ or to form a $\mathbb{Z}$-cycle. 
\label{toe}
\end{prop}

\begin{proof}
Let $E$ be a block of $\theta$ and $a_0$ be an element of $E$. Since $\phi$ is a 2-relation, there is a unique element distinct from $a_0$ that is $\phi\,$-related to $a_0$, we call this $b_0$. There is a unique element distinct from $b_0$ that is $\rho\,$-related to $b_0$. If this is $a_0$, then $E=\{a_0,b_0\}$ is a 2-element block of $\phi$ and~$\rho$, hence a 1-cycle. Otherwise let $a_1$ be the unique element distinct from $b_0$ that is $\rho\,$-related to~$b_0$. Then $a_0,b_0,a_1$ are distinct. Now there is a unique element distinct from $a_1$ that is $\phi\,$-related to $a_1$. This cannot be $a_0$ or $b_0$ since they are distinct and $\phi\,$-related to each other. So it must be some element $b_1$ making all of $a_0,b_0,a_1,b_1$ distinct. Then there is a unique element distinct from $b_1$ that is $\rho\,$-related to~$b_1$. This cannot be either of $b_0,a_1$ since these elements already have an element different from themselves and $b_1$ that is $\rho\,$-related to them. Either $a_0$ is $\rho\,$-related to $b_1$, so $\{a_0,b_0,a_1,b_1\}$ form a 2-cycle, or there is some element $a_2$ distinct from $a_0,b_0,a_1,b_1$ that is $\rho\,$-related to $b_1$. Then there is a unique element distinct from $a_2$ that is $\phi\,$-related to $a_2$. This cannot be any of $a_0,b_0,a_1,b_1$ since they already have an element other than themselves and distinct from $a_2$ that is $\phi\,$-related to them. So there is some $b_2$ making $a_0,b_0,a_1,b_1,a_2,b_2$ distinct and with $b_2$ $\phi\,$-related to $a_2$. 

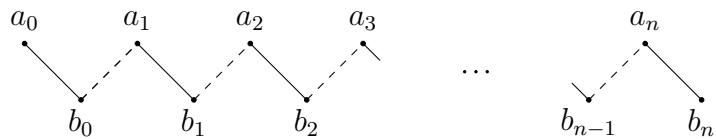
\vspace{1ex}
\begin{center} 
\begin{tikzpicture}[scale = .75]
\draw [fill] (0,1) circle [radius=.04]; \draw [fill] (1,0) circle [radius=.04]; \draw [fill] (2,1) circle [radius=.04]; \draw [fill] (3,0) circle [radius=.04]; 
\draw [fill] (4,1) circle [radius=.04]; \draw [fill] (5,0) circle [radius=.04]; \draw [fill] (6,1) circle [radius=.04]; 
\draw [fill] (10,0) circle [radius=.04]; \draw [fill] (11,1) circle [radius=.04]; \draw [fill] (12,0) circle [radius=.04];
\node at (0,1.4) {$a_0$}; \node at (2,1.4) {$a_1$}; \node at (4,1.4) {$a_2$}; \node at (6,1.4) {$a_3$}; \node at (11,1.4) {$a_n$};
\node at (1,-.4) {$b_0$}; \node at (3,-.4) {$b_1$}; \node at (5,-.4) {$b_2$}; \node at (10,-.4) {$b_{n-1}$}; \node at (12,-.4) {$b_n$};
\draw (0,1) -- (1,0); \draw (2,1) -- (3,0); \draw (4,1) -- (5,0); \draw (11,1) -- (12,0); \draw (6,1) -- (6.3,.7); \draw (9.7,.3) -- (10,0); 
\draw[dashed] (1,0) -- (2,1); \draw[dashed] (3,0) -- (4,1); \draw[dashed] (5,0) -- (6,1); \draw[dashed] (1,0) -- (2,1); \draw[dashed] (10,0) -- (11,1); 
\draw [fill] (7.8,.5) circle [radius=.02]; \draw [fill] (8,.5) circle [radius=.02]; \draw [fill] (8.2,.5) circle [radius=.02];
\end{tikzpicture}
\end{center}

Continuing, we get to a point where we have $a_0,b_0,\ldots,a_n,b_n$ distinct with $a_j$ $\phi\,$-related to $b_j$ for each $0\leq j\leq n$ and $b_j$ $\rho\,$-related to $a_{j+1}$ for each $0\leq j<n$. Then there is a unique element distinct from $b_n$ that is $\rho\,$-related to $b_n$. This cannot be any of $b_0,a_1,b_1,\ldots,a_n$ since these all have an element distinct from themselves and $b_n$ that is $\rho\,$-related to them. If $a_0$ is $\rho\,$-related to $b_n$, then $E=\{a_0,b_0,\ldots,a_n,b_n\}$ is an $n$-cycle. Otherwise, there is some element $a_{n+1}$ distinct from all of $a_0,b_0,\ldots,a_n,b_n$ with $a_{n+1}$ $\rho\,$-related to $b_n$. Then there is a unique element distinct from $a_{n+1}$ that is $\phi\,$-related to it, and this element cannot be any of $a_0,b_0,\ldots,a_n,b_n$ since these are all $\phi\,$-related to an element distinct from themselves and $a_{n+1}$. So there is $b_{n+1}$ that is $\phi$ related to $a_{n+1}$ making all of $a_0,b_0,\ldots,a_{n+1},b_{n+1}$ distinct. Then by an obvious induction we may continue. 

So either $E$ is an $n$-cycle for some finite $n\geq 1$, or this process continues indefinitely to produce $a_n,b_n$ for each $n\geq 0$ with all distinct, $a_n$ $\phi\,$-related to $b_n$, and $b_n$ $\rho\,$-related to $a_{n+1}$. This however is not a block as there can be no element in this collection that is distinct from $a_0$ and is $\rho\,$-related to $a_0$. Let $b_{-1}$ be the unique such element. Then there is a unique element distinct from $b_{-1}$ that is $\phi\,$-related to it, and it cannot be any of the elements we have so far obtained since they are all $\phi\,$-related to an element distinct from themselves and $b_{-1}$, hence there must be a new element $a_{-1}$. This process must now continue indefinitely through the negative integers. To see this, suppose we have $a_{-n},b_{-n},a_{-n+1},b_{-n+1},\ldots$. Now each element except $a_{-n}$ is $\rho\,$-related to an element in this list distinct from itself. So $a_{-n}$ is $\rho\,$-related to a new element $b_{-n-1}$. Then $b_{-n-1}$ is $\phi\,$-related to some element not in this extended list, and we call this $a_{-n-1}$, and this process continues. This yields $E=\{a_n,b_n:n\in\mathbb{Z}\}$, and this is a $\mathbb{Z}$-cycle. 
\end{proof}

\begin{defn}
Let $Y=A\cup B$ a partition of a set $Y$. A 2-relation $\phi$ on $Y$ is a matching between $A$ and $B$ if each block of $\phi$ contains one element of $A$ and one element of $B$. 
\end{defn}

\begin{prop}
Let $\phi$ and $\rho$ be 2-relations on $Y$. Then $Y$ can be partitioned $Y=A\cup B$ so that both $\phi$ and $\rho$ are matchings between $A$ and $B$. 
\label{gffr}
\end{prop}

\begin{proof}
Let $\theta$ be the transitive closure of $\phi\cup\rho$ and let $E_i$ $(i\in I)$ be an enumeration of the blocks of $\theta$. It may be that $\theta$ has just one block, or that the cardinality of the set of blocks of $\theta$ equals that of $Y$. It is not of importance. Let $i\in I$. By Proposition~\ref{toe} the elements of $E_i$ can be indexed to form either an $n$-cycle for some $n\geq 1$ or a $\mathbb{Z}$-cycle. In either case, there is a partition $E_i=A_i\cup B_i$ so that each block of $\phi$ that is contained in $E_i$ contains one element of $A_i$ and one element of $B_i$, and each element of $\rho$ contains one element of $A_i$ and one element of $B_i$. Set $A=\bigcup_IA_i$ and $B=\bigcup_IB_i$. 
\end{proof}

\begin{prop}
If $\phi$ and $\rho$ are 2-relations on $Y$, then there is a finite sequence $\nu_0,\ldots,\nu_n$ of 2-relations with $\nu_0=\phi$, $\nu_n=\rho$ where $\nu_i$ and $\nu_{i+1}$ overlap in 4-sets for each $i=0,\ldots,n-1$. 
\label{lizzie}
\end{prop}

\begin{proof}
Using Proposition~\ref{gffr} there is a partition $Y=A\cup B$ so that $\phi$ and $\rho$ are matchings between $A$ and $B$. For each matching $\lambda$ between $A$ and $B$ and each permutation $\sigma$ of $B$, define the 2-relation $\sigma(\lambda)$ so that 

\[\sigma(\lambda) \mbox{ has as blocks } \{a,\sigma(b)\} \mbox{ where }\{a,b\}\mbox{ is a block of $\lambda$ with $a\in A$ and $b\in B$}\]
\vspace{-1ex}

We now show that if $\lambda$ is a matching between $A$ and $B$ and $\delta$ is an involution of $B$ that has infinitely many fixed points, then $\lambda$ and $\delta(\lambda)$ overlap in 4-sets. Let $P$ be the set of points of $B$ that are fixed points of $\delta$, and $Q$ be the set of points of $B$ that are not fixed points of $\delta$. Note that $b\in Q$ iff $\delta(b)\in P$. Define a 4-relation $\gamma$ so that 

\[ \{a,b,c,\delta(b)\} \mbox{ is a block of $\gamma$ if $b\in Q$ and $\{a,b\},\{c,\delta(b)\}$ are blocks of $\lambda$}\]
\vspace{-1ex}

\noindent This gives a collection of 4-element subsets of $Y$, each consisting of two blocks of $\lambda$, with these 4-element sets pairwise disjoint. The union of these 4-element sets is not all $Y$ since it does not contain any elements of $P$. To extend this, choose an arbitrary pairing of the infinitely many blocks of $\lambda$ that contain a fixed point $b\in P$ and take as the remaining blocks of $\gamma$ the unions of these pairs of blocks. This gives a 4-relation $\gamma$ that by definition contains $\lambda$. If $b\in P$ is a fixed point of $\delta$, the block of $\delta(\lambda)$ that contains $b$ is equal to the block of $\lambda$ that contains $b$, so is contained in $\gamma$. Suppose $b\in Q$. Let $c\in A$ be such that $\{c,\delta(b)\}$ is a block of $\lambda$. Since $\delta(\delta(b))=b$, the block of $\delta(\lambda)$ that contains $b$ is $\{c,b\}$, and this is contained in a block of $\gamma$. Therefore $\gamma$ is an upper bound of $\lambda$ and $\delta(\lambda)$, so $\lambda$ and $\delta(\lambda)$ overlap in 4-sets. 

As we have seen in the proof of Theorem~\ref{tiger}, every permutation $\sigma\in\Perm(B)$ is obtained as the composition of finitely many involutions with each having infinitely many fixed points. Since $\phi$ and $\rho$ are both matchings between $A$ and $B$, there is a permutation $\sigma$ with $\sigma(\phi)=\rho$. Take a finite sequence $\delta_1,\ldots,\delta_n$ of involutions of $B$, with each having infinitely many fixed points, where $\sigma=\delta_n\circ\cdots\circ\delta_1$. Define recursively $\nu_0=\phi$ and $\nu_{i+1}=\delta_{i+1}(\nu_i)$. Then for each $i=0,\ldots,n-1$ we have that $\nu_i$ and $\nu_{i+1}$ overlap in 4-sets, and $\nu_n=(\delta_n\circ\cdots\circ\delta_1)(\phi) = \sigma(\phi)=\rho$. 
\end{proof}

We are now prepared to remove the overlap conditions entirely from our string of results. 

\begin{prop}
Suppose that $\pi$ and $\lambda$ are 2-relations and that $A$ is a 4-element set that is the union of two blocks of $\pi$ and is the union of two blocks of $\lambda$. Then $\sigma^\pi(A)=\sigma^\lambda(A)$. 
\label{jjj}
\end{prop}

\begin{proof}
Let $Y=X\setminus A$, let $\phi$ be the restriction of $\pi$ to $Y$ and $\rho$ be the restriction of $\lambda$ to $Y$. By Proposition~\ref{lizzie} there is a finite sequence $\nu_0,\ldots,\nu_n$ of 2-relations on $Y$ so that $\nu_0=\phi$, $\nu_n=\rho$ where $\nu_i$ and $\nu_{i+1}$ have overlap in 4-sets for each $i=0,\ldots,n-1$. For $i=0,\ldots,n-1$ define a relation $\mu_i$ on $X$ to agree with $\nu_i$ on $Y$ and to agree with $\pi$ on $A$. Define $\mu_n$ to agree with $\nu_n$ on $Y$ and with $\lambda$ on $A$. Note that $\mu_0=\pi$ and $\mu_n=\lambda$.
Then for each $i=0,\ldots,n-1$ we have (1) $\mu_i$ and $\mu_{i+1}$ have overlap in 4-sets and (2) $A$ is the union of two blocks of $\mu_i$ and $A$ is the union of two blocks of $\mu_{i+1}$. Then by Proposition~\ref{jj} we have that $\sigma^{\mu_i}(A) = \sigma^{\mu_{i+1}}(A)$ for each $i=0,\ldots,n-1$. It follows that $\sigma^\pi(A)=\sigma^{\mu_0}(A)=\sigma^{\mu_n}(A)=\sigma^\lambda(A)$ as required. 
\end{proof}

There remains a final task in our series of ever stronger results. We wish to extend matters from 4-element sets to 2-element sets. This is provided by the following. 

\begin{prop}
Suppose that $\pi$ and $\lambda$ are 2-relations and that $P$ is a 2-element set that is a block of $\pi$ and is a block of $\lambda$. Then $\sigma^\pi(P)=\sigma^\lambda(P)$. 
\label{ll}
\end{prop}

\begin{proof}
Let $Q$ and $R$ be two blocks of $\pi$ distinct from each other and from $P$. Let $A=P\cup Q$ and $B=P\cup R$. Then by Proposition~\ref{w}

\[\sigma^\pi(P)=\sigma^\pi(A)\cap\sigma^\pi(B)\]
\vspace{-1ex}

It is given that $P$ is a block of $\lambda$. Take two blocks $S$ and $T$ of $\lambda$ that are distinct from~$P$, distinct from each other, and disjoint from $Q$ and $R$. Set $C=P\cup S$ and $D=P\cup T$. Then by Proposition~\ref{w}
\[\sigma^\lambda(P)=\sigma^\lambda(C)\cap\sigma^\lambda(D)\]
\vspace{-1ex}

Construct a 2-relation $\theta$ that has as blocks $P,Q,R,S,T$ and whose behavior on the rest of $X$ is unimportant. Then $A,B,C,D$ are 4-element subsets of $\theta$ that are unions of blocks of $\theta$. So by Proposition~\ref{w}
\[\sigma^\theta(A)\cap\sigma^\theta(B)=\sigma^\theta(P)=\sigma^\theta(C)\cap\sigma^\theta(D)\]
\vspace{-1ex}

Proposition~\ref{jjj} gives $\sigma^\pi(A)=\sigma^\theta(A)$ and $\sigma^\pi(B)=\sigma^\theta(B)$, and also that $\sigma^\theta(C)=\sigma^\lambda(C)$ and $\sigma^\theta(D)=\sigma^\lambda(D)$. Therefore 
\vspace{-2ex}

\begin{align*}
\sigma^\pi(P)&\,\,=\,\,\sigma^\pi(A)\,\,\cap\,\,\sigma^\pi(B)\\
&\,\,=\,\,\sigma^\theta(A)\,\,\,\cap\,\,\,\sigma^\theta(B)\\
&\,\,=\,\,\sigma^\theta(C)\,\,\,\cap\,\,\sigma^\theta(D)\\
&\,\,=\,\,\sigma^\lambda(C)\,\,\cap\,\,\sigma^\lambda(D)\\
&\,\,=\,\,\sigma^\lambda(P)
\end{align*}
\noindent This concludes the proof of the result. 
\end{proof}

To conclude this section we note that all of the results obtained for the maps $\sigma^\pi$ for a 2-relation $\pi$ using the size preserving automorphism $\beta$ of $\Eq*(X)$ also hold for the maps $\mu^\pi$ obtained from $\pi$ using the size preserving automorphism $\beta^{-1}$. We state only the strongest of these results, the final one, in this setting. 

\begin{prop}
Suppose that $\pi$ and $\lambda$ are 2-relations and that $P$ is a 2-element set that is a block of $\pi$ and is a block of $\lambda$. Then $\mu^\pi(P)=\mu^\lambda(P)$. 
\label{ll*}
\end{prop}

\section{Constructing a permutation of $X$ and the proof of the main theorem}

In this section we produce a permutation of $X$ from the mappings $\sigma^\pi$ of the previous section. This is used to prove Theorem~\ref{todo}, which states that every size preserving automorphism of $\Eq*(X)$ is obtained from a permutation of $X$. As discussed in section~4, this establishes the proof of the main theorem, Theorem~\ref{main}. 

\begin{prop}
Let $P$ and $Q$ be two 2-element subsets of $X$. Then for any 2-relations $\pi$ and $\lambda$ with $\sigma^\pi(P)$ and $\sigma^\lambda(Q)$ defined we have the following. 
\vspace{1ex}
\begin{enumerate}
\item $P = Q$ iff $\sigma^\pi(P)=\sigma^\lambda(Q)$
\item $P$ and $Q$ are disjoint iff $\sigma^\pi(P)$ and $\sigma^\lambda(Q)$ are disjoint
\end{enumerate}
\label{nn}
\end{prop}

\begin{proof}
It is sufficient to prove the forward direction of each statement since the backwards direction corresponds to the forward direction applied to the inverse maps $\mu^{\beta(\pi)}$ and $\mu^{\beta(\lambda)}$. The forward direction of (1) is given by Proposition~\ref{ll}. For the forward direction of (2) suppose that $P$ and $Q$ are disjoint. Then there is a 2-relation $\theta$ with $P$ and $Q$ distinct blocks of~$\theta$. Since $\sigma^\theta$ is a bijection, $\sigma^\theta(P)$ and $\sigma^\theta(Q)$ are distinct blocks of $\beta(\theta)$, so they are disjoint. But by Proposition~\ref{ll}, $\sigma^\theta(P)=\sigma^\pi(P)$ and $\sigma^\theta(Q)=\sigma^\lambda(Q)$. 
\end{proof}

\begin{prop}
Let $P$ and $Q$ be 2-element subsets of $X$. Then for any 2-relations $\pi$ and $\lambda$ with $\sigma^\pi(P)$ and $\sigma^\lambda(Q)$ defined
\[ |\, P\cap Q\, |\,\, = \,\, |\, \sigma^\pi(P)\cap\sigma^\lambda(Q)\, |\]
\label{oo}
\vspace{-2ex}
\end{prop}

\begin{proof}
Since $P$ and $Q$ are both 2-element sets, they are either disjoint, have exactly one element in common, or are equal and have 2 elements in common. The same holds for $\sigma^\pi(P)$ and $\sigma^\lambda(Q)$. By Proposition~\ref{nn}, $P$ and $Q$ are equal iff $\sigma^\pi(P)$ and $\sigma^\lambda(Q)$ are equal, and $P$ and $Q$ are disjoint iff $\sigma^\pi(P)$ and $\sigma^\lambda(Q)$ are disjoint. So the remaining possibility, that $P$ and $Q$ have exactly one element in common, occurs iff $\sigma^\pi(P)$ and $\sigma^\lambda(Q)$ have exactly one element in common. 
\end{proof}

\begin{prop}
Suppose that $P,Q,P'$ and $Q'$ are 2-element sets and that $\pi,\lambda,\pi'$ and $\lambda'$ are 2-relations with $\sigma^\pi(P),\sigma^\lambda(Q),\sigma^{\pi'}(P')$ and $\sigma^{\lambda'}(Q')$ defined. Then 

\[ P\,\cap\, Q \,=\, P'\,\cap\, Q'\quad \mbox{ iff }\quad \sigma^\pi(P)\,\cap\,\sigma^\lambda(Q)\, = \, \sigma^{\pi'}(P')\,\cap\,\sigma^{\lambda'}(Q') \]
\vspace{-2ex}
\label{qq}
\end{prop}

\begin{proof}
We prove the forward direction of the statement. The backwards direction follows from the forward direction applied to the inverse maps $\mu^{\beta(\pi)},\mu^{\beta(\lambda)},\mu^{\beta(\pi')}$ and $\mu^{\beta(\lambda')}$. The forward direction follows from Proposition~\ref{nn} in the case that $P$ and $Q$ are disjoint, since $P'$ and $Q'$ must also be disjoint. The forward direction follows from Proposition~\ref{nn} if $P\cap Q$ has two elements, since then $P,Q,P',Q'$ are all equal. It remains to show the forward direction when $P$ and $Q$ intersect in a single element. 

Assume that 
\[P=\{x,y\}\quad \mbox{ and }\quad Q=\{x,z\}\]
\vspace{-2ex}

\noindent If $P=P'$ and $Q=Q'$, then by Proposition~\ref{ll} we have $\sigma^\pi(P)=\sigma^{\pi'}(P')$ and $\sigma^\lambda(Q)=\sigma^{\lambda'}(Q')$ and the result follows. The majority of our effort will go to establishing the next simplest case, that where $P=P'$ and $Q\neq Q'$. In this case we may assume that 
\[\,\, P'=\{x,y\}\quad \mbox{ and }\quad Q'=\{x,z'\}\]
\vspace{-2ex}

Note that our assumptions that $P$ and $Q$ are 2-element sets with $P\cap Q$ having one element gives that $x,y,z$ are distinct, and the same assumptions for $P'$ and $Q'$ give that $x,y,z'$ are distinct. Our assumption that $Q\neq Q'$ gives that $z\neq z'$. So $x,y,z,z'$ are distinct. Consider the following diagram that additionally shows an element $w$ distinct from $x,y,z,z'$. 

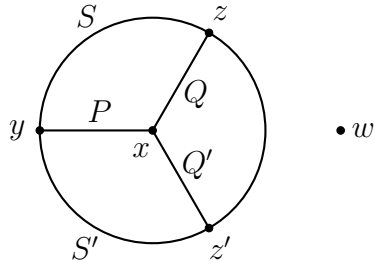
\vspace{1ex}
\begin{center}
\begin{tikzpicture}[scale = 1]
\draw [thick] (0,0) circle [radius=1.5];
\draw [fill] (0,0) circle [radius=.05];
\draw [fill] (180:1.5) circle [radius=0.05];
\draw [fill] (60:1.5) circle [radius=0.05];
\draw [fill] (-60:1.5) circle [radius=0.05];
\draw [thick] (0,0) -- (180:1.5);
\draw [thick] (0,0) -- (60:1.5);
\draw [thick] (0,0) -- (-60:1.5);
\node at (240:0.3) {$x$};
\node at (180:1.8) {$y$};
\node at (60:1.8) {$z$};
\node at (-60:1.8) {$z'$};
\node at (160:.75) {$P$};
\node at (40:.75) {$Q$};
\node at (-37:.75) {$Q'$};
\node at (120:1.75) {$S$};
\node at (240:1.8) {$S'$};
\draw [fill] (0:2.5) circle [radius=0.05];
\node at (0:2.8) {$w$};
\end{tikzpicture}
\end{center}
\vspace{1ex}

We write $\sigma(P)$ in place of $\sigma^\pi(P)$ and so forth since Proposition~\ref{ll} shows that the result does not depend on the choice of $\pi$. Since $P$ and $Q$ are 2-element sets that intersect in one element, the same is true of $\sigma(P)$ and $\sigma(Q)$ by Proposition~\ref{oo}. So we may assume that there are distinct $a,b,c$ in $X$ with 
\begin{align*}
\sigma(P)&=\{a,b\}\\ 
\sigma(Q)&=\{a,c\}
\end{align*}
\vspace{-1ex}

Note that $\sigma(P)\cap\sigma(Q)=\{a\}$. We now begin to argue by contradiction and assume that $\sigma(P)\cap\sigma(Q')\neq\{a\}$. Since $|P\cap Q'|=1$, by Proposition~\ref{oo} we have $|\sigma(P)\cap\sigma(Q')|=1$. This one element cannot be $a$, and must belong to $\sigma(P)=\{a,b\}$. Therefore $\sigma(P)\cap\sigma(Q')=\{b\}$. Since $|Q\cap Q'|=1$, by Proposition~\ref{oo} we have that $|\sigma(Q)\cap\sigma(Q')|=1$. It cannot be that $a\in\sigma(Q')$ since that would mean that $a$ belonged to $\sigma(P)\cap\sigma(Q')$. So the one element belonging to $\sigma(Q)\cap\sigma(Q')$ must be the other element of $\sigma(Q)$, the element $c$. Since we have shown that the distinct elements $b$ and $c$ belong to the 2-element set $\sigma(Q')$ we have 
\begin{align*}
\sigma(Q')&=\{b,c\}
\end{align*}
\vspace{-1ex}

Let $S=\{y,z\}$. Then there is a 2-relation $\theta$ with $S$ a block of $\theta$, so $\sigma^\theta(S)$ is defined, and as we have been doing, we write this as $\sigma(S)$. Since $S$ and $Q'$ are disjoint, by Proposition~\ref{oo} so are $\sigma(S)$ and $\sigma(Q')$. So neither $b,c$ is an element of $\sigma(S)$. Since $|S\cap P|=1$, by Proposition~\ref{oo} we have that $|\sigma(S)\cap\sigma(P)|=1$. Since $b,c$ do not belong to $\sigma(S)$ and $\sigma(P)=\{a,b\}$, it follows that $\sigma(S)\cap\sigma(P)=\{a\}$. Then as $\sigma(S)$ has 2 elements and $b,c\not\in\sigma(S)$, there must be an element $d$ of $X$ distinct from $a,b,c$ with 
\begin{align*}
\sigma(S)&=\{a,d\}
\end{align*}
\vspace{-2ex}

Let $S'=\{y,z'\}$. There is a 2-relation $\theta'$ with $S'$ a block of $\theta'$, so $\sigma^{\theta'}(S')$ is defined, and we write this $\sigma(S')$. Since $S'$ and $Q$ are disjoint, by Proposition~\ref{oo} so are $\sigma(S')$ and $\sigma(Q)$. So neither $a,c$ belongs to $\sigma(S')$. Since $|S'\cap P|=1$, by Proposition~\ref{oo} we have $|\sigma(S')\cap\sigma(P)|=1$. Since neither $a,c$ belongs to $\sigma(S')$ and $\sigma(P)=\{a,b\}$, we must have that $\sigma(S')\cap\sigma(P)=\{b\}$. Also $|S\cap S'|=1$, so by Proposition~\ref{oo} we have $|\sigma(S)\cap\sigma(S')|=1$. Since $\sigma(S)=\{a,d\}$ and $a$ is not an element of $\sigma(S')$, we must have $\sigma(S)\cap\sigma(S')=\{d\}$. Then since the distinct elements $b,d$ belong to the 2-element set $\sigma(S')$ we have 
\vspace{-2ex}

\begin{align*}
\sigma(S')&=\{b,d\}
\end{align*}
\vspace{-2ex}

Since $X$ is infinite, there is an element $w$ of $X$ that is distinct from $x,y,z,z'$. Let $U=\{x,w\}$. Then there is a 2-relation $\theta''$ with $U$ a block of $\theta''$. So $\sigma^{\theta''}(U)$ is defined, and again, we write this $\sigma(U)$. Since $U$ is disjoint from $S$ and from $S'$, by Proposition~\ref{oo} we have that $\sigma(U)$ is disjoint from $\sigma(S)=\{a,d\}$ and from $\sigma(S')=\{b,d\}$. Therefore $a,b$ are not elements of $\sigma(U)$. But $|U\cap P|=1$, so by Proposition~\ref{oo} we have $|\sigma(U)\cap\sigma(P)|=1$. But $\sigma(P)=\{a,b\}$ and neither $a,b$ belongs to $\sigma(U)$. This contradiction shows that the assumption $\sigma(P)\cap\sigma(Q)\neq\sigma(P)\cap\sigma(Q')$ is faulty. So we have verified the claim in the case that $P\cap Q$ has a single element and $P=P'$. 

We next consider the general case when $P\cap Q$ has a single element. Again we assume  
\[P=\{x,y\}\quad \mbox{ and }\quad Q=\{x,z\}\]
\vspace{-2ex}

\noindent The case when $P=P'$ and $Q=Q'$ has been considered, and the case $P=Q'$ and $Q=P'$ is an instance of this. The case when $P=P'$ has been considered, and the cases where one of $P,Q$ is equal to one of $P',Q'$ are instances of this. So assume that $P,Q,P'$ and $Q'$ are all distinct and by assumption that $P\cap Q = \{x\}=P'\cap Q'$. Then there are elements $y'$ and $z'$ with $x,y,z,y',z'$ distinct and 
\[\,P'=\{x,y'\}\quad \mbox{ and }\quad Q'=\{x,z'\}\]
\vspace{-1ex}

Since $P\cap Q = P\cap Q'$ and $P\cap Q' = P'\cap Q'$, the case we have proved yields
\begin{align*}
\sigma(P)\,\cap\,\sigma(Q)\, &\,\,=\,\,\, \sigma(P)\,\cap\,\sigma(Q')\\
\sigma(P)\,\cap\,\sigma(Q') &\,\,=\,\, \sigma(P')\,\cap\,\sigma(Q')
\end{align*}
\vspace{-2ex}

\noindent Combining these gives the desired result $\sigma(P)\cap\,\sigma(Q) = \sigma(P)\cap\,\sigma(Q')$.
\end{proof}

\begin{prop}
There are mutually inverse permutations $\sigma$ and $\pi$ of $X$ so that for any distinct $x,y,z$ and any 2-relations $\pi$ and $\lambda$ with $P=\{x,y\}$ a block of $\pi$ and $Q=\{x,z\}$ a block of $\lambda$ 
\vspace{-4ex}

\begin{align*}
\sigma(x) &\mbox{ is the unique element of }\sigma^\pi(P)\,\cap\,\sigma^\lambda(Q)\\
\mu(x) &\mbox{ is the unique element of }\mu^\pi(P)\,\cap\,\mu^\lambda(Q)
\end{align*}
\vspace{-3ex}
\label{rr}
\end{prop}

\begin{proof}
Define maps $\sigma,\mu:X\to X$ as follows. Given $x\in X$ find $y,z$ distinct from $x$ and \mbox{2-relations} $\pi$ and $\lambda$ so that $P=\{x,y\}$ is a block of $\pi$ and $Q=\{x,z\}$ is a block of $\lambda$. Then set 
\begin{align*}
\sigma(x) &\mbox{ is the unique element of }\sigma^\pi(P)\,\cap\,\sigma^\lambda(Q)\\
\mu(x) &\mbox{ is the unique element of }\mu^\pi(P)\,\cap\,\mu^\lambda(Q)
\end{align*}
\vspace{-2ex}

\noindent Proposition~\ref{qq} shows that the results obtained are independent of the choices of $y,z$ and $\pi,\lambda$ so these are well-defined maps. 

Let $x,y,z$ and $\pi,\lambda$ be as described above. Say $\sigma^\pi(P)=\{u,v\}$ and $\sigma^\lambda(Q)=\{u,w\}$. So $\sigma(x)=u$. Note that $R=\{u,v\}$ is a block of $\beta(\pi)$ and $S=\{u,w\}$ is a block of $\beta(\lambda)$. So 

\[ \mu(u) \mbox{ is the unique element of }\mu^{\beta(\pi)}(R)\,\cap\,\mu^{\beta(\lambda)}(S)\]
\vspace{-1ex}

\noindent Proposition~\ref{w} gives that $\mu^{\beta(\pi)}(R)=P$ and $\mu^{\beta(\lambda)}(S)=Q$. Therefore $\mu(u)=x$. Thus $\mu\circ\sigma$ is the identity, and a similar argument shows that $\sigma\circ\mu$ is the identity. So $\sigma$ and $\mu$ are mutually inverse permutations of $X$. 
\end{proof}

\begin{prop}
For any 2-relation $\pi$ and any block $P$ of $\pi$
\vspace{-3ex}

\begin{align*}
\sigma(P) &\,=\, \{\sigma(x):x\in P\}\\
\mu(P) &\,=\, \{\mu(x):x\in P\}
\end{align*}
\vspace{-3ex}
\label{ss}
\end{prop}

\begin{proof}
Suppose $P=\{x,y\}$ is a block of $\pi$.  Choose $z\in X$ that is distinct from $x,y$ and a 2-relation $\lambda$ with $Q=\{x,z\}$ a block of $\lambda$. By definition, $\sigma(x)$ is the unique element in $\sigma^\pi(P)\,\cap\,\sigma^\lambda(Q)$. In particular, $\sigma(x)$ is an element of $\sigma^\pi(P)$. The same argument shows that $\sigma(y)$ is also an element of $\sigma^\pi(P)$. Since $\sigma$ is a permutation, $\sigma(x)$ and $\sigma(y)$ are distinct elements of the 2-element set $\sigma^\pi(P)$. Therefore $\sigma^\pi(P)=\{\sigma(x),\sigma(y)\}$. The result for $\mu$ is similar.
\end{proof}

We now come to what we want, a proof of Theorem~\ref{todo}.

\begin{thm}
Suppose that $\beta$ is a size preserving automorphism of $\Eq*(X)$. Then there is a permutation $\sigma$ of $X$ so that for every 2-relation $\pi$ we have 

\[\beta(\pi)=\{(\sigma(x),\sigma(y)):(x,y)\in\pi\}\] 
\vspace{-2ex}
\end{thm}

\begin{proof}
By Proposition~\ref{w}, $\sigma^\pi_2$ is a bijection between the blocks of $\pi$ and the blocks of $\beta(\pi)$, and Proposition~\ref{ss} shows that for a block $P=\{x,y\}$ of $\pi$ that $\sigma^\pi_2(P)=\{\sigma(x),\sigma(y)\}$. It follows that for $x\neq y$ in $X$ that $(x,y)\in\pi$ iff $P=\{x,y\}$ is a block of $\pi$ iff $\{\sigma(x),\sigma(y)\}$ is a block of $\beta(\pi)$ iff $(\sigma(x),\sigma(y))\in \beta(\pi)$. The corresponding result when $x=y$ is trivial. 
\end{proof}

We have now established all of the results used in section~4 to prove the main theorem, Theorem~\ref{main}, that for an infinite set $X$ there is a group isomorphism 

\[\Gamma:\Perm(X)\to\Aut(\Fact(X))\]
\vspace{-1ex}

\section{Concluding remarks}

In addition to proving a version of Wigner's theorem for the decompositions $\Fact(X)$ of an infinite set $X$, several results of independent interest have been obtained. Proposition~\ref{b'} shows that $\Fact(X)$ is atomistic in a strong way. Propositions~\ref{crud} and~\ref{move} describe properties related to the transitivity of the automorphism group. It is not transitive on atoms, but is transitive on $p$-atoms for a given prime $p$. The investigation here has also touched on a number of other questions. We have intentionally bounded the scope of the study to matters necessary to obtain the main theorem, but we feel the questions below seem worthy of separate study. 

\begin{question}
Give necessary and sufficient conditions in terms of the cardinalities of sets $X$ and $Y$ for there to exist an orthomodular poset isomorphism between $\Fact(X)$ and $\Fact(Y)$. 
\end{question}

It is clear that if $X$ and $Y$ are finite sets with a prime number of elements, then $\Fact(X)$ and $\Fact(Y)$ are isomorphic since each is a 2-element Boolean algebra. For a finite set $X$ of cardinality $n$, the number of factors counted with multiplicity in the prime factorization of $n$ is the number of atoms in each maximal Boolean subalgebra of $\Fact(X)$. We suspect that with the exception of sets with a prime number of elements, that $\Fact(X)$ is isomorphic to $\Fact(Y)$ iff $X$ and $Y$ have the same cardinality. Proposition~\ref{crud} contained results that were an expedient bypass to this question sufficient for the purposes here. 

\begin{question}
For which sets $X$ does a version of Wigner's theorem hold? More exactly, for which sets $X$ is there a group isomorphism $\Gamma:\Perm(X)\to\Aut(\Fact(X))$. 
\end{question} 

We have shown there is a positive answer for any infinite set $X$. What is known about the situation for finite sets is described in \cite{Tim}. Wigner's theorem does not hold for sets $X$ of finite cardinality $n$ where $n$ has only one or two prime factors. It also fails when $n=2^3$. The results of \cite{Tim} shows that it holds when $n=3^3$. All other cases are unknown. It is suspected that result is true whenever there are at least 3 odd prime factors. 

\begin{question} 
For atoms $a$ and $b$ of $\Fact(X)$ define $a\sim b$ if they have an upper bound of height three. Let the distance from $a$ and $b$ be the least $n$ where there is a sequence $a\sim c_1\sim \cdots\sim c_n\sim b$ and be $\infty$ if there is no such sequence. Let the diameter of $\Fact(X)$ be the maximum of the distances between atoms. What is the diameter of $\Fact(X)$?
\end{question}

Results in section~7 are related. The relations $\sim_2$ and $\sim_3$ are more restrictive than $\sim$. It is shown in section~7 that any two 2-atoms with the same first coordinate are in the transitive closure of $\sim_2$, and therefore have finite distance, and similarly for 3-atoms with the same first coordinates. However, more could be shown. A finite bound can be placed on the distance between two 2-atoms with the same first coordinate. This requires the observation that there is an $n$ so that any element of the permutation group $\Perm(X)$ can be obtained as the composite of at most $n$ involutions. 

\begin{question}
Is every automorphism of the underlying poset of $\Fact(X)$ an orthomodular poset isomorphism?
\end{question} 

In regards to this question, there is a discussion of similar issues in \cite{Tim}.

\end{document}